\documentclass[pra,twocolumn,superscriptaddress]{revtex4-2}
 
 \usepackage{latexsym}
 
\usepackage{cmap}   
\usepackage[utf8]{inputenc}
\usepackage{nicefrac}
\usepackage[english]{babel}

\usepackage{amsmath,amssymb}    
\usepackage{xcolor}
\usepackage{graphicx} 
\usepackage{multirow}
\usepackage{dcolumn} 
\usepackage{bm} 
\usepackage{braket} 
\usepackage{amsthm}
\usepackage[normalem]{ulem}
\usepackage{hyperref}
\usepackage{booktabs}
\usepackage{adjustbox}
 
\usepackage[ruled]{algorithm2e}

\newtheorem{theorem}{Theorem}
\newtheorem{corollary}{Corollary}
\newtheorem{lemma}{Lemma}
\newtheorem{proposition}{Proposition}
\newtheorem{definition}{Definition}

\newtheorem*{theorem-non}{Theorem}
\newtheorem*{lemma-non}{Lemma}
\newtheorem*{corollary-non}{Corollary}
\newtheorem*{proposition-non}{Proposition}

\usepackage{appendix}

\newcommand{\bx}{\bm{x}}
\newcommand{\bz}{\bm{z}}
\newcommand{\by}{\bm{y}}
\newcommand{\bxi}{\bm{x}^{(i)}}
\newcommand{\bzi}{\bm{z}^{(i)}}
\newcommand{\boi}{\bm{o}^{(i)}}
\newcommand{\bmw}{\bm{\mathsf{w}}}

\newcommand{\yi}{y^{(i)}}
\newcommand{\hatboi}{\hat{\bm{o}}^{(i)}_{\ba^*}}
\newcommand{\hatboit}{\hat{\bm{o}}^{(i)}_{\ba^*,t}}
\newcommand{\hatbo}{\hat{\bm{o}}}

\newcommand{\hatsigma}{\hat{\sigma}_n}
\newcommand{\hatsigmaTriGeoClc}{\hat{\sigma}_n^{(\mathsf{c})}}
\newcommand{\hatkappa}{\kappa_{\Lambda}}

\newcommand{\tilderho}{\tilde{\rho}}
 \newcommand{\Prob}{\mathbb{P}}

\newcommand{\btheta}{\bm{\theta}}
\newcommand{\bomega}{\bm{\omega}}
\newcommand{\ba}{\bm{a}}
\newcommand{\barba}{\bar{\bm{a}}}

\newcommand{\llangle}{\langle\langle}
\newcommand{\rrangle}{\rangle\rangle}

\newcommand{\hath}{\hat{h}}

\DeclareMathOperator{\RZ}{RZ}
\DeclareMathOperator{\RY}{RY}
\DeclareMathOperator{\RX}{RX}
\DeclareMathOperator{\RP}{RP}
 \DeclareMathOperator{\H5}{\mathsf{H}_{\mathsf{H}_5}}

\DeclareMathOperator{\CNOT}{CNOT}
 
 \DeclareMathOperator{\bUnitary}{\mathfrak{U}}
 \DeclareMathOperator{\bRZ}{\mathfrak{R}_Z}
 \DeclareMathOperator{\bRP}{\mathfrak{R}_P}
  \DeclareMathOperator{\Z}{Z}
  \DeclareMathOperator{\ROPT}{\mathsf{R}}

\DeclareMathOperator{\Tr}{Tr}

\DeclareMathOperator{\DKL}{D_{\text{KL}}}

\newcommand{\PP}{\mathbb{P}}
\newcommand{\hatrho}{\hat{\rho}}
\newcommand{\BQP}{\mathsf{BQP}}
\newcommand{\PPoly}{\mathsf{P/poly}}
\DeclareMathOperator{\CI}{\text{CI}}

\usepackage{chngcntr}
\usepackage{apptools}
\AtAppendix{\counterwithin{lemma}{section}}
\AtAppendix{\counterwithin{theorem}{section}}
\AtAppendix{\counterwithin{proposition}{section}}
\AtAppendix{\counterwithin{corollary}{section}}

\usepackage{pifont}
\newcommand{\xmark}{\ding{55}}
\newcommand{\cmark}{\ding{52}}
\DeclareMathSymbol{\qm}{\mathalpha}{operators}{"3F}
\DeclareMathAlphabet{\mathbbold}{U}{bbold}{m}{n}

\usepackage{ragged2e} 
\usepackage[justification=justified,singlelinecheck=false]{caption}
\DeclareCaptionJustification{justified}{\justifying}

\begin{document}

\title{Efficient Learning for Linear Properties of Bounded-Gate Quantum Circuits} 
 
\author{Yuxuan Du}
\email{duyuxuan123@gmail.com}
\affiliation{College of Computing and Data Science, Nanyang Technological University, Singapore 639798, Singapore}

\author{Min-Hsiu Hsieh}
\email{minhsiuh@gmail.com}
\affiliation{Hon Hai (Foxconn) Research Institute, Taipei, Taiwan}

\author{Dacheng Tao}
\email{dacheng.tao@gmail.com}
\affiliation{College of Computing and Data Science, Nanyang Technological University, Singapore 639798, Singapore}

\maketitle

\noindent\textbf{ABSTRACT\\
The vast and complicated large-qubit state space forbids us to comprehensively capture the dynamics of modern quantum computers via classical simulations or quantum tomography. Recent progress in quantum learning theory prompts a crucial question: can  linear properties of a large-qubit circuit with $d$ tunable RZ gates and $G-d$ Clifford gates be efficiently learned from measurement data generated by varying classical inputs? In this work, we prove that the sample complexity scaling linearly in $d$ is required to achieve a small prediction error, while the corresponding computational complexity may scale exponentially in $d$. To address this challenge, we propose a kernel-based method leveraging classical shadows and truncated trigonometric expansions, enabling a controllable trade-off between prediction accuracy and computational overhead. Our results advance two crucial realms in quantum computation: the exploration of quantum algorithms with practical utilities and learning-based quantum system certification. We conduct numerical simulations to validate our proposals across diverse scenarios, encompassing quantum information processing protocols, Hamiltonian simulation, and variational quantum algorithms up to 60 qubits.}

\medskip
\noindent\textbf{INTRODUCTION} \\
\noindent Advancing efficient methodologies to characterize the behavior of quantum computers is an endless pursuit in quantum science, with pivotal outcomes contributing to designing improved quantum devices and identifying computational merits. In this context, quantum tomography and classical simulators have been two standard approaches. Quantum tomography, spanning state \cite{leonhardt1995quantum}, process  \cite{altepeter2003ancilla,mohseni2008quantum}, and shadow tomography \cite{aaronson2007learnability,aaronson2018shadow}, provides concrete ways to benchmark current quantum computers. Classical simulators, transitioning from state-vector simulation to tensor network methods  \cite{markov2008simulating,villalonga2019flexible,cirac2021matrix,pan2022simulation} and Clifford-based simulators \cite{bravyi2016improved,bravyi2019simulation,beguvsic2023fast}, not only facilitate the design of quantum algorithms without direct access to quantum resources \cite{chen2023quantum} but also push the boundaries to unlock practical utilities \cite{arute2019quantum}. Despite their advancements, (shadow) tomography-based methods are quantum resources intensive,  necessitating extensively interactive access to quantum computers, and classical simulators are confined to handling specific classes of quantum states. Accordingly, there is a pressing need for innovative approaches to effectively uncover the dynamics of modern quantum computers with hundreds or thousands of qubits \cite{kim2023evidence,bluvstein2023logical}.

Machine learning (ML) is fueling a new paradigm for comprehending and reinforcing quantum computing \cite{gebhart2023learning}. This hybrid approach, distinct from prior purely classical or quantum methods, synergies the power of classical learners and quantum computers. Concisely, the process starts by collecting samples from target quantum devices to train a classical learner, which is then used to predict unseen data from the same distribution, either with minimal quantum resources or purely classically. Empirical studies have showcased the superiority of ML compared to traditional methods in many substantial tasks, such as real-time feedback control of quantum systems  \cite{torlai2017neural,zeng2023approximate, dehghani2023neural,reuer2023realizing}, correlations and entanglement quantification \cite{canabarro2019machine,chen2022certifying,koutny2023deep}, and enhancement of variational quantum algorithms \cite{bennewitz2022neural,zhang2022variational,wei2023neural}.  

Despite their empirical successes, the theoretical foundation of these ML-based algorithms holds far-reaching consequences, where rigorous performance guarantees or scaling behaviors remain unknown. A critical, yet unexplored, question in the field is:\begin{center}
 Is there a provably efficient ML model that can predict dynamics of large-qubit bounded-gate quantum circuits?
\end{center} 
The focus on the bounded-gate quantum circuits stems from the facts that many interested states in nature often exhibit bounded gate complexity and this has broad applications in variational quantum algorithms \cite{cerezo2021variational} and system certification \cite{eisert2020quantum}. Given the offline nature of the inference stage, any advancements in this area could significantly reduce the quantum resource overhead, which is especially desirable given their scarcity.  

\begin{figure*}[t]
	\centering
	\includegraphics[width=0.99\textwidth]{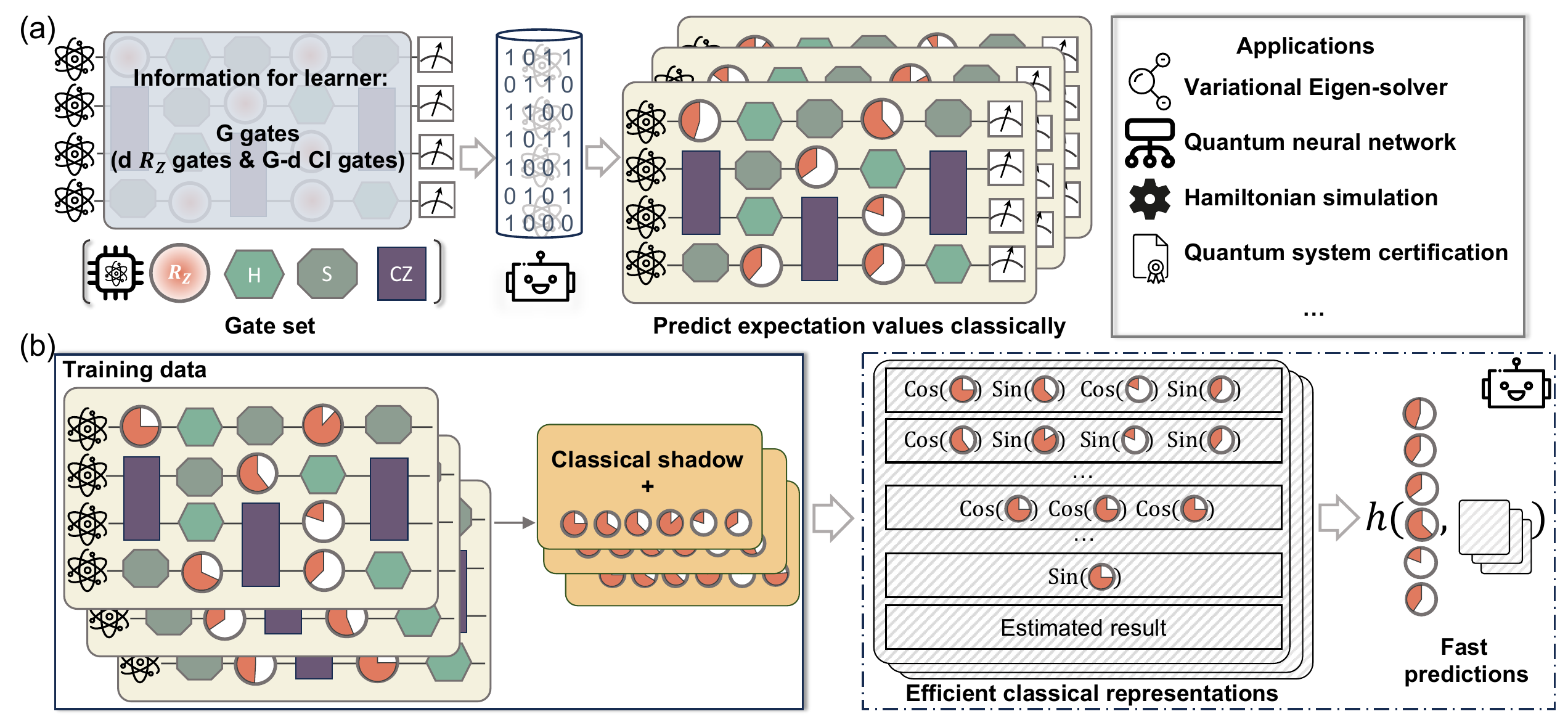}
	\caption{\small{\textbf{Learning protocols for quantum circuits with the bounded number of non-Clifford gates}. (a) \textsc{The visualization of learning bounded-gate quantum circuits with incoherent measurements.} Given a circuit composed of $G-d$ Clifford gates and $d$ $\RZ$ gates, a classical learner feeds $n$ classical inputs, i.e., $n$ tuples of the varied angles of $\RZ$ gates, to the quantum device and collects the relevant measured results as data labels. The collected $n$ labeled data are used to train a prediction model $h$ that can perform purely classical inference to predict the linear properties of the generated state over new input $\bx$, i.e., $\Tr(\rho(\bx)O)$ with $O$ being an observable sampled from a prior distribution, can be accurately estimated.  Following conventions, the interaction between the learner and the system is restrictive in which the learner can only access the quantum computer via incoherent measurements with finite shots.  (b) \textsc{The proposed kernel-based ML model consists of two steps.} First, the learner collects the training dataset, i.e., $n$ labeled data via classical shadow. Then, the learner applies shadow estimation and the trigonometric monomial expansion to the collected training dataset to obtain efficient classical representations, where any new input of the explored quantum circuits can be efficiently predicted offline. }}
	\label{fig:scheme}
\end{figure*}

Here we answer the above question affirmably by exploring a specific class of quantum circuits that has broad applications in quantum computing. Particularly, this class of $N$-qubit circuits consists of an arbitrary initial state, a bounded-gate circuit composed of $d$ tunable RZ gates and $G-d$ Clifford gates, where the dynamics refer to linear properties of the resulting state measured by the operator $O$. Our first main result uncovers that (i) with high probability, $\widetilde{\Omega}(\frac{(1-\epsilon)d}{\epsilon T})  \leq n \leq   	\widetilde{\mathcal{O}}(\frac{B^2d + B^2NG}{\epsilon^2})$ training examples are sufficient and necessary for a classical learner to achieve an $\epsilon$-prediction error on average, where $B$ refers to the bounded norm of $O$ and $T$ is the number of incoherent measurements; (ii) there exists a class of $G$-bounded-gate quantum circuits that no polynomial runtime algorithms can predict their outputs within an $\epsilon$-prediction error. These results deepen our understanding of the learnability of bounded-gate quantum circuits with classical learners, whereas Ref.~\cite{zhao2023learning} focuses on quantum learners. Furthermore,  the achieved findings invoke the necessity of developing a learning algorithm capable of addressing the exponential gap between sample efficiency and computational complexity. 

To address this challenge, we further utilize the concepts of classical shadow \cite{huang2020predicting} and trigonometric expansion to design a kernel-based ML model that effectively balances prediction accuracy with computational demands. We prove that when the classical inputs are sampled from the uniform distribution, with high probability, the runtime and memory cost of the proposed ML model is no larger than $\widetilde{\mathcal{O}}(TNB^2 c/\epsilon)$  for a large constant $c$ to achieve an $\epsilon$ prediction error in many practical scenarios. Our proposed algorithm paves a new way of predicting the dynamics of large-qubit quantum circuits in a provably efficient and resource-saving manner.

\medskip
\noindent\textbf{RESULTS}\\
\noindent \textbf{Problem setup.}--- We define the state space associated to an $N$-qubit bounded-gate quantum circuit as 
\begin{equation}\label{eqn:state-space}
   \mathcal{Q} = \left\{\rho(\bx) = U(\bx)\rho_0 U(\bx)^{\dagger} \Big|\bx \in [-\pi, \pi]^d \right\}, 
\end{equation}   
where $\rho_0$ denotes an $N$-qubit state and $U(\bx)$ refers to the bounded-gate quantum circuit depending on the classical input $\bx$ with $d$ dimensions. Due to the universality of Clifford (abbreviated as $\CI$) gates with $\RZ$ gates \cite{ross2014optimal}, an $N$-qubit $U(\bx)$ can always be expressed as 
\begin{equation}
	U(\bx) = \prod_{l=1}^{d}(\RZ(\bx_l)u_e),
\end{equation}
where $d$ is the number of $\RZ$ gates, and $u_e$ refers to the unitary composed of $\CI$ gates and the identity gate $\mathbb{I}_2$ with $\CI=\{H, S, \CNOT\}$. We denote the total number of gates and CI gates in $U(\bx)$ as $G$ and $G-d$, respectively.  When $\rho(\bx)$ is  measured by a known observable $O$, its incoherent dynamics are described the mean-value space, i.e., $f(\bx, O)= \Tr(\rho(\bx)O)$ with $\bx \in [-\pi, \pi]^d$. This formalism encompasses diverse tasks in quantum computing, e.g., variational quantum algorithms \cite{cerezo2021variational}, numerous applications of shadow estimation \cite{elben2022randomized}, and quantum system certification \cite{eisert2020quantum} (see Supplementary Material (SM)~A for the elaboration).

The purpose of most classical learners is to obtain a learning model $h(\bx,O)$ to predict $f(\bx, O)$. The general paradigm is shown in Fig.~\ref{fig:scheme}(a). To accommodate the constraints of modern quantum devices, the training data are exclusively gathered by incoherent measurements. As such, the classical learner feeds $\bx$ into the quantum circuit and receives $\tilde{f}_T(\bx, O)$, as the estimation of $f(\bx, O)$ incurred by $T$ finite shots. By repeating this procedure with varied inputs, the classical learner constructs a training dataset $\mathcal{T}=\{\bxi, \tilde{f}_T(\bx, O)\}_{i=1}^n$. Besides, the classical learner is kept unaware of the circuit layout details, except for the gate count $G$ and the dimension of classical inputs $d$. Last, the inferred model $h(\bx,O)$ should be offline, where the predictions for new inputs $\bx$ is performed solely on the classical side. Given the scarcity of quantum computers, the offline property of $h(\bx,O)$ significantly reduces the inference cost for studying incoherent dynamics compared to quantum learners, such as variational quantum algorithms \cite{jerbi2023power,gibbs2024dynamical}. 

According to the above paradigm, the concept class of the mean-value space of bounded-gate circuit is 
\begin{equation}\label{eqn:function-space}
   \mathcal{F} = \left\{f(\bx, O)= \Tr(\rho(\bx)O) \Big| U\in \text{Arc}(U, d, G) \right\},
\end{equation}
where $\bx \in [-\pi, \pi]^d$ and $\text{Arc}(U, d, G)$ refers to the set of all possible gate arrangements consisting of $d$ $\RZ$ gates and $G-d$ $\CI$ gates. The hardness of learning $\mathcal{F}$ by classical learners is explored by separately assessing the required sample complexity and runtime complexity of training a classical model $h(\bx,O)$ that attains a low average prediction error with $f(\bx,O)$. Mathematically,  the average performance of the learner is expected to satisfy \begin{equation}\label{eqn:risk_task}
 \mathbb{E}_{\bx\sim \mathbb{D}_\mathcal{X}}\left|h(\bx,O) - f(\bx, O) \right|^2 \leq \epsilon, 
\end{equation} 
where the classical input $\bx$ is sampled from the distribution $\mathbb{D}_\mathcal{X}$.  Here we suppose $O$ consists of $q$ local observables with a $B$-bounded norm, i.e., $O=\sum_{i=1}^q O_i$ and $\sum_l \|O_i\|_{\infty}\leq B$, and the maximum locality of $\{O_i\}$ is $K$.

\smallskip 
\noindent {Learnability of bounded-gate quantum circuits}.--- The following theorem demonstrates the learnability of $\mathcal{F}$ in Eq.~(\ref{eqn:function-space}), where the formal statement and the proof are presented in SM~B-E (see Methods for the proof sketch). 
\begin{theorem}[informal]\label{thm:learn-complexity}
 Following notations in Eq.~(\ref{eqn:function-space}), let $\mathcal{T}=\{\bxi, \tilde{f}_T(\bxi)\}_{i=1}^n$ be a dataset containing $n$ training examples with $\bxi\in [-\pi, \pi]^d$ and $\tilde{f}_T(\bxi)$ being the estimation of $f(\bxi, O)$ using $T$ incoherent measurements with $T\leq (N\log2)/\epsilon$. 
Then, the training data size 
\begin{equation}\label{eqn:sample-bound-thm1-main}
\widetilde{\Omega}\left(\frac{(1-\epsilon) d}{\epsilon T}\right)  \leq n \leq   	\widetilde{\mathcal{O}}\left(\frac{B^2d+B^2NG}{\epsilon}\right)  
\end{equation}
is sufficient and necessary to achieve Eq.~(\ref{eqn:risk_task}) with high probability. However, unless $\BQP\subseteq \mathsf{P/poly}$, there exists a class of $G$-bounded-gate quantum circuits that no algorithm can achieve Eq.~(\ref{eqn:risk_task}) in a polynomial time.	
\end{theorem}
The exponential separation in terms of the sample and computational complexities underscores the non-trivial nature of learning the incoherent dynamics of bounded-gate quantum circuits. While the matched upper and lower bounds in Eq.~(\ref{eqn:sample-bound-thm1-main}) indicate that a linear number of training examples with $d$ is sufficient and necessary to guarantee a satisfactory prediction accuracy, the derived exponential runtime cost hints that identifying these training examples may be computationally hard. Note that the upper bound does not depend non-trivially on $T$, so we omit it. Besides, an interesting future research direction is to explore novel techniques to match the factors $N$, $B$, and $G$ in the lower and upper bounds, whereas such deviations do not affect our key results.

These results enrich quantum learning theory \cite{anshu2023survey,banchi2023statistical}, especially for the learnability of bounded-gate quantum circuits. Ref.~\cite{zhao2023learning} exhibits that for a quantum learner, reconstructing quantum states and unitaries with bounded-gate complexity is sample-efficient but computationally demanding. While Theorem~\ref{thm:learn-complexity} shows a similar trend, the varied learning settings (quantum learner vs. classical learner) and different tasks suggest that the implications and proof techniques in these two studies are inherently different. As detailed in SM~J, these two works are complementary, which collectively reveal the capabilities and limitations of learning bounded-gate quantum circuits. 
     
\smallskip
\noindent {A provably efficient protocol to learn $\mathcal{F}$.}--- The exponential separation of the sample and computational complexity pinpoints the importance of crafting provably efficient algorithms to learn $\mathcal{F}$ in Eq.~(\ref{eqn:function-space}). To address this issue, here we devise a kernel-based ML protocol adept at balancing the average prediction error $\epsilon$ and the computational complexity, making a transition from exponential to polynomial scaling with $d$ when $\mathbb{D}_{\mathcal{X}}$ is restricted to be the uniform distribution, i.e., $\bx\sim [-\pi, \pi]^d$.

\begin{figure*} 
\centering
\includegraphics[width=0.995\textwidth]{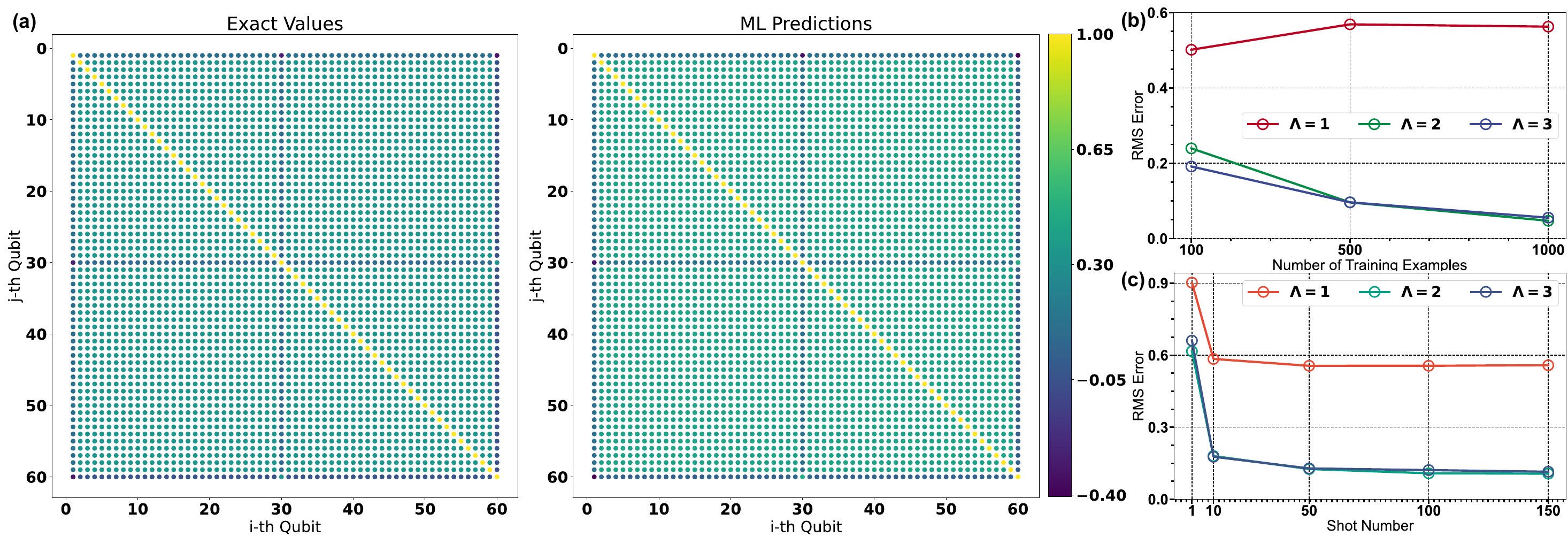}
\caption{\small{\textbf{Numerical results of predicting properties of rotational 60-qubit GHZ states}. (a) \textsc{Two-point correlation}.  Exact values and ML predictions of the expectation value of the correlation function $C_{ij} = (X_iX_j +Y_iY_j +Z_iZ_j)/3$ for all qubit pairs of 60-qubit GHZ states. The node’s color indicates the exact value and predicted value of the two subplots, respectively. (b-c) \textsc{Prediction error}. Subplot (b) depicts the root mean squared (RMS) prediction error of the trained ML model with varied truncation $\Lambda$ and number of training examples $n$. Subplot (c) shows the RMS prediction error of the trained ML model with varied truncation $\Lambda$ and the shot number $T$. The setting $\Lambda=3$ refers to the full expansion. }}
\label{fig:Rot-GHZ-main}
\end{figure*}

Our proposal, as depicted in Fig.~\ref{fig:scheme}(b), contains two steps: (i) Collect training data from the exploited quantum device; (ii) Construct the learning model and use it to predict new inputs. In Step~(i), the learner feeds different $\bxi \in [-\pi, \pi]^d$ to the circuit and collects classical information of $\rho(\bxi)$ under Pauli-based classical shadow with $T$ snapshots, denoted by $\tilderho_T(\bxi)$. In this way, the learner constructs the training dataset $\mathcal{T}_{\mathsf{s}}=\{\bxi \rightarrow \tilderho_T(\bxi)\}_{i=1}^n$ with $n$ training examples. Then, in Step~(ii), the learner utilizes $\mathcal{T}_{\mathsf{s}}$ to build a kernel-based ML model $h_{\mathsf{s}}$, i.e., given a new input $\bx$, its prediction yields  
\begin{equation}\label{eqn:generic-learner}
	h_{\mathsf{s}}(\bx, O) = \frac{1}{n}\sum_{i=1}^n\kappa_{\Lambda}\left(\bx, \bxi \right)g\left(\bxi,O \right),
\end{equation}  
where $g(\bxi,O)=\Tr(\tilderho_T(\bxi)O)$ refers to the shadow estimation of $\Tr(\rho(\bxi) O)$, $\kappa_{\Lambda}(\bx, \bxi)$ is the truncated trigonometric monomial kernel with 
\begin{equation}\label{eqn:trigono-kernel} 
	\kappa_{\Lambda}\left(\bx, \bxi \right) = \sum_{\bomega, \|\bomega\|_0 \leq \Lambda} 2^{\|\bomega\|_0}\Phi_{\bomega}(\bx)\Phi_{\bomega}\left(\bxi \right) \in \mathbb{R},
\end{equation}
and $\Phi_{\bomega}(\bx)$ with $\bomega\in \{0, 1, -1\}^d$ is the trigonometric monomial basis with   values 
\begin{equation} \label{eqn:trigono-basis} 
	\Phi_{\bomega}(\bx) = \prod_{i=1}^d \begin{cases}
		 1 ~ & \textnormal{if}~ \bomega_i = 0 \\
		 \cos(\bx_i) & \textnormal{if}~\bomega_i = 1 \\
		 \sin(\bx_i) &  \textnormal{if}~ \bomega_i = -1
	\end{cases}.
\end{equation}
A distinctive aspect of our proposal is its capability to predict the incoherent dynamics $\Tr(\rho(\bx)O)$ across various $O$ purely on the classical side. This is achieved by storing the shadow information $\{\tilderho_T(\bxi)\}$ into the classical memory, where the shadow estimation $\{g(\bxi, O) \}$ for different $\{O\}$ can be efficiently processed through classical post-processing (refer to Method for details).

With the full expansion as $\Lambda=d$, the cardinality of the frequency set $\{\bomega\}$ in Eq.~(\ref{eqn:trigono-kernel}) is $3^d$, impeding the computational efficiency of our proposal when the number of $\RZ$ gates becomes large. To remedy this, here we adopt a judicious frequency truncation to strike a balance between prediction error and computational complexity. Define the truncated frequency set as $\mathfrak{C}(\Lambda) =\{\bomega|\bomega \in \{0, \pm 1\}^d, ~s.t.~\|\bomega\|_0\leq \Lambda\}$. The subsequent theorem provides a provable guarantee for this tradeoff relation, with the formal description and proof deferred to SM~F.
\begin{theorem}[Informal]\label{thm:learning-non-trunc}
 Following notations in Eqs.~(\ref{eqn:function-space})-(\ref{eqn:generic-learner}), suppose $\mathbb{E}_{\bx\sim [-\pi, \pi]^d}\|\nabla_{\bx} \Tr(\rho(\bx)O)\|_2^2\leq C$. When the frequency is truncated to $  \Lambda=4C/\epsilon$ and the number of training examples is $n = |\mathfrak{C}(\Lambda)|  {2  B^2 9^K}{\epsilon^{-1}}  \log ({2 |\mathfrak{C}(\Lambda)|}/{\delta})$, even for $T=1$ snapshot per training example, with probability at least $1-\delta$, 
\begin{equation}
	\mathbb{E}_{\bx\sim [-\pi, \pi]^d} \left| h_{\mathsf{s}}(\bx, O) - f(\rho(\bx),O)  \right|^2 \leq \epsilon.
\end{equation}	
\end{theorem}

\noindent The obtained results reveal that both sample and computational complexities of $h_{\mathsf{s}}$ are predominantly influenced by the cardinality of $\mathfrak{C}(\Lambda)$, or equivalently the quantity $C$ as $\Lambda=4C/\epsilon$. That is, a polynomial scaling of $|\mathfrak{C}(\Lambda)|$ with $N$ and $d$ can ensure both a polynomial runtime cost to obtain $\kappa_{\Lambda}(\bx, \bxi )g(\bxi,O)$ and a polynomial sample complexity $n$, leading to an overall polynomial computational complexity of our proposal (see Methods and SM~G for details). In contrast, for an unbounded $C$ such that $|\mathfrak{C}(\Lambda)|$ exponentially scales with the number of $\RZ$ gates $d$, the computational overhead of $h_{\mathsf{s}}$ becomes prohibitive for a large $d$, aligning with the findings from Theorem~\ref{thm:learn-complexity}.

\begin{figure*}
	\centering
\includegraphics[width=0.91\textwidth]{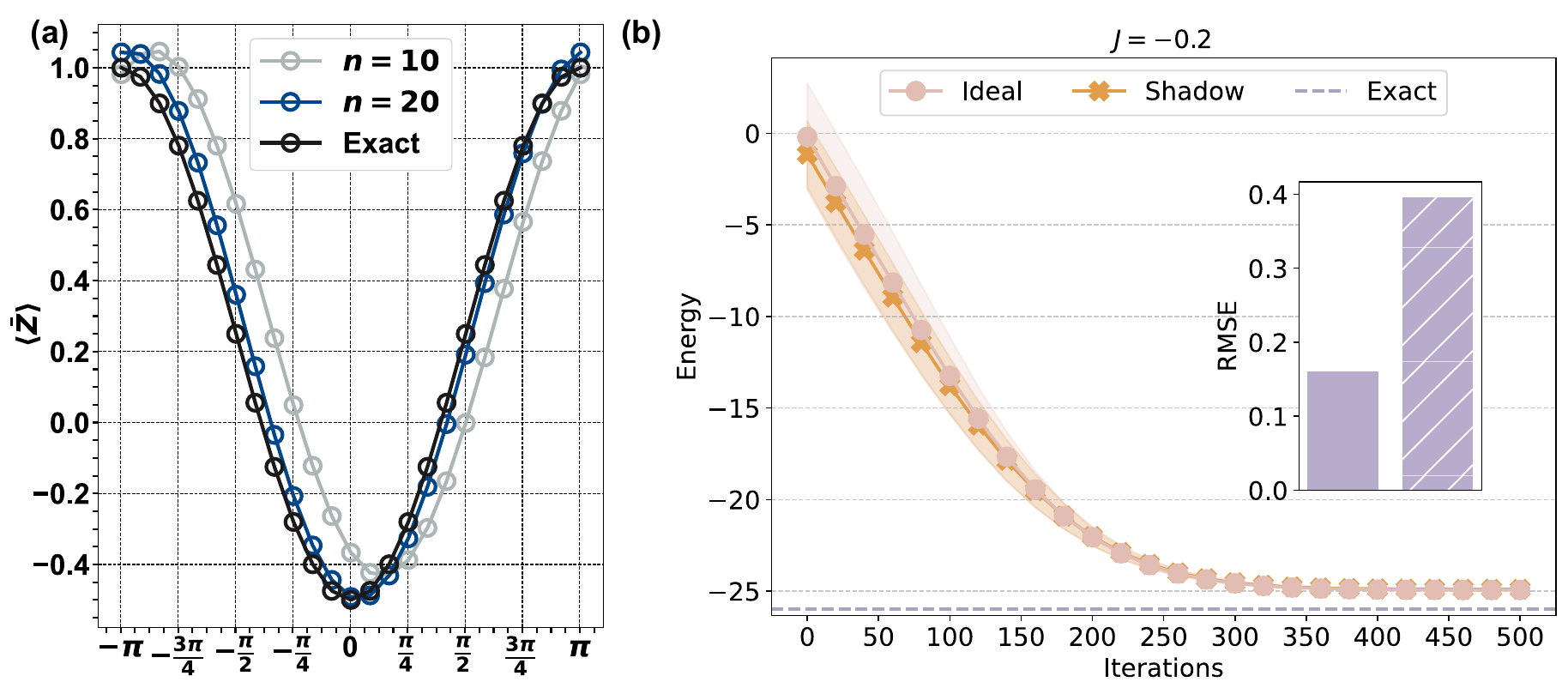}
	\caption{\small{\textbf{Numerical results of predicting properties of quantum states evolved by  60-qubit global Hamiltonians}. (a) \textsc{Prediction on the evolved state with $d=1$}. The notation $n=a$ refers to the number of training examples used to form the classical representation as $a$. The y-axis denotes the magnetization $\braket{\bar{Z}}=\frac{1}{60}\sum_i \braket{Z_i}$. (b) \textsc{Pretraining performance of the proposed ML model on a 50-qubit TFIM}. The inset presents the prediction error  of the proposed ML model for both the ideal (solid bar) and shadowed (Hatching patterns with `/') cases, i.e., $T\rightarrow \infty$ and $T=300$. The main plot depicts the training dynamics of the proposed ML model over $500$ iterations under ideal (circle markers) and shadowed (star markers) settings. The shaded areas represent the variance, while the dashed line indicates the exact ground-state energy.  }}
	\label{fig:main-Ham-sim}
\end{figure*}
 
We next underscore that in many practical scenarios, the quantity $C$  can be effectively bounded, allowing the proposed ML model to serve as a valuable complement to quantum tomography and classical simulations in comprehending quantum devices (See SM~K for heuristic evidence). One scenario involves characterizing near-Clifford quantum circuits consisting of many CI gates and few non-Clifford gates, which find applications in quantum error correction \cite{calderbank1997quantum,terhal2015quantum} and efficient implementation of approximate t-designs \cite{haferkamp2020quantum}. In this context, adopting the full expansion with $\Lambda=d$ is  computationally efficient, as  $|\mathfrak{C}(d)|\sim O(N)$. Meanwhile, when focused on a specific class of quantum circuits described as $\CI +\RZ$ with a fixed layout, our model surpasses classical simulation methods \cite{aaronson2004improved,lai2022learning} in runtime complexity by eliminating the dependence on the number of Clifford gates $G-d$. 

 Another scenario involves the advancement of variational quantum algorithms (VQAs) \cite{cerezo2021variational,tian2022recent,bharti2021noisy}, a leading candidate of leveraging near-term quantum computers for practical utility in machine learning \cite{li2022recent}, quantum chemistry, and combinatorial optimization \cite{farhi2016quantum}. Recent studies have shown that in numerous VQAs, the gradients norm of trainable parameters yields $C\leq 1/\text{poly}(N)$ \cite{pesah2021absence,zhang2021toward,wang2022symmetric,larocca2023theory} or $C\leq 1/\text{exp}(N)$, a.k.a, barren plateaus \cite{mcclean2018barren,cerezo2020cost,marrero2021entanglement,arrasmith2022equivalence}. These insights, coupled with the results in Theorem 2, suggest that our model can be used to pretrain VQA algorithms on the classical side to obtain effective initialization parameters before quantum computer execution, preserving access to quantum resources \cite{dborin2022matrix,rudolph2023synergistic}. Theoretically, our model broadens the class of VQAs amenable to classical simulation, pushing the boundaries of achieving practical utility with VQAs \cite{schreiber2023classical,cerezo2023does,gan2024concept}.

Last, the complexity bound in Theorem~\ref{thm:learning-non-trunc} hints that the locality of the observable $K$ is another factor dominating the performance of $h_{\mathsf{s}}$. This exponential dependence arises from the use of Pauli-based classical shadow, and two potential solutions can be employed to alleviate this influence. The first solution involves adopting advanced variants of classical shadow to enhance the sample complexity bounds \cite{huang2021efficient,nguyen2022optimizing,zhang2023minimal,vermersch2023enhanced}. The second solution entails utilizing classical simulators to directly compute the quantity $\{\Tr(\rho(\bxi)O)\}$ instead of shadow estimation $\{g(\bxi,O)\}$ in Eq.~(\ref{eqn:generic-learner}), with the sample complexity summarized in the following corollary.
 \begin{corollary}[Informal]\label{append:Coro:predict-error-clc-backend-full-expand}
 	Following notations in Theorem~\ref{thm:learning-non-trunc}, when $\{\Tr(\rho(\bxi)O)\}_i$  are computed by classical simulators and  $n \sim  \widetilde{\mathcal{O}}(3^d B^2 d /\epsilon)$, with high probability, the average prediction error is upper bounded by $\epsilon$. 
 \end{corollary}
 \noindent Although using the classical simulators can improve the dependency of the locality of observable and remove the necessity of quantum resources, the price to pay is increasing the computational overhead and only restricting to a small constant $d$.  Refer to SM~H for the proofs, implementation details, and more discussions. 
 
\noindent\textbf{Remark}. For clarity, we focus on showing how the proposed ML model applied to the bounded-gate circuits with $\CI$ and $\RZ$ gates. In SM~I, we illustrate that our proposal and the relevant theoretical findings can be effectively extended to a broader context, i.e., the circuit composed of $\CI$ gates alongside parameterized multi-qubit gates generated by arbitrary Pauli strings. In SM~J, we show how our proposal differs from Pauli path simulators.

\medskip 
\noindent \textbf{Numerical results}.--- We conduct numerical simulations on $60$-qubit quantum circuits to assess the efficacy of the proposed ML model. The omitted details, as well as the demonstration of classically optimizing VQAs with smaller qubit sizes, are provided in  Method and SM~J.

We first use the proposed ML model to predict the properties of rotational GHZ states. Mathematically,  we define an $N$-qubit rotational GHZ states with $N=60$ as $\ket{\text{GHZ}(\bx)}=U(\bx)({\ket{0}^{\otimes N} + \ket{1}^{\otimes N}})/{\sqrt{2}}$, where $U(\bx) =  \RY_1(\bx_1) \otimes \RY_{30}(\bx_2) \otimes  \RY_{60}(\bx_3)$ and the subscript $j$ refers to apply $\RY$ gate to the $j$-th qubit. At the training stage, we constitute the dataset $\mathcal{T}_s$ containing  $n=30$ examples, where each example is obtained by uniformly sampling $\bx$ from $[-\pi, \pi]^3$ and applying classical shadow to $\ket{\text{GHZ}(\bx)}$ with the shot number $T=1000$. 

The first subtask is predicting a two-point correlation function, i.e., the expectation value of $C_{ij} = (X_iX_j +Y_iY_j +Z_iZ_j)/3$ for each pair of qubits ($i,j$), at new values of $\bx$. To do so, the proposed ML model leverages $\mathcal{T}_s$ to form the classical representations with $\Lambda=3$ and exploits these representations to proceed with prediction at $\bx$. Fig.~\ref{fig:Rot-GHZ-main}(a) depicts the predicted and actual values of the correlation function for a particular value of $\bx$, showing reasonable agreement. The second subtask is exploiting how the prediction error depends on the truncation value $\Lambda$, the number of training examples $n$, and the shot number $T$. Fig.~\ref{fig:Rot-GHZ-main}(b) showcases that when $T=1000$ and the training data are sufficient (i.e., $n\geq 500$), the root mean squared (RMS) prediction error on $10$ unseen test examples is almost the same for the full expansion (i.e., $\Lambda =3$) and the proper truncation (i.e., $\Lambda =2$). Besides, Fig.~\ref{fig:Rot-GHZ-main}(c) indicates that when $n=500$, the prediction error on the same test examples reaches a low value for both $\Lambda =3$ and $\Lambda =2$ once the shot number $T$ exceeds a threshold value ($T\geq 50$). These results are consistent with our theoretical analysis.

We next apply the proposed ML model to predict properties of the state evolved by a global Hamiltonian $\mathsf{H}=J_1\otimes_{i=1}^NZ_i + \sum_{i=1}^N X_i$, where $J_1$ is a predefined constant and the qubit count is set to $N=60$. The initial state is $\ket{0}^{\otimes N}$ and  $U(\bx)$ is the Trotterized time evolution of $\mathsf{H}$. By properly selecting the evolution time at each Trotter step and $J_1$, the Trotterized time-evolution circuit takes the form as $U(\bx)=\prod_{j=1}^d (e^{-\imath \bx_j \otimes_{i=1}^N Z_i }\otimes_{i=1}^N\RX(\pi/3))$. We set the total number of Trotter steps as $d=1$. At the training stage, we constitute the dataset $\mathcal{T}_{\mathsf{s}}$ following the same rule presented in the last simulation. The only difference is the dataset size and the shot number, which is $n = 20$ and $T=500$. The task is predicting the magnetization with $\braket{\bar{Z}}=\frac{1}{60}\sum_i \braket{Z_i}$. 

The comparison between the exact value and the prediction with $\Lambda=1$ (full expansion) is shown in Fig.~\ref{fig:main-Ham-sim}(a). We select $25$ test examples evenly distributed across the interval $[-\pi, \pi]$. By increasing the number of training examples $n$ from $10$ to $20$, the prediction values of the proposed ML model almost match the exact results.

We last evaluate the performance of the proposed ML model as a classical surrogate for variational quantum eigen-solver (VQE) in estimating the ground state energy of the Transverse Field Ising Model (TFIM). The explored $50$-qubit TFIM takes the form as $\mathsf{H}= \sum_{i=1}^{N-1} -0.2 Z_iZ_{i+1} -0.5 \sum_{i=1}^N X_i$ with $N=50$. The initialized state is $\ket{0}^{\otimes 50}$ and the employed ansatz is  $U(\bx)=\prod_{i=1}^{N-1}\text{RZZ}(\bxi;i)(\otimes_{i=1}^N \RX(\bx_{i+N-1})) \text{Ha}^{\otimes N}$, where $\text{Ha}$ refers to the Hadamard gate and $\text{RZZ}(\bxi;i)$ denotes applying RZZ gate to the $i$-th and $i+1$-th qubits. 

The pretraining process involves three steps. In the first two steps, we construct the training dataset $\mathcal{T}_{\mathsf{s}}$ by randomly sampling $\bx\in [-\pi, \pi]^{99}$ from a uniform distribution and inputting these values into the ansatz $U(\bx)$ to generate classical shadows. The number of trainable examples is set as $n=1500$ and the snapshot for each quantum state $\rho(\bx)$ has two settings, which are $T=300$ (shadow case) and $T\rightarrow \infty$ (ideal case). Once $\mathcal{T}_{\mathsf{s}}$ is collected, we move to the third step, where the proposed ML model $h_{\mathsf{s}}(\bx, \mathsf{H})$ built on $\mathcal{T}_{\mathsf{s}}$ to pretrain VQEs. The pretraining amounts to an optimization task with $\min_{\bx} h_{\mathsf{s}}(\bx, \mathsf{H})$. The Adam optimizer is used to minimize this surrogate loss, with a learning rate of $0.01$ and a total of $500$ iterations. Refer to SM~K for more details. 

 The achieved results are illustrated in Fig.~\ref{fig:main-Ham-sim}(b). The inset illustrates the prediction error between the exact result and the optimized $h_{\mathsf{s}}(\bx, \mathsf{H})$, evaluated using $500$ test examples sampled from the same data distribution as the training examples. The prediction error for both the shadow  and the ideal datasets remains remarkably low, demonstrating the success of the learning process.  The main plot shows the dynamics of the proposed ML model during the minimization of the surrogate loss. For both the noisy and ideal cases, the estimated energies after pretraining are close to the exact value.

\medskip
\noindent\textbf{DISCUSSION}\\
In this study, we prove that learning bounded-gate quantum circuits with incoherent measurements is sample efficient but computationally hard. Furthermore, we devise a provably efficient ML algorithm to predict the incoherent dynamics of bounded-gate quantum circuits, transitioning from exponential to polynomial scaling. The achieved results provide both theoretical insights and practical applications, demonstrating the efficacy of ML in comprehending and advancing quantum computing. 

Several crucial research avenues merit further exploration. First, our study addresses the scenario of noiseless quantum operations. An important and promising area for future investigation is the development of provably efficient learning algorithms capable of predicting the incoherent dynamics of noisy bounded-gate quantum circuits \cite{shao2023simulating,fontana2023classical,huang2023learning}. Secondly, it is vital to extend our theoretical results from average-case scenarios to worst-case scenarios, wherein classical input can be sampled from arbitrary distributions rather than solely from the uniform distribution \cite{arunachalam2017guest,huang2021information}. Such extensions would deepen our understanding of the capabilities and limitations of employing machine learning to comprehend quantum circuits. Moreover, there exists ample opportunity to enhance our proposed learning algorithm by exploring alternative kernel methods, such as the positive good kernels \cite{stein2011fourier} adopted in Ref.~\cite{che2023exponentially}. In addition, independent of this work, a crucial research topic is understanding the hardness of classically simulating the incoherent dynamics of bounded-gate quantum circuits with simple input quantum states in the measure of the averaged prediction error. Last, it would be intriguing to explore whether deep learning algorithms \cite{qian2023multimodal} can achieve provably improved prediction performance and efficiency for specific classes of quantum circuits.

\medskip
\noindent\textbf{METHODS}

\noindent{\textbf{Computational efficiency}}. The proposed ML model, as mentioned in the main text, is not only sample efficient but also computationally efficient when the quantity $C$ is well bounded. Here we briefly discuss why the proposed ML model is computationally efficient and defer the detailed analysis in SM~G. 

Recall that our proposal consists of two stages: training dataset collection and model inference. As a result, to demonstrate the efficiency of our proposal, it is equivalent to exhibiting the computational efficiency at each stage. 

 \noindent{Training dataset collection}. This stage amounts to loading the collected training dataset $\mathcal{T}_{\mathsf{s}}=\{\bxi, \tilde{\rho}_T(\bxi)\}_{i=1}^n$ to the classical memory. For the classical input controls $\{\bxi \}_{i=1}^n$, $\mathcal{O}(d n)$ computation cost is sufficient. For the classical shadow $\{\tilde{\rho}_T(\bxi) \}_{i=1}^n$, the required computational cost is $\mathcal{O}(nNT)$, where $T$ refers to the number of snapshots. According to Theorem~\ref{thm:learning-non-trunc}, the number of training examples $n$ polynomially scales with $N$ and $d$ for an appropriate $C$. Taken together, this stage is computationally efficient when $C$ is well-bounded.    

 \noindent{Inference  stage}. Following the definition of our model in Eq.~(\ref{eqn:generic-learner}), its inference involves summing over the assessment of each training example  $(\bxi, \tilderho_T(\bxi))$, and the evaluation of each training example is composed of two components, i.e., the shadow estimation $\Tr(O\tilde{\rho}_T(\bxi))$ and the kernel function calculation $\kappa_{\Lambda}(\bx, \bxi)$. 
 
 As for the Pauli-based shadow estimation, the computation of each $\Tr(O_i\tilde{\rho}_T(\bxi))$ can be completed in $\mathcal{O}(T)$ time after storing the classical shadow in the classical memory. In other words, the computation can be completed in $\mathcal{O}(Tq)$ runtime with $O=\sum_{i=1}^q O_i$. As for the  kernel function calculation, it can be completed in $
 	\mathcal{O}(|\mathfrak{C}(4C/\epsilon)|)$ runtime. 
 
 Taken together, the inference stage is also computationally efficient when $C$ is well bounded, where the required computational cost is   	$\mathcal{O}(n(Tq+|\mathfrak{C}(4C/\epsilon)|))$.

\smallskip
\noindent \textbf{Proof Sketch of Theorem~\ref{thm:learn-complexity}}. The proof of Theorem~\ref{thm:learn-complexity} is reached by separately analyzing the lower and upper bounds of the sample complexity, and the lower bound of the runtime complexity.

\noindent{{Sample complexity}}. We separately analyze the lower and upper bounds of the sample complexity. As for the lower bound, we design a specific learning task in which the concept class is $ \hat{\mathcal{F}}= \{f_{\bm{a}}(\bx) \big| \bm{a}\in \{0, 1\}^d, \bx\sim [-\pi,\pi]^d   \}$ with $f(\bx) = \sqrt{2\epsilon} B   \cos (\sum_{i=1}^d (1-\bm{a}_i)\bx_i)$. We demonstrate that such concepts can be effectively prepared by a two-qubit circuit with $d$ $\RZ$ gates and $G-d$ Clifford gates. We then harness Fano's lemma to establish the lower bound of the sample complexity when learning $\hat{\mathcal{F}}$, i.e., $\Pr[\bm{a}^* \neq \bar{\bm{a}}] \geq 1 -\frac{I(\bm{a}^*; \bar{\bm{a}}) + \log 2}{\log | \hat{\mathcal{F}}|}$ with $\bm{a}^*$ being the target concept and $\bar{\bm{a}}$ being the inferred concept. Specifically, we independently quantify the upper bound of the mutual information $I(\bm{a}^*; \bar{\bm{a}})$ and the cardinality of the concept class $|\hat{\mathcal{F}}|$. Supported by the results of information theory, we obtain $ I(\ba^*; \barba) \leq n\cdot \min\{\frac{\epsilon T}{1 - \epsilon}, N\log2\}$, independent with the gate count $d$ or $G$. Meanwhile, we have $|\hat{\mathcal{F}}|=2^d$. Taken together, we reach the lower bound. 

 As for the upper bound, we leverage hypothesis testing to reformulate its derivation as a task of quantifying the cardinality of the packing net for the concept class $\mathcal{F}$. Particularly, the cardinality is determined by all possible arrangements of $d$ $\RZ$ gates and $G-d$ Clifford gates across $N$-qubit wires,  which is inherently a combinatorial problem. Through analysis, we obtain that the upper bound of the packing number of $\mathcal{F}$ is $ \binom{G}{d}\cdot N^d \cdot 3^{G-d} \cdot \binom{N}{2}^{G-d}$. Then, according to the result of statistical learning theory, the required sample complexity logarithmically depends on this cardinality, leading to the linear scaling with $d$ and $G$.

\noindent{{Runtime complexity}}. To demonstrate the computational hardness, we leverage results from the $\BQP$ complexity class. Specifically, we construct a tailored concept class $\mathcal{F}$ that incorporates a concept $f^*$ realized by the $\BQP$ circuit. Consequently,  if an efficient classical learner were to exist for this task, it would imply that any $\BQP$ language could be decided in $\mathsf{P/poly}$—a scenario that is widely considered implausible.

The implementation of  $f^*$ follows the similar manner with Ref.~\cite{molteni2024exponential}. Let $\mathsf{L}$ be an arbitrary $\BQP$ language. Since $\mathsf{L} \in \BQP$, there exists a corresponding quantum circuit $U^*$ which decides input bitstrings $\bx\in \{0, 1\}^N$ correctly on average with respect to the data distribution. In particular, define $O'=\Z_1\otimes \mathbb{I}_{2^{N-1}}$ as the observable, which measures the Pauli-Z operator on the first qubit. The target concept $f^*$ is defined as $f^*(\bx,O') = \Tr(O'\rho^*(\bx))$ with $\rho=U^*\ket{\bx}\bra{\bx}U^{*\dagger}$. The result $f^*(\bx,O')$, whether positive or negative, determines the decision regarding whether the input bitstrings $\bx$ belongs to  $\mathsf{L}$ or not. As such, we connect the problem of learning mean-value space of bounded-gate circuits with complexity theory. 

The remaining step to complete the proof involves demonstrating that the quantum state $\rho^*(\bx)$ can be efficiently prepared using RZ and Clifford gates. Taken together, this completes the proof.

\smallskip
\noindent \textbf{Proof Sketch of Theorem~\ref{thm:learning-non-trunc}}. We first exhibit that the expected risk is upper bounded by the truncation error $\text{Err}_1$ and the estimation error $\text{Err}_2$ with $\mathbb{E}_{\bx}| h_{\mathsf{s}}(\bx, O) - f(\rho(\bx),O) |^2 \leq (\sqrt{\text{Err}_1} + \sqrt{\text{Err}_2})$, where $\text{Err}_1 = \mathbb{E}_{\bx} [ | \Tr \left(O \rho_{\Lambda}(\bx) \right) - \Tr(O\rho(\bx))  |^2 ]$    $\text{Err}_2 = \mathbb{E}_{\bx} [ | \Tr \left(O \rho_{\Lambda}(\bx) \right) - h_{\mathsf{s}}(\bx, O)  |^2 ]$. Here we expand the explored state under the trigonometric monomial basis in Eq.~(\ref{eqn:trigono-basis}) with $\rho_{\Lambda}=  \sum_{\bomega \in \mathfrak{C}(\Lambda)}\Phi_{\bomega}(\bx)$  and  
	$\rho= \rho_{d}$.  Intuitively, $\text{Err}_1$ measures the discrepancy between the target state and the truncated state, and $\text{Err}_2$ measures the difference between the prediction and the estimation of the truncated state.

 The upper bound of $\text{Err}_1$ is established by connecting with the gradient norm of the expectation value $\Tr(\rho(\bx)O)$. More concisely, supported by the results of Fourier analysis, we prove that $\text{Err}_1$ is upper bounded by $\mathbb{E}_{\bx\sim [-\pi, \pi]^d}\|\nabla_{\bx} \Tr(\rho(\bx)O)\|_2^2/{\Lambda}$. Then, under the assumption $\mathbb{E}_{\bx\sim [-\pi, \pi]^d}\|\nabla_{\bx} \Tr(\rho(\bx)O)\|_2^2\leq C$, we obtain $\text{Err}_1\leq C/\Lambda$. The upper bound of $\text{Err}_2$ is obtained by exploiting the results of Pauli-based shadow estimation. Particularly, we first prove that $h_{\mathsf{s}}(\bx, O)$ is an unbiased estimator of $\rho_{\Lambda}(\bx)$. Then we make use of the results of Pauli-based shadow to quantify the estimation error, i.e., $\text{Err}_2\leq \tilde{\mathcal{O}}(|\mathfrak{C}(\Lambda)|  \frac{1}{2n} B^2 9^K)$. Taken together, we reach the complexity bound in Theorem~\ref{thm:learning-non-trunc}.

\smallskip
\noindent{\textbf{Code implementation and simulation details}}. All numerical simulations in this work are conducted using Python, primarily leveraging the Pennylane and PastaQ libraries. For small qubit sizes with $N\leq 20$, we use Pennylane to generate training examples and make predictions for new inputs. For larger qubit sizes with $N> 20$, we utilize PastaQ to generate training data and Numpy to handle post-processing and new input predictions.

The code implementation of our proposal proceeds as follows. The collected shadow states are alternately stored in the HDF5 binary data format and the standard binary file format in NumPy. For instance, when $N=60$ and $T=5000$, the file size of each shadow state is 14.4 MB. Utilizing advanced storage strategies can further reduce memory costs. The representation used in the kernel calculation is obtained via an iterative approach. In particularly, we iteratively compute a set of representations $\{\Phi_{\bomega}\}$ satisfying $\|\bomega\|_0=\lambda$ and $\lambda\leq \Lambda$. Given a specified $\lambda$, we identify the indices with nonzero values via a combinatoric iterator and compute the corresponding representations according to Eq.~(\ref{eqn:trigono-kernel}). It is noteworthy that adopting distributed computing strategies can further improve the computational efficiency of our method.    
 
In the task of predicting the two-point correlation of rotational GHZ states (i.e., Fig.~\ref{fig:Rot-GHZ-main}(a)), the test value used in the prediction stage is $\bx = [1.234, -1.344, -1.716]$. All diagonal elements in the correlation matrix are set as $1$.

	\medskip
	
	\noindent \textbf{DATA AVAILABILITY}
	
	The data generated and analyzed during the current study are publicly available at the following repositories:\\ Github: \url{https://github.com/yuxuan-du/Efficient_Predicting_Bounded_Gate_QC}; \\Google Drive:  \url{https://drive.google.com/drive/folders/1lzDyZ7d3Wsm2ncFMGhkUjJg7HrI3ouXy}. 
	
These datasets support the findings of the study and include all data necessary to reproduce the results presented in the main text and supplementary materials.

	\medskip
	
	\noindent \textbf{CODE AVAILABILITY}
	
The code used in this study is available at the Github repository \url{https://github.com/yuxuan-du/Efficient_Predicting_Bounded_Gate_QC}.

 	\medskip 
 	
 	\noindent \textbf{ACKNOWLEDGEMENTS}\\
 	D.T. is supported by the National Research Foundation, Singapore, under its NRF Professorship Award No. NRF-P2024-001.  
 	
  	\medskip 	
	\noindent \textbf{AUTHOR CONTRIBUTIONS}\\
	The project was conceived by Y.D., and D.T. Theoretical results were proved by Y.D., and M.H. Numerical simulations and analysis were performed by Y.D. All authors contributed to the write-up.

	\medskip
	
	\noindent \textbf{COMPETING INTERESTING}\\
	The authors declare no competing interests.

\clearpage
\newpage

\onecolumngrid

\appendix

\tableofcontents
\renewcommand{\appendixname}{SM}
\renewcommand{\thetable}{\arabic{table}}
\renewcommand{\tablename}{Supplementary Table}
\renewcommand{\thefigure}{\arabic{figure}}
\renewcommand{\figurename}{Supplementary Fig.}
\setcounter{figure}{0}

\bigskip

\section{Preliminary and literature review}
This section contains the necessary backgrounds of the whole work. Specifically, in Supplementary Material (SM)~\ref{append:subsec:classical-shadows}, we present the basic information of classical shadow. Then, in SM~\ref{append:subsec:trigo-monomial-exp-QC}, we recap the trigonometric expansion of quantum circuits implemented by $\RZ$ and Clifford gates.  Afterward, in SM~\ref{append:subsec:A:generality},  we elaborate on the generality of the explored problem in our work. Last, we provide a literature review in SM~\ref{append:subsec:literature-review}. 

\subsection{Classical shadow and its application in estimating incoherent dynamics of quantum states}\label{append:subsec:classical-shadows}
 Classical shadow represents a computationally and memory-efficient approach for storing quantum states on classical computers, primarily used for estimating the expectation values of local observables \cite{huang2020predicting}. The fundamental principle of classical shadow lies in the `measure first and ask questions later' strategy. In this subsection, we outline the utilization of classical shadow to estimate linear functions under Pauli-based measurements. Interested readers can refer to tutorials and surveys \cite{huang2022learning,elben2022randomized} for more comprehensive details.

\medskip 
\noindent\textbf{Formalism of classical shadow.} The general scheme of classical shadow for an unknown $N$-qubit state $\rho$ is repeating the following procedure $T$ times. At each time, the state $\rho$ is first operated with a unitary $U$ randomly sampled from the predefined unitary ensemble $\mathcal{U}$ and then each qubit is measured under the Z basis to obtain an $N$-bit string denoted by $\bm{b} \in \{0, 1\}^N$. There exists a linear map $\mathcal{M}(\cdot)$ satisfying 
\begin{equation}
	\mathcal{M}(\rho) =\mathbb{E}_{U\sim \mathcal{U}}\mathbb{E}_{\bm{b}\sim \PP(\bm{b})} U^{\dagger}\ket{\bm{b}}\bra{\bm{b}} U    = \sum_{\bm{b}} \int \mathsf{d} U U^{\dagger}\ket{\bm{b}}\bra{\bm{b}} U\bra{\bm{b}}U\rho U^{\dagger}\ket{\bm{b}},
\end{equation} 
where $\PP(\bm{b})=\bra{\bm{b}}U\rho U^{\dagger}\ket{\bm{b}}$. Such a linear map implies that the unknown state $\rho$ can be formulated as 
\begin{equation}
	\rho =  \sum_{\bm{b}} \int \mathsf{d}U \mathcal{M}^{-1}\Big(U^{\dagger}\ket{\bm{b}}\bra{\bm{b}} U\Big) \bra{\bm{b}}U\rho U^{\dagger}\ket{\bm{b}}.
\end{equation}

In other words, the state $\rho$ can be estimated by sampling the snapshot with $M$ times following the distribution $\PP(\bm{b})$. Define the $t$-th snapshot as $U_t^{\dagger}\ket{\bm{b}_t}\bra{\bm{b}_t}U_t$ with $t\in [T]$ and $U_t\sim \mathcal{U}$. The estimated state of $\rho$  is 
 \begin{equation}
	\hat{\rho}_T = \frac{1}{T}\sum_{t=1}^T \hatrho_t,~\text{with}~\hatrho_t=\mathcal{M}^{-1}(U_t^{\dagger}\ket{\bm{b}_t}\bra{\bm{b}_t} U_t).
\end{equation}

As pointed out in Ref.~\cite{huang2020predicting}, the random Pauli basis is not only experimentally friendly but also enables a succinct form of the classical shadow. When Pauli-based measurements are adopted, it is equivalent to setting the unitary ensemble $\mathcal{U}$ as single-qubit Clifford gates, i.e., $U_t=U_{1,t}\otimes  \cdots U_{j,t} \cdots \otimes U_{N,t}\sim \mathcal{U}=\CI(2)^{\otimes N}$ with uniform weights. In this case, the inverse snapshot takes the form as  
\begin{equation}\label{append:SM-A-CS-form}
	\hatrho_t  = \mathcal{M}^{-1}\left(\bigotimes_{j=1}^N U_{j,t}^{\dagger} |\bm{b}_{j,t}\rangle\langle \bm{b}_{j,t}|U_{j,t}\right) = \bigotimes_{j=1}^N \mathcal{D}_{1/3}^{-1} \left(U_{j,t}^{\dagger} |\bm{b}_{j,t}\rangle\langle \bm{b}_{j,t}|U_{j,t} \right) = \bigotimes_{j=1}^N  \left(3U_{j,t}^{\dagger} |\bm{b}_{j,t}\rangle\langle \bm{b}_{j,t}|U_{j,t} - \mathbb{I}_2 \right), 
\end{equation} 
 where $\mathcal{D}_{1/3}^{-1}(Y)= 3Y - \Tr(Y)\mathbb{I}$. 

\medskip  
\noindent\textbf{Estimate incoherent dynamics (linear properties of quantum states).} The tensor product form of the classical shadow in Eq.~(\ref{append:SM-A-CS-form}) allows an efficient procedure to predict linear properties of the state $\rho$. A typical instance is estimating the expectation value $\Tr(\rho O)$ with $O$ being a local observable. Mathematically, suppose the local observable to be a Pauli string, i.e., $O=P_1\otimes ...P_i...\otimes P_N$ with $P_i\in\{X,Y,Z,\mathbb{I}\}$ for $\forall i \in [N]$, the estimation of the classical shadow is
\begin{equation}
	\Tr\left(\hatrho_T (P_1\otimes ...P_i...\otimes P_N)  \right) = \frac{1}{T}\sum_{t=1}^T\prod_{j=1}^N \Tr\left(\left(3U_{j,t}^{\dagger} |\bm{b}_{j,t}\rangle\langle \bm{b}_{j,t}|U_{j,t} - \mathbb{I}_2 \right)P_{j}\right), 
\end{equation}  
 which is memory and computation efficient. Namely, one only needs $\mathcal{O}(NT)$ bits to store $\hat{\rho}_T$ and $\mathcal{O}(NT)$ computational time to load  $\hat{\rho}_T$ to the classical memory. Next, when the locality of the observable $O=\sum_iO_i$ is $K\sim \mathcal{O}(1)$, the shadow estimation of the expectation value, i.e., $\Tr(\hat{\rho}O_i)$, can be performed in $\mathcal{O}(T)$ time after $\hat{\rho}_T$ is loaded into the classical memory \cite{huang2020predicting}.

\subsection{Pauli transfer matrix and Trigonometric  expansion of $\RZ+\CI$ quantum circuits}\label{append:subsec:trigo-monomial-exp-QC}
\noindent\textbf{Pauli Transfer Matrix}. Here we use Pauli-Liouville representation to reformulate the quantum state and the observable. Specifically, an $N$-qubit state can be represented as a $4^N$-dimensional vector, whose $i$-th entry is 
\begin{equation}
	|\cdot \rrangle =\Tr(\cdot P_i)~\text{with}~ P_i\in \{\mathbb{I}, X, Y, Z\}^N. 
\end{equation}
For example, the state $\ket{0}^{\otimes N}$ satisfies $(\ket{0}\bra{0})^{\otimes N}=((\mathbb{I} + Z)/2)^{\otimes N}$, which indicates that its representation under the Pauli basis yields 
\begin{equation}
\ket{0}^{\otimes N} \equiv |\bm{0}\rrangle = ((\mathbb{I} + Z)/2)^{\otimes N} = \left([1, 0, 0, 1]^{\top}\right)^{\otimes N}.	
\end{equation} 
Similarly, the normalized Pauli operator $O$ under the Pauli basis yields 
\begin{equation}
	|O \rrangle = \left[\Tr\left(O P_1 \right), \cdots,  \Tr\left(O P_{4^N}\right)\right]^{\top}.
\end{equation}
The unitary operator can also be represented by Pauli basis, i.e., given a parameterized circuit $U(\btheta)$, its Pauli Transfer Matrix (PTM) $\bUnitary(\btheta)$ yields
\begin{equation}\label{append:eqn:PTM_general}
	[\bUnitary(\btheta)]_{ij} = \llangle P_i|\bUnitary(\btheta)| P_j \rrangle = \Tr\left(P_iU(\btheta)P_j U(\btheta)^{\dagger}\right). 
\end{equation}
For example, the PTM representation of $\RZ(\bx_j)$ gate is 
\begin{equation}\label{append:eqn:PTM_RZ}
\bRZ(\bx_j) = 
\begin{bmatrix}
   1 & 0 & 0 & 0 \\
   0 & \cos(\bx_j) & \sin(\bx_j) & 0 \\
    0 & -\sin(\bx_j) &  \cos(\bx_j) & 0 \\
     0 & 0 & 0 & 1 
   \end{bmatrix} = D_0 + \cos(\bx_j)D_1 + \sin(\bx_j) D_{-1},
\end{equation}
 where   $D_0=\begin{bmatrix}
   1 & 0 & 0 & 0 \\
   0 & 0 & 0 & 0 \\
    0 & 0 &  0 & 0 \\
     0 & 0 & 0 & 1 
   \end{bmatrix}$, $D_1 = \begin{bmatrix}
   0 & 0 & 0 & 0 \\
   0 & 1 & 0 & 0 \\
    0 & 0 &  1 & 0 \\
     0 & 0 & 0 & 0 
   \end{bmatrix}$, and $D_{-1}=\begin{bmatrix}
   0 & 0 & 0 & 0 \\
   0 & 0 & -1 & 0 \\
    0 & 1 &  0 & 0 \\
     0 & 0 & 0 & 0 
   \end{bmatrix}$.

\medskip   
\noindent\textbf{Trigonometric expansion of $\RZ+\CI$ quantum circuits}. A critical research line in quantum computing involves determining if variational quantum algorithms can provide meaningful advantages over state-of-the-art classical methods through dequantization via Fourier expansion \cite{schuld2022quantum,sweke2023potential}. Despite differing objectives, the techniques developed, particularly the Low-Weight Efficient Simulation Algorithm (LOWESA) \cite{fontana2023classical,rudolph2023classical}, serve as inspiration for our approach to predicting the incoherent dynamics of bounded-gate quantum circuits. In the following context, we delve into the mechanics of LOWESA and postpone the discussion on the connection and difference between our work and this research line to SM~\ref{append:subsec:literature-review}.

We now recap the mechanism of LOWESA. When an $N$-qubit quantum circuit $U(\bx)$ is composed of $d$ $\RZ$ gates and $G-d$ $\CI$ gates, the state representation under Pauli-basis expansion yields
\begin{equation}\label{append:tri-vqa-PTM}
	\rho(\bx) = U(\bx)(\ket{0}\bra{0})^{\otimes N} U(\bx)^{\dagger} =  \sum_{\bomega}\Phi_{\bomega}(\bx) \llangle 0| \bUnitary_{\bomega}^{\dagger} \equiv  \sum_{\bomega}\Phi_{\bomega}(\bx) \rho_{\bomega}.
\end{equation}
The notation $\Phi_{\bomega}(\bx)$ with $\bomega\in \{0, \pm 1\}^d$ refers to the \textit{trigonometric monomial basis} with  values 
\begin{equation}\label{append:eqn:basis-tri-comp}
	\Phi_{\bomega}(\bx) = \prod_{i=1}^d \begin{cases}
		 1 ~ & \textnormal{if}~ \bomega_i = 0 \\
		 \cos(\bx_i) & \textnormal{if}~\bomega_i = 1 \\
		 \sin(\bx_i) &  \textnormal{if}~ \bomega_i = -1
	\end{cases}.
\end{equation}
Moreover, $\bUnitary_{\bomega}$ in Eq.~(\ref{append:tri-vqa-PTM}) refers to the \textit{purely-Clifford circuit} for the path indexed by $\bomega$ in the sense that in the path $\bomega$,  each $\RZ$ gate at the $i$-th position with $i\in [d]$  is replaced by the operator $D_{\bomega_i}$ in Eq.~(\ref{append:eqn:PTM_RZ}).  The quantum mean value under an observable $O$ is
\begin{equation}\label{append:eqn:PTM-mean-value}
	f(\bx, O) \equiv \Tr(\rho(\bx)O)= \sum_{\bomega}\Phi_{\bomega}(\bx) \llangle 0| \bUnitary_{\bomega}^{\dagger} | O \rrangle.
\end{equation} 
It is noteworthy that LOWESA is a purely classical approach to estimating expectation values of simple variational quantum circuits. The term `simple' is reflected by the fact that even for low dimensional classical inputs (i.e., a small $d$), its computational complexity is prohibited by other two key factors: the initial state and the number of $\CI$ gates $G-d$. This is because when the initial state $(\ket{0}\bra{0})^{\otimes N}$ is substitute a complicated one (e.g., a state with exponentially many terms under Pauli basis), and $G-d$ becomes large, the computational overhead of calculating $\llangle 0| \bUnitary_{\bomega}^{\dagger} | O \rrangle$ is unaffordable by classical simulators.

\subsection{Generality of exploring the learnability of $\mathcal{F}$}\label{append:subsec:A:generality}

Here we discuss the generality of the explored mean-value space in the main text, i.e.,
\begin{equation}\label{append:eqn:A3:generality1}
   \mathcal{F} = \left\{f(\bx, O)= \Tr\left(U(\bx)\rho_0 U(\bx)^{\dagger}O\right) \Big|\bx \in [-\pi, \pi]^d \right\},
\end{equation}  
where $U(\bx) = \prod_{l=1}^{d}(\RZ(\bx_l)u_e)$ is composed of $d$ $\RZ(\bx_l)$ gates, $G-d$ $\CI$ gates denoted by $U_e$  with $\CI=\{H, S, \CNOT\}$, $O$ is constituted by multiple local observables with a bounded norm, i.e., $O=\sum_{i=1}^q O_i$ and $\sum_l \|O_i\|_{\infty}\leq B$, and the maximum locality of $\{O_i\}$ is $K$. This formalism encompasses many applications for quantum computing, including variational quantum algorithms, classical-shadow-based algorithms, and quantum system certification. A common feature of these applications is that the exploited quantum circuit can be hardware-efficient. As a result, for a specified quantum device, its executable quantum circuits can be described by $\RZ+\CI$ gates with a fixed layout but with different angles, i.e., $U(\bx)$ in Eq.~(\ref{append:eqn:A3:generality1}). Given an unknown state $\rho_0$ evolved by $U(\bx)$ and incoherently measured by an observable $O$, it forms the mean-value space $  \mathcal{F}$ explored in this work. In the following, we detail how  $\mathcal{F}$ relates to variational quantum algorithms, classical-shadow-based applications, and quantum system certification, respectively.

\smallskip
\noindent\textbf{Variational quantum algorithms}. We briefly review the mechanism of variational quantum algorithms (VQAs). Interested readers can refer to the surveys \cite{cerezo2021variational,bharti2021noisy} for detailed information. VQAs generally consist of an $N$-qubit circuit and a classical optimizer.  In the training stage, VQAs follow an iterative manner to proceed with optimization, where the optimizer continuously leverages the output of the quantum circuit to update trainable parameters of the adopted ansatz, i.e., $\bx$ of $U(\bx)$, to minimize a predefined objective function $\mathcal{L}(\cdot)$. Mathematically, at the $t$-th iteration, the updating rule for the trainable parameters $\bx$ is 
\begin{equation}\label{append:eqn:A3:generality2}
	\bx^{(t+1)}=\bx^{(t)} - \eta \frac{\partial \mathcal{L}\left(f(\bx^{(t)},O), c_1\right)}{\partial \bx},
\end{equation}
where $\eta$ is the learning rate, $c_1\in\mathbb{R}$ is the target label, and $f(\bx^{(t)},O)$ amounts to the output of the quantum circuit defined in Eq.~(\ref{append:eqn:A3:generality1}). The optimization is terminated when the training loss is lower than a threshold or the total number of iterations $T$ exceeds a predefined value~\cite{mcclean2018barren,cerezo2020cost,miao2024equivalence}. Two main protocols of VQAs are quantum neural networks (QNNs) and variational quantum Eigen-solver (VQEs). The former is utilized to solve machine learning tasks such as classification and regression, and the latter is exploited to estimate the ground state energy of a given Hamiltonian. Notably, when QNNs are applied, the classical input $\bx$ should be divided into two parts, where the first part is used to encode training examples (without updating) and the second part refers to the trainable parameters. Refer to SM~\ref{append:sec:num-res} for details.

\smallskip
\noindent\textbf{Classical-shadow-based applications}. A major application of classical shadow, as introduced in SM~\ref{append:subsec:classical-shadows}, is estimating linear properties of quantum states, i.e., $\Tr(\rho O)$. When the measured states $\rho$ are generated by large-qubit quantum devices, whose circuit can be described by $\RZ+\CI$ gates with a fixed layout but with different angles, the formed function space coincides with $\mathcal{F}$ in Eq.~(\ref{append:eqn:A3:generality1}). It is noteworthy  that classical-shadow-based applications are highly entangled with variational quantum algorithms \cite{boyd2022training,sack2022avoiding}, quantum system certification (explained below), and two novel quantum machine learning protocols---flipped model and shadow model \cite{jerbi2023shadows}. For the flipped model, the training examples $\bm{z}$ are encoded into the measurement observables rather than the quantum state, i.e., the prediction of the flipped model yields $\Tr(U(\bx)\rho_0 U(\bx)^{\dagger}O(\bm{z}))$, where $\bx$ refers to trainable parameters. As the mean-value space $\mathcal{F}$ in Eq.~(\ref{append:eqn:A3:generality1}) supports a class of measurement operators $O$ sampled from a distribution $\mathbb{D}_O$, the achieved results in this work can be harnessed to empower the flipped model. Similarly, the key concept of the shadow model is using classical shadow to obtain the classical description of the state $U(\bx)\rho_0 U(\bx)^{\dagger}$, followed by classically estimating $\Tr(U(\bx)\rho_0 U(\bx)^{\dagger}O(\bm{z}))$, which is also under the framework of $\mathcal{F}$ in Eq.~(\ref{append:eqn:A3:generality1}). Taken together, the developed algorithm in this work can greatly reduce the required quantum resources for developing the relevant algorithms. 

\smallskip
\noindent\textbf{Quantum system certification}. As quantum devices scale up to larger system sizes, the demand for application-specific certification tools becomes apparent. These tools must surpass standard approaches, such as fully simulating a device on a classical computer or performing full tomographic reconstruction, which incurs exponential computational overhead with the qubit count. In response, various certification protocols have been developed to extract different levels of information from the explored quantum chip by using minimal quantum resources. Among them, a representative class of certification protocols is estimating the linear property of quantum states generated by the employed quantum chip, such as direct fidelity estimation \cite{flammia2011direct}, entanglement witnessing \cite{horodecki2001separability}, and two-point correlator \cite{huang2022provably}. Notably, all of these tasks can be described by the mean-value space $\mathcal{F}$ in Eq.~(\ref{append:eqn:A3:generality1}).

\subsection{Literature review}\label{append:subsec:literature-review}

Relevant prior literature to our study can be categorized into three groups: solving quantum many-body problems using machine learning, learnability of quantum circuits, and the combination of quantum computing (e.g., variational quantum algorithms and quantum system certification) and classical machine learning. In the following, we elucidate how our study aligns with and distinguishes itself from these earlier works.

 \medskip
 \noindent\textbf{Machine learning for quantum many-body problems}. The seminal work in this context is presented in Huang et al.~\cite{huang2022provably}, demonstrating that machine learning algorithms, informed by data collected in physical experiments, can effectively tackle certain quantum many-body problems that are challenging for classical algorithms. In particular, the $l_2$ \textit{Dirichlet kernel} is proposed to predict the ground state properties of a class of Hamiltonians. 
 
Recent follow-up works have further augmented the capacity of machine learning in addressing quantum many-body problems. Specifically, Ref.~\cite{che2023exponentially} explored a specific class of scenarios in quantum many-body problems and devised \textit{positive good kernels}, achieving a polynomial sample complexity for predicting quantum many-body states and their properties. Furthermore, when concentrating on learning the average of observables with a locality assumption, Refs.~\cite{Lewis2024Improved,onorati2023efficient} achieved a quasi-polynomial sample complexity. Additionally, Ref.~\cite{onorati2023provably} extended the results of learning phases of quantum matter characterized by exponential decay of correlations, to the task of learning local expectation values for all states within a phase.

Our study diverges from this line of research in terms of its distinct objectives. While previous works concentrate on predicting the properties of quantum many-body states, our focus is on predicting the incoherent dynamics of bounded-gate quantum circuits. This disparity in objectives results in the utilization of \textit{different prior information}, with the former relying primarily on assumptions about the explored Hamiltonian, such as locality, while the latter leverages the properties of quantum gates, such as $\RZ + \CI$. As a result, the proposed trigonometric monomial kernel in our proposal is more succinct compared to the $l_2$ Dirichlet kernel in terms of the frequency set, i.e., $\bomega\in \{0, \pm 1\}^d$ in Eq.~(6) versus $\bm{k}\in \mathbb{Z}^d$ in Dirichlet kernel.  Furthermore, the learning efficiency of the algorithm proposed in Ref.~\cite{onorati2023provably} relies on \textit{the smoothness of the Hamiltonians} in the quantum many-body systems under exploration, an assumption that does not apply to the learnability of bounded-gate quantum circuits. A key technical contribution of this work is establishing a connection between the effectiveness of the proposed ML model and the gradients of bounded-gate quantum circuits with respect to the classical inputs.

 \medskip
 \noindent\textbf{Learning quantum circuits}. Several studies have been undertaken to explore the learnability of quantum circuits, a critical aspect of quantum learning theory \cite{anshu2023survey,banchi2023statistical}. Notably, in Ref.~\cite{zhao2023learning}, it was demonstrated that when learning a state generated by a quantum circuit with $G$ two-qubit gates to a small trace distance, a sample complexity scaling linearly in $G$ is both necessary and sufficient. However, the computational complexity for learning states and unitaries must scale exponentially in $G$. While our work shares similar complexity scaling, it diverges in two key aspects: (1) \textit{we address a different problem}, focusing on predicting linear properties of bounded-gate quantum circuits with a \textit{classical learner}, while Ref.~\cite{zhao2023learning} investigated the quantum learner such that quantum resources are allowed in both the training and inference time; and (2) we provide concrete algorithms to balance the trade-off between sample and computational complexities.
 
In Ref.~\cite{huang2024learning}, a polynomial-time classical algorithm was devised to learn the description of any unknown $N$-qubit shallow quantum circuit $U$. Furthermore, they developed another polynomial-time classical algorithm to learn the description of any unknown $N$-qubit state prepared by a shallow quantum circuit $U$ on a 2D lattice. However, unlike their study, our proposed algorithm does not necessitate the circuit to be shallow. Additionally, the key components of our algorithm include classical shadow, kernel method, and trigonometric expansion, whereas their algorithm relies on a quantum circuit representation based on local inversions and a technique to combine these inversions.

In Ref.~\cite{jerbi2023power}, the authors examined the power and limitations of learning unitary dynamics using quantum neural networks through both coherent and incoherent measurements. Similar to Ref.~\cite{zhao2023learning}, the learner in this study is still quantum, meaning that quantum resources are required during inference. In contrast, our work focuses on a different problem where the learner is classical, indicating that the techniques employed to achieve the final results are fundamentally different.

 \medskip
 \noindent\textbf{Combination between quantum computing and machine learning}. For clarity, here we separately discuss how previous studies harness  machine learning to enhance variational quantum algorithms and quantum system certification. 
 
\smallskip
\noindent\textit{Variational quantum algorithms}. The incorporation of machine learning to improve VQAs follows two main approaches. The first approach involves integrating deep neural networks and variational quantum circuits to form a hybrid learning model capable of addressing diverse computational tasks \cite{cervera2021meta,huang2021experimental,zhang2021neural,jain2022graph,lei2024neural}. These learning protocols often lack theoretical guarantees and fall outside the scope of our study. 

The second approach entails designing \textit{classical surrogates} capable of inferring the output of VQAs. The outcomes of this approach not only benefit from assessing whether VQAs can offer meaningful advantages over state-of-the-art classical methods but also contribute to conserving quantum resources for the development of novel VQAs. A notable paradigm in this context, complementary to tensor network methods and Clifford-based simulators, involves the dequantization of VQAs via Fourier expansion \cite{sweke2023potential}. In particular, Refs.~\cite{fontana2022efficient,fontana2023classical,rudolph2023classical,nemkov2023fourier} proposed Fourier-based algorithms to simulate variational quantum algorithms when applied to simulate the expectation value of an observable for an initial state evolved under a unitary quantum circuit, e.g., the tasks covered by variational quantum Eigen-solvers and quantum approximate optimization algorithms. Besides, Refs.~\cite{schreiber2023classical,landman2022classically,casas2023multi,gan2024concept} harness Fourier features to dequantize VQAs, i.e., quantum neural networks, when applied to solve classical machine learning tasks.  

Our approach distinguishes itself from previous Fourier-based methods by offering broader applicability beyond the classical simulation of VQAs. A specific illustration is that prior Fourier-based approaches often impose constraints on the initial state, such as requiring it to be a simple product state $(\ket{0}\bra{0})^{\otimes N}$ interpreted in SM~\ref{append:subsec:trigo-monomial-exp-QC}. In contrast, our proposal eliminates this requirement. This advancement is rooted in the hybrid nature of our approach, which utilizes shadow information from quantum computers to build the training dataset. Technically, our method integrates three distinct techniques—classical shadow, kernel method, and trigonometric expansion—paving the way for the development of novel strategies to enhance variational quantum algorithms with provable guarantees.

\smallskip
\noindent\textit{Quantum system certification}. We next explain how our work relates to learning-based quantum system certification. Prior literature related to this topic can be classified into two categories. The first category is using deep learning to improve quantum tomography  \cite{torlai2018neural}, which has two distinct research lines. The first line involves explicit state reconstruction, where the output of neural networks represents the density matrix of the target quantum state \cite{du2023shadownet,ahmed2021quantum,cha2021attention}. This line exploits the generalization ability of neural networks, where the optimized neural networks can predict the density matrix of unseen states when sampled from the same distribution of the training data. The second line focuses on implicit state reconstruction, wherein neural networks are optimized to emulate the behavior of a given quantum state  \cite{carleo2018constructing,carrasquilla2019reconstructing,smith2021efficient,schmale2022efficient}. Note that the second line merely exploits the fitting power of neural networks, where the optimized neural networks can only reconstruct a single state and do not possess the generalization ability. As such, this research line is beyond the scope of our work. The second category is using deep learning to predict the properties of quantum states, including fidelity estimation \cite{zhang2021direct}, energy estimation \cite{zhu2022flexible}, entropy estimation \cite{nir2020machine}, cross-platform verification \cite{xiao2022intelligent},  and similarity testing \cite{wu2023quantum}. 

Despite distinct applications, the learning paradigm of these two categories is very similar. In particular, a labeled dataset needs to be first collected to conduct supervised learning. The data features of training examples generally are random measurement results, and the label corresponds to the specific tasks, i.e., the density matrix for quantum state tomography and entanglement entropy for entropy estimation. After training, the optimized neural network can predict the unseen state by feeding into the random measurement results. As addressed in the main text, a critical caveat of learning-based quantum system certification is that it lacks a theoretical guarantee. Our work fills this knowledge gap and provides concrete evidence of using various machine learning techniques to comprehend quantum systems.    

\medskip
\noindent\underline{Remark}. We conclude this section by highlighting the complementary nature of our work to quantum tomography and classical simulators in comprehending the behavior of large-qubit quantum devices. In contrast to tomography-based approaches, the offline capability of our proposal could markedly reduce the quantum resource overhead. Moreover, compared to classical simulators, our proposal offers two notable advantages. Firstly, it outperforms classical simulation methods in terms of runtime complexity by removing the reliance on the number of Clifford gates. Secondly, our approach accommodates arbitrarily complex input states, whereas classical simulators often necessitate simple input states such as the product state. A summary of the key focuses between different works is shown in the table below.

\begin{table}[h!]
\caption{\justifying\small{\textbf{Summary of the key focuses between different works}.}}
\label{tab:sum}
\begin{tabular}{|c|c|c|c|}
\hline
             & Learner Type & Access Quantum Resources at Inference & Learning Model                                    \\ \hline
Huang et al. \cite{huang2022provably} & Classical    & No                                    & Dirichlet Kernel                                  \\ \hline
Zhao et al. \cite{zhao2023learning} & Quantum      & Yes                                   & -                                                 \\ \hline
Jerbi et al. \cite{jerbi2023power} & Quantum      & Yes                                   & Quantum Neural Networks                           \\ \hline
Schreiber et al. \cite{schreiber2023classical}   & Classical    & No                                    & Least square  regression for Fourier-based models \\ \hline
This work    & Classical    & No                                    & Trigonometric Monomial Kernel                     \\ \hline
\end{tabular}
\end{table}

\section{Proof of Theorem 1}
 \begin{theorem-non}[Formal statement of Theorem 1] 
 Following notations in Eq.~(3), denote a dataset $\{\bxi, \tilde{f}_T(\bxi, O)\}_{i=1}^n$ containing $n$ training examples with $\bxi\in [-\pi, \pi]^d$ and $\tilde{f}_T(\bxi, O)$ being the estimation of $f(\bxi, O)$ using $T$ incoherent measurements. Then, when  $T\leq \frac{(1-\epsilon^2)}{\epsilon}N\log2$ and $\epsilon\ll 1$,  the training data size 
\begin{equation}\label{eqn:sample-bound-thm1}
\frac{(1 - \epsilon)(C_1 d - \log 2) }{ \epsilon T }  \leq n \leq   	\widetilde{\mathcal{O}}\left(\frac{B^2d+B^2NG}{\epsilon}\right)  
\end{equation}
is sufficient and necessary to achieve $\mathbb{E}_{\bx\sim \mathbb{D}_\mathcal{X}}\left| h(\bx,O)-\Tr(\rho(\bx),O) \right| \leq \epsilon$ with high probability. However, unless $\BQP \subseteq \PPoly$, there exists a class of $G$-bounded-gate quantum circuits that no algorithm can achieve $\mathbb{E}_{\bx\sim[-\pi, \pi]^d}\left| h(\bx,O)-\Tr(\rho(\bx),O) \right| \leq \epsilon$ in a polynomial time.
\end{theorem-non}

The proof of Theorem 1 can be broken down into three parts: the lower bound of sample complexity, the upper bound of sample complexity, and the lower bound of runtime complexity. The corresponding bounds are established in the following three theorems, and the proofs are presented in SM~\ref{append:sec:low-bound-sample}, SM~\ref{Appendix:sample-uppper-bound}, and SM~\ref{append:sec:exp-sepa-bound}, respectively.

 \begin{theorem}[Lower bound of sample complexity]\label{append:thm:low-bound-sample} Consider that a learner collects a dataset $\mathcal{T}=\{\bxi, \tilderho_T(\bxi)\}_{i=1}^n$ containing $n$ training examples to predict unseen states $\rho(\bx)$ with $\bx\in [-\pi, \pi]^d$. Suppose that with high probability, the learned model $h(\cdot,\cdot)$ can achieve 
 \begin{equation}
 	\mathbb{E}_{\bx\sim \mathbb{D}_\mathcal{X}}\left| h(\bx,O)-\Tr(\rho(\bx),O) \right| \leq \epsilon.
 \end{equation} 
Then, when  $T\leq \frac{(1-\epsilon)}{\epsilon}N\log2$,  the training data size must obey
\begin{equation}
	n\geq \frac{(1 - \epsilon)(C_1 d - \log 2) }{ \epsilon T },
\end{equation}
where $C_1\in (0, 1)$ is a constant.
 \end{theorem}
 
 \begin{theorem}[Upper bound of sample complexity]\label{thm:upper-bound-sample-comp}
 Consider a learner quires an $N$-qubit quantum circuits and  collects a dataset $\mathcal{T}=\{\bxi, \tilderho_T(\bxi)\}_{i=1}^n$  containing $n$ training examples to predict unseen states $\rho(\bx)$ with $\bx\in [-\pi, \pi]^d$. Suppose  with high probability, there exists a learned model $h(\cdot,\cdot)$ that can achieve 
 \begin{equation}
 	\mathbb{E}_{\bx\sim \mathbb{D}_\mathcal{X}}\left| h(\bx,O)-\Tr(\rho(\bx),O) \right|^2 \leq \epsilon.
 \end{equation} 
using training data of size
\begin{equation}
	n \leq \widetilde{\mathcal{O}}\left(\frac{B^2d+B^2NG}{\epsilon}\right).
\end{equation}
\end{theorem}
 
 \begin{theorem}[Lower bound of runtime complexity]\label{thm:lower-bound-runtime-comp}
There exists a class of bounded-gate quantum circuits composed of $d$ $\RZ$ gates and $G-d$ $\CI$ gates such that no polynomial-time algorithm exists to achieve      
\begin{equation}
 	\mathbb{E}_{\bx \sim  \mathbb{D}_\mathcal{X}}\left| h(\bx,O)-\Tr(\rho(\bx),O) \right|^2 \leq \epsilon,
 \end{equation} 
unless $\mathsf{BQP}\subseteq \PPoly$.
 \end{theorem}
 
\begin{proof}[Proof of Theorem 1]
The proof of Theorem 1 can be immediately obtained by combining the results of Theorems~\ref{append:thm:low-bound-sample}, \ref{thm:upper-bound-sample-comp}, and \ref{thm:lower-bound-runtime-comp}.
\end{proof}

\section{Lower bound of learning bounded-gate quantum circuits with incoherent measurement  (Proof of Theorem~\ref{append:thm:low-bound-sample})}\label{append:sec:low-bound-sample}
 
 In this section, we comprehend the fundamental limitation of learning quantum circuits with a bounded number of non-Clifford gates when only the classical input control and incoherent measurements are allowed, and the circuit layout is unknown to the learner. As elaborated on in the main text, these restrictions echo most experiments for large-qubit quantum computing, i.e., the prediction solely relies on shadow information about the output state without harnessing any prior knowledge about the circuit layout. In addition, as the data distribution $\mathbb{D}_\mathcal{X}$ is arbitrary, here we consider a specific setting in which the classical inputs $\bx$ are sampled from the uniform distribution, i.e., $\bx\sim [-\pi, \pi]^d$. 
 
 \begin{theorem-non}[Restatement of Theorem~ \ref{append:thm:low-bound-sample}]
 Consider that a learner collects a dataset $\mathcal{T}=\{\bxi, \tilderho_T(\bxi)\}_{i=1}^n$ containing $n$ training examples to predict unseen states $\rho(\bx)$ with $\bx\in [-\pi, \pi]^d$. Suppose that with high probability, the learned model $h(\cdot,\cdot)$ can achieve 
 \begin{equation}
 	\mathbb{E}_{\bx\sim[-\pi, \pi]^d}\left| h(\bx,O)-\Tr(\rho(\bx),O) \right| \leq \epsilon.
 \end{equation} 
Then, when  $T\leq \frac{(1-\epsilon)}{\epsilon}N\log2$,  the training data size must obey
\begin{equation}
	n\geq \frac{(1 - \epsilon)(C_1 d - \log 2) }{ \epsilon T },
\end{equation}
where $C_1\in (0, 1)$ is a constant.  
 \end{theorem-non}
 
To reach Theorem~\ref{append:thm:low-bound-sample}, we constitute a class of quantum circuits and analyze the lower bound of the required sample complexity for a learner to achieve the average prediction error $\epsilon$. In particular, we consider a simple family of two-qubit quantum circuits illustrated in Fig.~\ref{fig:circuit-fam}(a), which is composed of $d$ $\RZ$ gates and at most $6d$ $\CNOT$ gates, as each SWAP gate can be decomposed into three $\CNOT$ gates. The explicit expression of this class of quantum circuits is  
\begin{equation}\label{append:eqn:def-low-bound-circuit-clc}
	\mathcal{F}= \left\{f_{\bm{a}}(\bx) = \Tr\left(W(\bx, \bm{a}) (\rho_0\otimes \ket{0}\bra{0}) W(\bx, \bm{a})^{\dagger} O\right) \big| \bm{a}\in \{0, 1\}^d, \bx\sim [-\pi,\pi]^d   \right\},
\end{equation}
where $\rho_0\otimes \ket{0}\bra{0}$ is the initial two-qubit state with
\begin{equation}
	\rho_0 = \begin{bmatrix}
	\left(\left(\sqrt{1-2\sqrt{\epsilon}} + 	\sqrt{1+2\sqrt{\epsilon}}\right)/2\right)^2 & \sqrt{2\epsilon}/2\\
	\sqrt{2\epsilon}/2 & \left(\left(\sqrt{1-2\sqrt{\epsilon}} - 	\sqrt{1+2\sqrt{\epsilon}}\right)/2\right)^2
\end{bmatrix},
\end{equation}
$W(\bx, \bm{a})$ refers to the parameterized unitary depending on the classical input $\bx$ and the label  $\bm{a}$ with  
\begin{equation}
W(\bx, \bm{a})  = G_d^{\bm{a}_d} (\RZ(\bx_d) \otimes \mathbb{I})G_d^{\bm{a}_d} \cdots G_j^{\bm{a}_j} (\RZ(\bx_j) \otimes \mathbb{I})G_j^{\bm{a}_j} \cdots G_1^{\bm{a}_1}(\RZ(\bx_1) \otimes \mathbb{I})G_1^{\bm{a}_1}
\end{equation}
and $G_0^{0}=\cdots=G_d^{0}=\mathbb{I}_4$ and $G_1^{1}=\cdots=G_d^{1}=\text{SWAP}$, and $O =  B (X \otimes \mathbb{I}_2)$ is the observable.

\begin{figure}[h!]
	\centering
\includegraphics[width=0.65\textwidth]{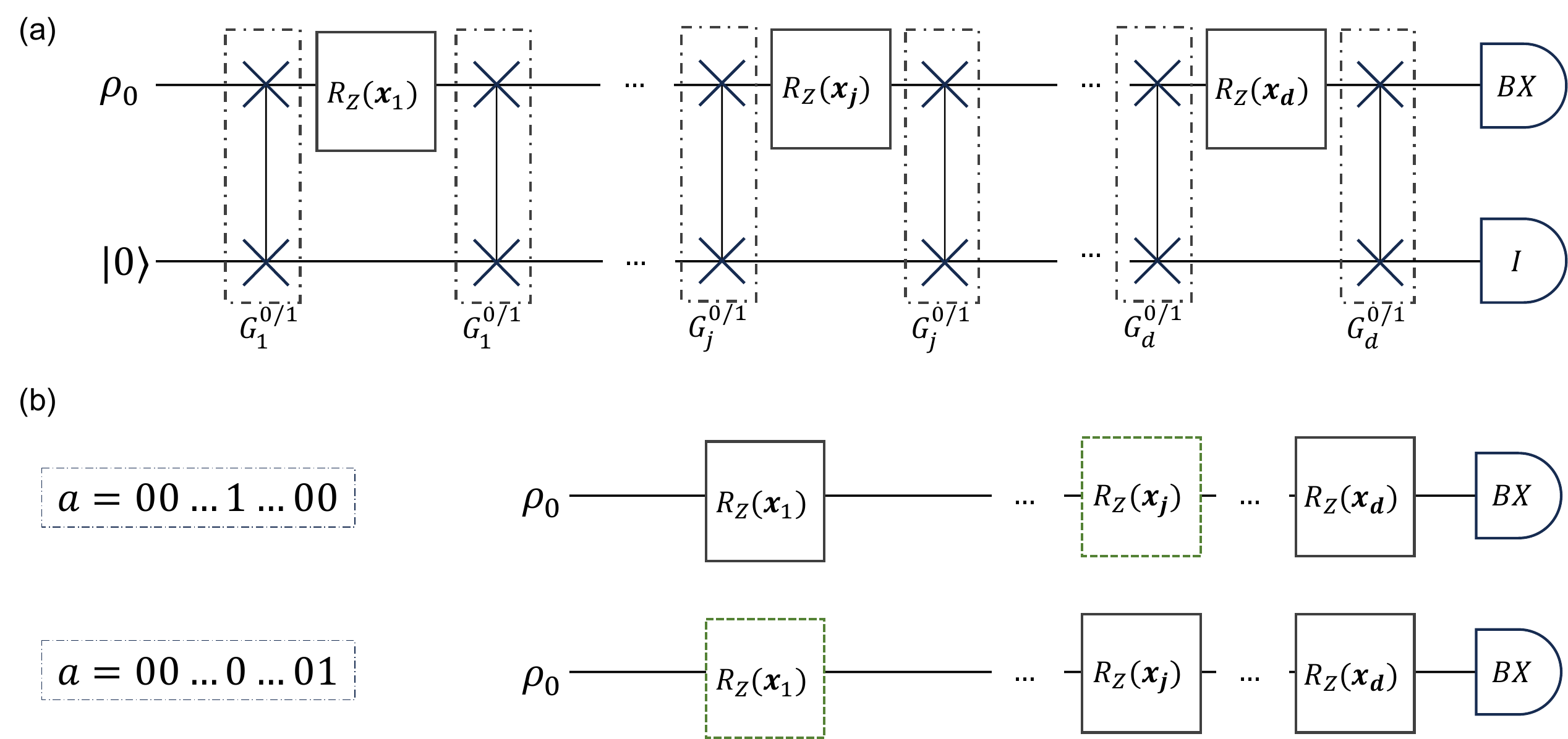}
\caption{\justifying\small{\textbf{Some intuitions about the proof of Theorem B.1}. (a)  \textsc{Illustration of the constructed family of quantum circuits}. The dashed box refers that $\forall i\in[d]$, the SWAP gate $G_i$ is optionally applied to the circuit if the bit string satisfies $\ba_i=1$. (b) \textsc{The visualization of the simplified circuit when $\ba=00...1...0$ (upper) and $\ba=00...0...1$ (lower)}. The dashed green box stands for replacing the $\RZ$ gate with the identity gate.}}
\label{fig:circuit-fam}
\end{figure} 

Note that any candidate in $\mathcal{F}$ can be simplified to a single-qubit circuit, as demonstrated in Fig.~\ref{fig:circuit-fam}(b). That is, the equivalent single-qubit circuit only consists of a sequence of $\RZ(\bx_j)$ gates whose index $\bm{a}_j$ is zero applied to the initial state $\rho_0$, followed by the measurement operator $BX$.  Thus, the simplified expression of this family of circuits is 
\begin{equation}\label{append:eqn:def-low-bound-circuit-clc-simplify}
	\mathcal{F}= \left\{f_{\bm{a}}(\bx) =   B \Tr\left( V^{\bm{a}_d}(\bx_d)  \cdots V^{\bm{a}_1}(\bx_1) \rho_0 V^{\bm{a}_1}(\bx_1)^{\dagger} \cdots V^{\bm{a}_d}(\bx_d)^{\dagger}  X \right) \big| \bm{a}\in \{0, 1\}^d, \bx\sim [-\pi,\pi]^d   \right\},
\end{equation}
where  $V^0(\bx_j)=\RZ(\bx_j)$ and $V^1(\bx_j)=\mathbb{I}_2$ for $\forall j \in [d]$.  It is evident that the cardinality of $\mathcal{F}$  is $|\mathcal{F}|=2^d$.

To proceed further analysis, we now quantify the average distance of different functions in $\mathcal{F}$. To do so, we reformulate $f_{\bm{a}}(\bx)\in \mathcal{F}$ based on the trigonometric expansion. According to PTM representations in Eq.~(\ref{append:eqn:PTM_general}),  the initial state $\rho_0$ and Pauli-$X$ can be written as
\begin{equation}
	|\rho_0\rrangle = \left[1,  \sqrt{2\epsilon}, 0, 0 \right]^{\top} \quad \text{and} \quad |X\rrangle = [0, 1, 0, 0]^{\top}, \quad\text{respectively}.  
\end{equation}
Moreover, following the PTM representations of $\RZ$ in Eq.~(\ref{append:eqn:PTM_RZ}), every function in $\mathcal{F}$ satisfies 
   \begin{subequations} 
   \begin{eqnarray} \label{eqn:explicit-circuit-low-bound}
   &&	f_{\bm{a}}(\bx) = B \left\langle\left\langle \epsilon \Big|\bRZ^{\dagger}\left(\sum_{j,\bm{a}_j=0} \bx_j\right)\Big|X\right\rangle\right\rangle  \\
   	= &&  \sqrt{2\epsilon} B \cos\left(\sum_{j,\bm{a}_j=0} \bx_j\right) \\
   	= && \sqrt{2\epsilon} B   \cos\left((1-\bm{a}_1)\bx_1 + (1-\bm{a}_2)\bx_2+...+(1- \bm{a}_d)\bx_d \right). 
   \end{eqnarray}
\end{subequations}  
The derived explicit form  allows us to quantify the average discrepancy of any two candidates in $\mathcal{F}$, i.e., given $\forall f_{\bm{a}}, f_{\bm{a}'}\in \mathcal{F}$ with  $\bm{a} \neq \bm{a}'$, we have
 \begin{subequations}
 	 \begin{eqnarray}
 	&& \mathbb{E}_{\bx \sim [-\pi, \pi]^d} |f_{\bm{a}}(\bx, O) - f_{\bm{a}'}(\bx, O)|^2  \\
 	= && 2\epsilon B^2  \mathbb{E}_{\bx \sim [-\pi, \pi]^d} \left|\cos\Big((1-\bm{a}_1)\bx_1 +  ...+(1- \bm{a}_d)\bx_d \Big)  - \cos\Big((1-\bm{a}'_1)\bx_1 +  ...+(1- \bm{a}'_d)\bx_d \Big)\right|^2 \label{append:eqn:thm-low-bound-discp-1} \\
 	= && 2\epsilon B^2\left(\frac{1}{2} + \frac{1}{2}\right) -   2\epsilon B^2\mathbb{E}_{\bx \sim [-\pi, \pi]^d}\cos\left((1-\bm{a}_1)\bx_1  +...+(1- \bm{a}_d)\bx_d \right)  \cos\left((1-\bm{a}'_1)\bx_1  +...+(1- \bm{a}'_d)\bx_d \right)\label{append:eqn:thm-low-bound-discp-1-1} \\
 	= && 2\epsilon B^2,  \label{append:eqn:thm-low-bound-discp-2}
 \end{eqnarray}
 \end{subequations}
where the first equality follows the explicit form of $f_{\bm{a}}$ and $f_{\bm{a}'}$ in Eq.~(\ref{eqn:explicit-circuit-low-bound}), the second equality employs the two facts: (i) $\forall \bm{a} \in \{0, 1\}^d$, there are $d-\|\bm{a}\|_0$ effective and independent variables $\{\bx_j\}$; (ii) these effective variables satisfy
\allowdisplaybreaks  
\begin{subequations}\label{append:expect-circuits}
	\begin{eqnarray}
	 && \mathbb{E}_{\bx \sim [-\pi, \pi]^d} \left|\cos\left((1-\bm{a}_1)\bx_1 +  ...+(1- \bm{a}_d)\bx_d \right)\right|^2 \\
	=	&& \mathbb{E}_{\bx \sim [-\pi, \pi]^{d- \|\bm{a}\|_0}}\cos^2\Big(\bx_1 +\bx_2 + ... + \bx_{d-\|\bm{a}_0\|}\Big) \\ 
	 =  && \frac{1}{(2\pi)^{d-\|\bm{a}_0\|}} \int \frac{1 + \cos\Big(2\bx_1 + 2\bx_2 + ... + 2\bx_{d-\|\bm{a}_0\|}\Big) }{2} \mathsf{d}\bx \\
	 = && \frac{1}{2} + \frac{1}{(2\pi)^{d-\|\bm{a}_0\|}} \int \frac{ \cos(2\bx_1) \cos\Big(2\bx_2 + ... + 2\bx_{d-\|\bm{a}_0\|}\Big) - \sin(2\bx_1) \sin\Big(2\bx_2 + ... + 2\bx_{d-\|\bm{a}_0\|}\Big) }{2} \mathsf{d}\bx \\
	 = && \frac{1}{2} + \frac{1}{(2\pi)^{d-\|\bm{a}_0\|}} \int_{-\pi}^{\pi}\cos(2\bx_1)\mathsf{d}\bx_1 \int \frac{\cos\Big(2\bx_2 + ... + 2\bx_{d-\|\bm{a}_0\|}\Big)}{2} \mathsf{d}\bx \nonumber \\
	  && \quad - \frac{1}{(2\pi)^{d-\|\bm{a}_0\|}} \int_{-\pi}^{\pi}\sin(2\bx_1)\mathsf{d}\bx_1 \int \frac{\sin\Big(2\bx_2 + ... + 2\bx_{d-\|\bm{a}_0\|}\Big)}{2} \mathsf{d}\bx \\
	  = && \frac{1}{2},
	\end{eqnarray}
\end{subequations}
and the last equality exploits the fact that the second term in Eq.~(\ref{append:eqn:thm-low-bound-discp-1-1}) is zero as explained below. Specifically, we partition $\bx_1,...,\bx_d$ into two groups, depending on the index list $\bm{a}$ and $\bm{a}'$. Denote $R = \sum_{i, \bm{a}_i=\bm{a}'_i =0} \bx_i$ as the summation of entries whose index is zero in both  $\bm{a}$ and $\bm{a}'$, and $\bar{R}=(1-\bm{a}_1)\bx_1 +  ...+(1- \bm{a}_d)\bx_d - R$  and  $\bar{R}' = (1-\bm{a}'_1)\bx_1 +  ...+(1- \bm{a}'_d)\bx_d -  R$ as the summation of the rest effective variables with the index list $\bm{a}$ and $\bm{a}'$, respectively. Without loss of generality, we suppose $\|\ba\|_0 > \|\ba'\|_0$, implying that $R$ contains at least one effective variable. Then, the second term in Eq.~(\ref{append:eqn:thm-low-bound-discp-1-1}) follows    
\allowdisplaybreaks
\begin{subequations}\label{append:eqn:expect-circuits-2}
	\begin{eqnarray}
		&&  \mathbb{E}_{\bx \sim [-\pi, \pi]^d}\cos\left((1-\bm{a}_1)\bx_1  +...+(1- \bm{a}_d)\bx_d \right)  \cos\left((1-\bm{a}'_1)\bx_1  +...+(1- \bm{a}'_d)\bx_d \right) \\
		= && \mathbb{E}_{\bx \sim [-\pi, \pi]^d} \cos(R + \bar{R})  \cos(R+ \bar{R}') \\
		= && \mathbb{E}_{\bx \sim [-\pi, \pi]^d}[\cos(\bar{R})\cos(R) - \sin(\bar{R})\sin(R)][\cos(R)\cos(\bar{R}') - \sin(R)\sin(\bar{R}')] \\
		= && \mathbb{E}_{\bx \sim [-\pi, \pi]^d}[\cos(\bar{R})\cos^2(R)\cos(\bar{R}') -  \cos(\bar{R})\cos(R)\sin(R)\sin(\bar{R}') \nonumber\\
		&&  -  \sin(\bar{R})\sin(R)\cos(R)\cos(\bar{R}') + \sin(\bar{R})\sin^2(R)\sin(\bar{R}')]\\
		= && 0,
	\end{eqnarray}
\end{subequations}
where the last equality exploits the fact that when $R$ contains at least one effective variable, the symmetric property of integral with respect to trigonometric functions gives $ \mathbb{E}_{\bx \sim [-\pi, \pi]^d}\cos(R) = \mathbb{E}_{\bx \sim [-\pi, \pi]^d}\sin(R)=0$.

\medskip
The result in Eq.~(\ref{append:eqn:thm-low-bound-discp-2}) indicates that taking the expectation over $\bx$, all concept functions in $ \mathcal{F}$ are equally distant by $2\epsilon B^2$. Learning this family of circuits can be reduced to the following \textit{multiple hypothesis testing problem}: 
\begin{enumerate}
	\item Alice randomly and uniformly chooses a target concept $f_{\ba^*} \in \mathcal{F}$, or equivalently $\ba^*\in \{0, 1\}^d$;
	\item The training dataset $\mathcal{T}=\{(\bxi, \hatboi) \}_{i=1}^n$ with the size $n$ is collected based on $f_{\ba^*}$, i.e., for the $i$-th training example, the classical input $\bxi$ randomly sampled from $[-\pi, \pi]^d$ is fed into the quantum system described by  $f_{\ba^*}$ and the statistical estimation of $f_{\ba^*}(\bxi, O)$ is obtained by $T$ shots with $\hatboi=\frac{1}{T}\sum_t \hatboit$ and $\mathbb{E}(\hatboi)=f_{\ba^*}(\bxi, O)$;
	\item The learner leverages the training dataset $\mathcal{T}$ to conduct the empirical risk minimization, i.e., 
	\begin{equation}
		h_{\mathcal{T}}= \arg\min_{h\in \mathcal{F}} \frac{1}{n}\sum_{i=1}^n \ell\left(h(\bxi) - \hatboi\right),
	\end{equation}
	where $\ell$ refers to the loss function measuring the difference between the prediction $h(\bxi)$ and the label $\hatboi$;
	\item  The hypothesis testing is conducted to infer the target concept using the learned $h_{\mathcal{T}}$, i.e., the inferred index is 
	\begin{equation}
	\barba = \arg \min_{\ba, f_{\ba}\in \mathcal{F}} \mathbb{E}_{\bx\sim [-\pi, \pi]^d} \left|h_{\mathcal{T}}(\bx)-f_{\ba}(\bx)\right|^2 
	\end{equation}
	and the associated error probability is $\Prob(\barba \neq \ba^*)$. 
\end{enumerate}
Since the average discrepancy between any two concepts in $\mathcal{F}$ is $2\epsilon B^2$ as indicated in Eq.~(\ref{append:eqn:thm-low-bound-discp-2}), the error probability $\Prob(\barba \neq \ba^*)$ becomes zero when $\mathbb{E}_{\bx \sim [-\pi, \pi]^d}|h_\mathcal{T}(\bx)-f_{\ba^*}(\bx, O)|^2 < \epsilon B^2$. In this scenario, $h_{\mathcal{T}}$  has a large average prediction error for the data sampled from other $f_{\ba'}\in \mathcal{F}\setminus \{f_{\ba^*}\}$ in which the average error is at least $\epsilon B^2$, i.e.,   	 
\begin{equation}
	\mathbb{E}_{\bx\sim [-\pi, \pi]^d} |h_{\mathcal{T}}(\bx) - f_{\ba'}(\bx, O)|^2 \geq \mathbb{E}_{\bx\sim [-\pi, \pi]^d} |f_{\ba'}(\bx, O) - f_{\ba^*}(\bx, O)|^2 - \mathbb{E}_{\bx\sim [-\pi, \pi]^d} |h_{\mathcal{T}}(\bx) - f_{\ba^*}(\bx, O)|^2 > 2\epsilon B^2 - \epsilon B^2 = \epsilon B^2.
\end{equation}

The multiple hypothesis testing problem reformulated above enables us to use Fano's inequality to derive the lower bound of sample complexity in learning $f_{\ba^*}\in \mathcal{F}$.  To be concrete, Fano's lemma \cite{cover1999elements} states that 
	\begin{equation}\label{append:eqn:fano-inequ}
		\Prob[\ba^* \neq \barba] \geq 1 -\frac{I(\ba^*; \barba) + \log 2}{\log | \mathcal{F}|},
	\end{equation}
where $I(\ba^*; \barba)$ refers to the mutual information between random variables $\ba^*$ and $\barba$, and $| \mathcal{F}|$ denotes the cardinality of $\mathcal{F}$. In other words, the derivation of the lower bound of the sample complexity $n$ amounts to quantifying the upper bound of $I(\ba^*; \barba)$ and the lower bound of $|\mathcal{F}|$, which motivates the following proof of Theorem~\ref{append:thm:low-bound-sample}.  

\smallskip
\noindent\underline{Remark}. For ease of analysis, in the proof of Theorem~\ref{append:thm:low-bound-sample}, we consider the shot number $T$ is \textit{sufficiently large} such that the measured results can be approximated by the normal distribution with the mean $\mu_{\ba}=f_{\ba}(\bx)$. This assumption is widely utilized in quantum computation and quantum information theory  \cite{li2014second,mcclean2016theory,torlai2020precise,smith2021qubit}.  Besides, for each $\ba$, the corresponding variance of measured results is assumed to be equal with the varied $\bx$, i.e., for all $\bx$, the variance is $\nu_{\ba}=\mathbb{E}_{\bx\sim [-\pi, \pi]^d}\nu_{\ba}(\bx)$. Note that the assumption can be omitted if the primary focus is solely on the number of $\RZ$ gates $d$, without considering the measurement cost $T$. In this case, the upper bound of the mutual information $I(\bm{a}^*; \bar{\bm{a}})$ can be effectively derived using Holevo’s theorem \cite{holevo1973bounds,bengtsson2017geometry}, where a similar lower bound of the sample complexity can still be obtained with a linear dependence on $d$.

\begin{proof}[Proof of Theorem~\ref{append:thm:low-bound-sample}] According to the above elaboration, we next separately quantify $|\mathcal{F}|$ and the mutual information $I(\ba^*; \barba)$, followed by Eq.~(\ref{append:eqn:fano-inequ}) to attain the lower bound of sample complexity in learning the incoherent dynamics of bounded-gate quantum circuits.

\medskip
\noindent\underline{\textit{Cardinality of $\mathcal{F}$}}. The cardinality of $\mathcal{F}$ can be obtained following its definition in Eq.~(\ref{append:eqn:def-low-bound-circuit-clc}), i.e.,
\begin{equation}
	 |\mathcal{F}| = 2^{d}. 
\end{equation}

\noindent\underline{\textit{Upper bound of mutual information $I(\ba^*; \barba)$}}. Recall that the process of learning $f_{\ba^*}$ implies the Markov chain: 
\begin{equation}
\ba^*\rightarrow \rho_{1:n} \rightarrow \hatbo_{1:n} \rightarrow \barba,
\end{equation}
where for ease of notation, $\rho_{1:n}$ refers to the abbreviation of $n$ resulting states  $\rho_{\ba^*}(\bx^{(1)}), ..., \rho_{\ba^*}(\bx^{(n)})$ before taking the measurements, and $\hatbo_{1:n}$ refers to the abbreviation of the measured statistic results  $\hat{\bm{o}}_{\ba^*}^{(1)},...,\hat{\bm{o}}_{\ba^*}^{(n)}$ of $n$ training examples.   Then, according to the data processing inequality, the mutual information $I(\ba^*; \barba)$ is upper bounded by the mutual information between the target index $\ba^*$ and the measured statistical results $\hatbo_{1:n}$, i.e.,  
\begin{equation}
		I(\ba^*; \barba)   \leq  I(\ba^*; \hatbo_{1:n}).
\end{equation}
Moreover, the mutual information on the right-hand side is upper-bounded by 
\begin{subequations}
\begin{eqnarray}
	 && I(\ba^*; \hatbo_{1:n})   \\
	 \leq && \sum_{i=1}^n I(\ba^* ; \hatbo_i) \\
	\leq &&  \sum_{i=1}^n I(\ba^*; \hatbo_i, \bxi) \\
	 = && \sum_{i=1}^n I(\ba^* ;\bxi) + I(\ba^* ; \hatbo_i | \bxi) \\
	= &&  \sum_{i=1}^n  \mathbb{E}_{\bxi \sim [-\pi, \pi]^d} I(\ba^* |\bxi; \hatbo_i | \bxi) \\
	 = && \sum_{i=1}^n \mathbb{E}_{\bxi \sim [-\pi, \pi]^d} I(\ba^* ; \hatbo_i | \bxi).
\end{eqnarray}	
\end{subequations}
The first inequality stems from the chain rule and the fact that conditioning reduces entropy, i.e., $I(X;Y_{1:n})=\sum_{i=1}^n H(Y_i|Y_{1:i-1}) - H(Y_i|X,Y_{1:i-1}) = \sum_{i=1}^n H(Y_i|Y_{1:i-1}) - H(Y_i|X ) \leq \sum_{i=1}^n H(Y_i) - H(Y_i|X ) = \sum_{i=1}^n I(X;Y_i)$, the second inequality exploits the relation $I(X;Y,Z)=H(X)-H(X|Y,Z)\geq H(X)-H(X|Y)=I(X;Y)$, the first equality employs the chain rule of mutual information with $I(X;Y,Z)=I(X;Y)+I(X;Z|Y)$, the second equality adopts the independence between $\ba^*$ and $\bxi$ with $I(\ba^*;\bxi)=0$ and the KL divergence reformulation of mutual information, and the last equality comes from the independence between $\ba^*$ and $\bxi$. 

The above relation hints that the prerequisite to derive the upper bound of  $I(\ba^*, \hatbo_{1:n})$ is upper bounding $I(\ba^* ; \hatbo_i | \bxi)$. As such, we apply KL divergence formulation to $I(\ba^* ; \hatbo_i | \bxi)$ and obtain 
\begin{subequations}
	\begin{eqnarray}
		I(\ba^* ; \hatbo_i | \bxi) = && \DKL(P_{\ba^*, \bm{o}| \bxi} \|P_{\ba^*}P_{\bm{o}| \bxi}) \\
		= && \sum_{\ba^*} \int p(\ba^*, \bm{o}| \bxi)\log\frac{p(\ba^*, \bm{o}| \bxi)}{p(\ba^*)p(\bm{o}| \bxi)}d \bm{o} \\
		= && \sum_{\ba^*} p(\ba^*) \int p( \bm{o}|\ba^*, \bxi)\log\frac{p(\bm{o}|\ba^*, \bxi)}{ p(\bm{o}| \bxi)}d \bm{o} \\
		= && \frac{1}{|\mathcal{F}|} \sum_{\ba^*}\DKL(P_{\bm{o}|\ba^*, \bxi} \| P_{\bm{o}|\bxi}) \\
		\leq &&  \frac{1}{|\mathcal{F}|^2} \sum_{\ba^*, \ba'}\DKL(P_{\bm{o}|\ba^*, \bxi}\|  P_{\bm{o}|\ba', \bxi}), 
	\end{eqnarray}
\end{subequations}
where the third equality uses $p(\ba^*, \bm{o}| \bxi)=p(\ba^*)p(\bm{o}|\ba^*, \bxi)$, the fourth equality comes from the fact that $\ba^*$ is uniformly sampled with $p(\ba^*)=1/|\mathcal{F}|$, and the inequality employs the property of KL divergence.

Combining the above results, the mutual information is upper bounded by
\begin{equation}\label{append:muInfo-upper}
	I(\ba^*; \barba) \leq \frac{1}{|\mathcal{F}|^2} \sum_{i=1}^n \mathbb{E}_{\bxi\sim[-\pi, \pi]^d} \sum_{\ba^*, \ba'}     \DKL\left(P_{\bm{o}|\ba^*, \bxi}\big\| P_{\bm{o}|\ba', \bxi}\right).
\end{equation}
According to our assumption, when the  shot number $T$ becomes large, the central limit theorem suggests that $P_{\bm{o}|\ba^*, \bxi}$ follows the Gaussian distribution with the mean $\mu_{\ba^*}(\bxi)$ and the  variance $\nu^2$. Moreover, assuming that the variance of the varied candidate is the same, we have
\begin{equation}\label{append:eqn:low-bound-kl-gauss} 
\mathbb{E}_{\bxi \sim [-\pi, \pi]}   	\DKL\left(P_{\bm{o}|\ba^*, \bxi}\big\| P_{\bm{o}|\ba', \bxi}\right) =  \mathbb{E}_{\bxi \sim [-\pi, \pi]^d}  \frac{(\mu_{\ba^*}(\bxi) - \mu_{\ba'}(\bxi))^2}{2\nu^2}. 
\end{equation}
For the nominator in Eq.~(\ref{append:eqn:low-bound-kl-gauss}), the result of Eq.~(\ref{append:eqn:thm-low-bound-discp-2}) gives  
\begin{equation}\label{append:muInfo-upper-mu-v2}
\mathbb{E}_{\bxi \sim [-\pi, \pi]^d}  (\mu_{\ba^*}(\bxi) - \mu_{\ba'}(\bxi))^2  = \mathbb{E}_{\bxi \sim [-\pi, \pi]^d}  (f_{\ba^*}(\bxi) - f_{\ba'}(\bxi))^2 = 2\epsilon B^2. 
\end{equation}
For the generic variance $\nu^2$, it is derived as follows. When the state $\rho_{\ba}(\bx)$ is measured by $O=BX$, the probability of measuring the state   associated with the eigenvalue  $+1$ and the state associated with the eigenvalue $-1$ is
\begin{equation}
	\Prob(+1) -\Prob(-1) =  \frac{f_{\ba}(\bxi)}{B} \Leftrightarrow \Prob(+1)  = \frac{1 + f_{\ba}(\bxi)/B}{2}, \quad \Prob(-)  = \frac{1 - f_{\ba}(\bxi)/B }{2}.
\end{equation}
 This property indicates that when the shot number is $T$, the explicit expression of the variance for the input $\bxi$ is
\begin{equation} 
\nu_{\ba}^2(\bxi) = \frac{1}{T}  \left[ \Big(\Prob(+)*B^2 + \Prob(-) * (-B)^2 \Big)- f_{\ba}(\bxi)^2\right] = \frac{B^2}{T} - \frac{f_{\ba}(\bxi)^2}{T}.
\end{equation}
Taking expectation over $\bx$, we have
\begin{equation}\label{append:eqn:low-bound-variance}
\nu =	\mathbb{E}_{\bxi\sim [-\pi, \pi]^d}(\nu_{\ba, \bxi}^2)= \frac{B^2}{T} - \frac{\mathbb{E}_{\bxi\sim [-\pi, \pi]^d}\left(f_{\ba}(\bxi)^2\right)}{T}  = \frac{B^2 - B^2\epsilon}{T},
\end{equation}
where the last equality uses the result of Eq.~(\ref{append:expect-circuits}).

In conjunction with Eqs.~(\ref{append:muInfo-upper}), (\ref{append:eqn:low-bound-kl-gauss}), (\ref{append:muInfo-upper-mu-v2}), and (\ref{append:eqn:low-bound-variance}), the mutual information is upper bounded by 
\begin{equation}
	I(\ba^*;\barba)\leq   \frac{n\epsilon T}{1 - \epsilon}. 
\end{equation}

Note that the mutual information cannot continuously enhance with respect to the increased shot number $T$. When $T\rightarrow \infty$, the mutual information is upper bounded by
\begin{equation}
	\lim_{T\rightarrow \infty} I(\ba^*;\hat{\bm{o}}_{1:n})\leq I(\ba^*;\rho_{1:n}).
\end{equation} 
Following the results \cite[Lemma 9]{wang2023transition}, the right-hand side term is upper bounded by the number of training examples and the qubit count, i.e.,
\begin{equation}
	I(\ba^*;\rho_{1:n}) \leq nN\log 2.
\end{equation} 
Taken together, we have
\begin{equation}
	I(\ba^*; \barba) \leq n\cdot \min\left\{\frac{\epsilon T}{1 - \epsilon}, N\log2\right\}.
\end{equation}
Reusing the Fano's inequality, we obtain
\begin{subequations}
\begin{eqnarray}
&& I(\ba^*; \barba)  \geq (1 - 	\Prob[\ba^* \neq \barba]) \log | \mathcal{F}| -  \log 2 \\
\Rightarrow && n\cdot \min\left\{\frac{\epsilon T}{1 - \epsilon}, N\log2 \right\} \geq (1 - 	\Prob[\ba^* \neq \barba]) d - \log 2 \\
\Rightarrow && n \geq \frac{(1 - 	\Prob[\ba^* \neq \barba]) d - \log 2 }{\min\left\{\frac{\epsilon T}{1 - \epsilon}, N\log2 \right\}}.
 \end{eqnarray} 
 \end{subequations}
 When  $T\leq \frac{(1-\epsilon)}{\epsilon}N\log2$, the lower bound of the sample complexity can be simplified to
 \begin{equation}
 	n \geq (1-\epsilon) \frac{(1 - 	\Prob[\ba^* \neq \barba]) d - \log 2 }{ \epsilon T }.
 \end{equation} 
 \end{proof}

\section{Upper bound of the sample complexity when incoherently learning the dynamic with bounded-gates (Proof of Theorem~\ref{thm:upper-bound-sample-comp})}\label{Appendix:sample-uppper-bound}

In this section, we analyze the upper bound of the sample complexity when a learner can predict the output of quantum circuits with bounded gates within a tolerable error. A restatement of the corresponding theorem is as follows.
 \begin{theorem-non}[Restatement of Theorem~\ref{thm:upper-bound-sample-comp}]
 Consider a learner quires an $N$-qubit quantum circuits and collects a dataset $\mathcal{T}=\{\bxi, \tilderho_T(\bxi)\}_{i=1}^n$  containing $n$ training examples to predict unseen states $\rho(\bx)$ with $\bx\in [-\pi, \pi]^d$. Suppose  with high probability, there exists a learned model $h(\cdot,\cdot)$ that can achieve 
 \begin{equation}
 	\mathbb{E}_{\bx\sim \mathbb{D}_{\mathcal{X}}}\left| h(\bx,O)-\Tr(\rho(\bx),O) \right|^2 \leq \epsilon.
 \end{equation} 
using training data of size
\begin{equation}
	n \leq  \widetilde{\mathcal{O}}\left(\frac{B^2d+B^2NG}{\epsilon}\right).
\end{equation}
 \end{theorem-non}
Our proof is rooted in utilizing the packing net and packing number, which are advanced tools broadly used in statistical learning theory~\cite{vapnik1999nature} and quantum learning theory~\cite{caro2021generalization,huang2021information,du2022efficient}, to quantify the complexity of the class of function represented by bounded-gate quantum circuits. For elucidating, in the following, we first introduce some basic concepts and results that will be employed in our proof, followed by presenting the proof of Theorem \ref{thm:upper-bound-sample-comp}. Besides, we emphasize that our current focus is solely on the sample complexity, which implies that the computational cost of preparing the training examples may exhibit exponential scaling with the qubit number as indicated in Theorem~\ref{thm:lower-bound-runtime-comp}.

\smallskip
We now elaborate on how to use packing nets to derive the upper bound of the sample complexity of a learner tasked with predicting the incoherent dynamics of bounded-gate quantum circuits. The formal definition of the packing net and packing number is given below.   
\begin{definition}[Packing net/number]\label{def:pac-num}
Let $(\mathcal{U}, \mathfrak{d})$ be a metric space. The subset $\mathcal{V} \subset \mathcal{U}$ is an $\epsilon$-packing net of $\mathcal{U}$ if for any $A, B\in \mathcal{V}$,  $\mathfrak{d}(A, B) \geq \epsilon$. 
The packing number $\mathcal{M}(\mathcal{U}, \epsilon, \mathfrak{d})$ denotes the largest cardinality of an $\epsilon$-packing net of $\mathcal{U}$.  
\end{definition}
Intuitively, an $\epsilon$ packing number refers to the maximum number of elements that can be $\epsilon$-separated.

\smallskip 
We next leverage this definition to prove Theorem~\ref{thm:upper-bound-sample-comp}. Following notations introduced in the main text, given an observable $O$, the function space of quantum circuits with $G$  gates  is 
\begin{equation}\label{append:sample-complexity-function-class}
\mathcal{F} = \left\{f(\bx, O)= \Tr(\rho(\bx)O) \Big|\bx \in [-\pi, \pi]^d, \mathsf{Arc}(\RZ, \CI)\right\},
\end{equation} 
where $\bx \in [-\pi, \pi]^d$ is the classical input control, $\mathsf{Arc}(\RZ, \CI)$ refers to the set of circuit layouts formed by $d$ $\RZ$ gates and $G-d$ Clifford gates,  and $\rho(\bx)$ is the quantum state generated by the bounded gates with $\bx$, i.e., $\rho(\bx)=U(\bx)\rho_0U(\bx)^{\dagger}$ and the layout of $U(\bx)$ follows an element of $\mathsf{Arc}(\RZ, \CI)$. The diversity of circuit layouts stems from the constraint that the learner can solely utilize the classical input and the corresponding measurement results to infer the target concept, where prior information regarding the quantum circuit indicates that it comprises $G$ gates, with $d$ $\RZ$ gates and $d-G$ CI gates. From the perspective of the learner, the quantum computer can yield many circuit layouts, leading that the function space $\mathcal{F}$ contains in total $|\mathsf{Arc}(\RZ, \CI)|$ circuit layouts.      

When the quantum computer is specified with an \textit{unknown but fixed} circuit, the corresponding target concept is denoted by $f^*(\bx, O)\in \mathcal{F}$. To infer $f^*(\bx, O)$, in the training data collection procedure,  the learner samples $n$ classical inputs $\{\bxi\}_{i=1}^n$ from the distribution $\mathbb{D}_{\mathcal{X}}$, sends them into the quantum computer, and obtains the corresponding outcomes $\{\boi\}_{i=1}^n$. The outcome $\boi$ for $\forall i \in [n]$ is obtained by performing a \textit{single-shot} measurement of $O$ (i.e., $T=1$) on $\rho(\bxi)$ with $\mathbb{E}[\boi]=f^*(\bxi, O)$. Let the collected dataset be
\begin{equation}
	\mathcal{T}=\left\{\left(\bxi, \bm{o}^{(i)}\right)\right\}_{i=1}^n.
\end{equation}
What we are interested in here is the required number of training examples $n$ in $\mathcal{T}$ that allows the learner to produce a prediction model $h_{\mathcal{T}}(\cdot)$ whose prediction $h_{\mathcal{T}}(\bx)$ is close to $f^*(\bx)$ on average. More formally, the sample complexity explored here refers to the minimal number of training examples $n$  to ensure that with probability $1-\delta$, the expected risk satisfies  
\begin{equation}\label{append:eqn:sample-complex-def}
	\mathbb{E}_{\bx\sim \mathbb{D}_{\mathcal{X}}}\left|h_{\mathcal{T}}(\bx) - f^*(\bx)\right|^2\leq \mathcal{O}(\epsilon).
\end{equation}

Huang et al. \cite{huang2021information} presented a method that is efficient in sample complexity (but maybe computationally demanding) for solving this task. In particular, according to Definition~\ref{def:pac-num}, denote the $4\epsilon$-packing net of $\mathcal{F}$ as $\mathcal{F}_{4\epsilon}^p$. The learner conducts the empirical risk minimization on the training dataset $\mathcal{T}$ to estimate the optimal hypothesis, i.e.,
\begin{equation}\label{append:eqn:upper-sample-erm}
	\hath_{\mathcal{T}} = \arg\min_{f\in  \mathcal{F}_{4\epsilon}^p} \frac{1}{n}\sum_{i=1}^n \left|f(\bxi) - \bm{o}^{(i)}\right|^2. 
\end{equation}
Under this formalism, Ref.~\cite{huang2021information} proved that the required number of training examples $n$ to achieve an $\epsilon$ prediction error, as summarized in the following proposition.
\begin{proposition}[Adapted from Proposition 1, \cite{huang2021information}]\label{prop:cover-upper-func-learn}
Suppose that the observable $O$ satisfies $\sum_i\|O_i\|_{\infty}\leq B$. 	Let $\hath_{\mathcal{T}}$ be an element of the $4\epsilon$-packing net of $\mathcal{F}$  with the packing number $\mathcal{M}(\mathcal{F}, 4\epsilon, |\cdot|^2)$   that minimizes the empirical training error in Eq.~(\ref{append:eqn:upper-sample-erm}). Then for $\delta\in (0,1)$, the size of training data   
\begin{equation}
	n\geq \frac{38B^2\log(4\mathcal{M}(\mathcal{F}, 4\epsilon, |\cdot|^2)/\delta)}{\epsilon}
\end{equation}
implies 
	\begin{equation}
		\mathbb{E}_{\bx\sim \mathbb{D}_{\mathcal{X}}} \left|\hath_{\mathcal{T}}(\bx) - f^*(\bx) \right|^2\leq 12\epsilon \quad \text{with probability at least $1-\delta$}.
	\end{equation} 
\end{proposition}
\noindent\underline{Remark}. Note that although the original proof of Proposition \ref{prop:cover-upper-func-learn} only concerns $\|O\|\leq 1$, the pertinent proof can be readily extended to the setting of $\sum_i\|O_i\|_{\infty}\leq B$ concerned in Theorem 1, and therefore we omit it here. Moreover, it only focuses on the case of $T=1$, but the results still hold for the case of $T>1$. This is because in the extreme case with $T\rightarrow \infty$, the required training data size is reduced to  $n\geq {38B^2\log(2\mathcal{M}(\mathcal{F}, 4\epsilon, \mathfrak{d})/\delta)}/{\epsilon}$ \cite[Equation (C66)]{huang2021information}. In this regard, the setting of $T\geq 1$ only trivially influences the sample complexity bound (at most logarithmically). For this reason, we omit the relevant analysis in our proof. 

\smallskip
Supported by Proposition~\ref{prop:cover-upper-func-learn}, the proof of Theorem \ref{thm:upper-bound-sample-comp} amounts to quantifying the upper bound of the packing number $\mathcal{M}(\mathcal{F}, \epsilon, |\cdot|^2)$. Note that different from quantum neural networks \cite{caro2021generalization,du2022efficient,du2022power} whose hypothesis space is continuous and the packing number depends on the value of $\epsilon$, the function class $\mathcal{F}$ explored here is discrete and the number of elements in this class finite. As a result, we have $\mathcal{M}(\mathcal{F}, \epsilon, |\cdot|^2)\leq |\mathsf{Arc}(\RZ, \CI)|$ no matter how $\epsilon$ is. In light of this fact, we are now ready to show the proof of Theorem \ref{thm:upper-bound-sample-comp}.

\begin{proof}[Proof of Theorem \ref{thm:upper-bound-sample-comp}]
According to the above explanation, this proof is composed of two parts. The first part is to quantify the upper bound of the packing number $\mathcal{M}(\mathcal{F}, \epsilon, |\cdot|^2)$, or equivalently the cardinality of $\mathsf{Arc}(\RZ, \CI)$. And the second part combines the obtained packing number and Proposition \ref{prop:cover-upper-func-learn} to acquire the upper bound of the sample complexity to learn $\mathcal{F}$.  

\medskip

We now derive the upper bound of $|\mathsf{Arc}(\RZ, \CI)|$. The total number of possible layouts for an $N$-qubit circuit consisting with $d$ $\RZ$ gates and $G-d$ CI gates (i.e., $\text{CI}=\{H,S,\CNOT\}$) is
\begin{equation}\label{append:eqn:cadi-layouts}
	|\mathsf{Arc}(\RZ, \CI)| \leq   \binom{G}{d}\cdot N^d \cdot 3^{G-d} \cdot \binom{N}{2}^{G-d},
\end{equation} 
where the first term $\binom{G}{d}$ computes the number of different arrangements for placing $\RZ$ or $\CI$ gates at each circuit depth, the second term $N^d$  calculates the total combinations of placing $d$ $\RZ$ gates on different qubit wires, the third term $3^{G-d}$ counts the total combinations of choosing different gates from the $\CI$ gate set, and the last term $\binom{N}{2}^{G-d}$ calculates the upper bound for the total combinations of placing the selected $\CI$ gates on different qubit wires.

\medskip
The sample complexity can be efficiently obtained by combining Proposition \ref{prop:cover-upper-func-learn}	 and Eq.~(\ref{append:eqn:cadi-layouts}).  That is, the metric entropy of $4\epsilon$-packing net of $\mathcal{F}$ yields
\begin{subequations}
	\begin{eqnarray}
		&& \log(4\mathcal{M}(\mathcal{F}, 4\epsilon, |\cdot|)/\delta) \\
		\leq &&  \log\left(\binom{G}{d}\right)  + d\log\left(N\right) + (G-d)\log\left(3\right)  + (G-d)\log\left(\binom{N}{2}\right) + \log\left(\frac{4}{\delta}\right)\\
		\leq &&   d\log \left(\frac{eG}{d}\right)   + d\log\left(N\right) + (G-d)\log\left(3\right) + 2(G-d)N\log\left(\frac{eN}{2}\right) + \log\left(\frac{4}{\delta}\right),
	\end{eqnarray}
\end{subequations}
where the second inequality uses $\binom{a}{b}\leq (ea/b)^b$.  

The above result, accompanied by Proposition \ref{prop:cover-upper-func-learn}, suggests the upper bound of the sample complexity to achieve the $\epsilon$-prediction error, i.e., 
\begin{subequations}
\begin{eqnarray}
	n && =  \frac{38B^2\left(  d\log \left(\frac{eG}{d}\right)   + d\log\left(N\right) + (G-d)\log\left(3\right) + 2(G-d)N\log\left(\frac{eN}{2}\right) + \log\left(\frac{4}{\delta}\right)\right)}{\epsilon}  \\
	&& \leq \widetilde{\mathcal{O}}\left(\frac{B^2d+B^2NG}{\epsilon}\right).
\end{eqnarray}
\end{subequations}
	
\end{proof}

\section{Exponential separation between computational complexity and sample complexity (Proof of Theorem~\ref{thm:lower-bound-runtime-comp})}\label{append:sec:exp-sepa-bound}

We analyze the sample complexity of learning bounded-gate quantum circuits with incoherent measurements and showcase the sample efficiency of this learning task in SM~\ref{Appendix:sample-uppper-bound}. However, it is noteworthy that the sample efficiency does not necessarily imply the computational efficiency, as identifying the desired training examples may incur exponential running time. For instance, in the task of pretty-good tomography, the approach outlined in Ref.~\cite{aaronson2007learnability} requires only a linear number of training examples with the qubit count $N$ for a low prediction error, yet the runtime cost scales exponentially with $N$. 

In this section, we prove that learning bounded-gate quantum circuits with incoherent measurements also manifests an exponential separation between sample and computational complexity.

A very recent study \cite{molteni2024exponential} explores the computational hardness of learning quantum observables from the measured out data. Here we first introduce some necessary definitions, and then briefly review their key results, followed by elucidating how to generalize such results to our case to complete the proof of Theorem~\ref{thm:lower-bound-runtime-comp}.

\smallskip
\noindent\textit{$\mathsf{BQP}$ complexity class}. $\mathsf{BQP}$  (short for bounded-error quantum polynomial time) is the class of promise problems that can be solved by polynomial-time quantum computations that may have some small probability of making an error \cite{watrous2008quantum}. Let $A = (A_{\text{yes}}, A_{\text{no}})$ be a promise problem and let $a, b : \mathbb{N} \rightarrow [0, 1]$ be functions. Then $A \in  \mathsf{BQP}(a, b)$ if and only if there exists a polynomial-time generated family of quantum circuits $\mathcal{Q} = \{\mathcal{Q}_N : N \in \mathbb{N}\}$, where each circuit $\mathcal{Q}_N$ takes $N$ input qubits and produces one output qubit, that satisfies the following properties: (1) if $\bx \in A_{\text{yes}}$ then $\Pr(\mathcal{Q} ~ \text{accepts}~ \bx) \geq a(|\bx|)$, and (2) if $\bx \in A_{\text{no}}$ then $\Pr(\mathcal{Q} ~ \text{accepts} ~\bx) \leq b(|\bx|)$. The class is defined as $\mathsf{BQP}=\mathsf{BQP}(2/3, 1/3)$.

\smallskip
\noindent\textit{$\mathsf{\PPoly}$ complexity class}.  $\mathsf{\PPoly}$, different from $\BQP$, represents the class of decision problems that can be solved by a polynomial-time deterministic algorithm with access to an auxiliary `advice' bitstring \cite{arora2009computational}. This advice bitstring depends only on the input size and must be identical for all inputs of a given size, regardless of the specific input $\bx$. More specifically, a language $ \mathsf{L} : \{0,1\}^* \to \{0,1\} $ is in $\PPoly$ if there exists: (1) a polynomial-time deterministic algorithm $\mathcal{A}$, and (2) a family of `advice' bitstrings $ \{\mathsf{a}_N\}_{N \in \mathbb{N}} $, where $ \mathsf{a}_N \in \{0,1\}^{\text{poly}(N)}$,  such that for every input $\bx \in \{0,1\}^N$: $\mathcal{A}(\bx, \mathsf{a}_N) = \mathsf{L}(\bx)$. Here  $\{0,1\}^*$ denotes the set of all finite-length binary strings.
 
Note that it is widely believed that $\BQP$ is not entirely contained within $\PPoly$. Otherwise, problems like factoring would become solvable by polynomial-size classical circuits.

\medskip
\noindent\textit{Learning procedure}. Ref.~\cite{molteni2024exponential} leveraged the relation between $\BQP$ and $\PPoly$ classes to understand the fundamental limitation of classical learners when learning observables from the measurement data. Denote the concept class as 
\begin{equation}\label{append:eqn:concept-complexity}
		\mathcal{F}=\left\{f_{\bm{\alpha}}(\bx^{\perp})= \Tr\left(\rho_{\mathsf{H}}(\bx^{\perp}) O(\bm{\alpha})\right)\Big|\bm{\alpha} \in \{-1,1\}^q   \right\},
	\end{equation}
 where $\rho_{\mathsf{H}}(\bx^{\perp}) = U_{\mathsf{H}}\ket{\bx^{\perp}}\bra{\bx^{\perp}}U_{\mathsf{H}}^{\dagger}$ with $\bx^{\perp}$ being a bitstring,  $U_{\mathsf{H}}=e^{-\imath \mathsf{H}   \tau }$, and $O(\bm{\alpha})=\sum_{i=1}^q \bm{\alpha}_i P_i$ with $P_i$ being the $i$-th $k$-local Pauli string.	 

The goal of the classical learner is to learn a model $h(\bx^{\perp})$ which approximates the unknown concept $f_{\bm{\alpha}}(\bx^{\perp})$. The model  $h(\bx^{\perp})$ is inferred by the training dataset $\mathcal{T}=\{\bx^{\perp, (i)}, \hat{y}^{(i)}\}_{i=1}^n$ where  $\hat{y}^{(i)}$ refers to the estimated label of $\yi=\Tr(\rho_{\mathsf{H}}(\bx^{\perp, (i)}) O(\bm{\alpha}))$ with an additive error $\epsilon_2$ for $\forall i \in [n]$. The estimation error arises from the finite number of measurements, as the classical learner inputs data $\bx$ into the quantum system and relies on a limited number of measurements to obtain the estimated value. An immediate observation is that the procedure of learning by measure-out data is identical to the learning paradigm investigated in this work. Once the data collection process is completed, the learner utilizes $\mathcal{T}$ to infer the prediction model $h(\bx^{\perp})$ on the classical side. 
	
The efficacy of the classical learner towards the concept class $\mathcal{F}$ is quantified by the following definition. 
\begin{definition}[Efficient learning condition, Definition~4 in Ref.~\cite{molteni2024exponential}]\label{append:def:eff-cond}
A concept class $\mathcal{F}$ in Eq.~(\ref{append:eqn:concept-complexity}) is  efficiently learnable if there exists a polynomial-time algorithm $\mathcal{A}$ with complexity $\text{poly}(1/\epsilon, 1/\delta, 1/\epsilon_2, N)$, such that for any $\epsilon, \epsilon_2 > 0$, any $0 < \delta < 1$, any target function $f_{\alpha} \in \mathcal{F}$, and any input distribution $\mathbb{D}$, the following holds:

If $\mathcal{A}$ receives a training set $\mathcal{T}=\{\bx^{\perp, (i)}, \hat{y}^{(i)}\}_{i=1}^n$ with $n\sim \text{poly}(1/\epsilon, 1/\delta, 1/\epsilon_2, N)$, it produces a hypothesis $h(.) = \mathcal{A}(\mathcal{T}, \epsilon, \delta, .)$ such that, with probability at least $1 - \delta$, the expected error satisfies 
$\mathbb{E}_{\mathbf{\bx^{\perp}} \sim \mathbb{D}} \left[ |f_{\bm{\alpha}}(\bx^{\perp}) - h(\bx^{\perp})|^2 \right] \leq \epsilon.$
\end{definition}

\bigskip
By leveraging the concept of $\BQP$ and $\PPoly$ classes,  Ref.~\cite{molteni2024exponential} shows the hardness of learning $\mathcal{F}$ from the measurement data. 
\begin{lemma}[Adapted from Theorem 2, \cite{molteni2024exponential}]\label{append:lem:complexity-low-bound}
For any $\BQP$-complete language, there exists a Hamiltonian $\mathsf{H}_{\mathsf{hard}}$ to build the concept class 
\begin{equation}
		\mathcal{F}_{\mathsf{Hard}}=\left\{f_{\bm{\alpha}}(\bx^{\perp})= \Tr\left(\rho_{\mathsf{H}_{\mathsf{hard}}}(\bx^{\perp}) O(\bm{\alpha})\right)\Big| \bm{\alpha}\in \{-1,1\}^q   \right\},
	\end{equation}
where $\rho_{\mathsf{H}}(\bx^{\perp}) = U_{\mathsf{H}_{\mathsf{hard}}}\ket{\bx^{\perp}}\bra{\bx^{\perp}}U_{\mathsf{H}_{\mathsf{hard}}}^{\dagger}$ with $U_{\mathsf{H}_{\mathsf{hard}}}=e^{-\imath \mathsf{H}_{\mathsf{hard}} \tau }$ and $O(\bm{\alpha})=\sum_{i=1}^q \bm{\alpha}_i P_i$ with $P_i$ being the $i$-th $k$-local Pauli string. Moreover, no randomized polynomial-time classical algorithm satisfies the learning condition of Def.~\ref{append:def:eff-cond} for this concept class, unless $\BQP \subseteq \PPoly$.

\end{lemma}

The main idea of the proof for Lemma~\ref{append:lem:complexity-low-bound} is as follows. First, the authors correlate $\mathcal{F}_{\mathsf{Hard}}$ with $\mathsf{BQP}$ circuits. In particular, the unitary $U_{\mathsf{H}_{\mathsf{hard}}}$ in $\mathcal{F}_{\mathsf{Hard}}$ is specified to be a family of quantum circuit $\{U_{\mathsf{BQP}^N}\}_N$, which decides the $\mathsf{BQP}$-complete language $\mathsf{L}$, one circuit per size; the observable is set as $O(\bm{\alpha}^*)=Z\otimes \mathbb{I}_2\otimes \cdots \otimes \mathbb{I}_2$. In this way, the output of quantum circuits can correctly decide every $\bx^{\perp}\in \mathsf{L}$, i.e., $f(\bx^{\perp})> 0$ if $\bx \in \mathsf{L}$ and $f(\bx^{\perp})< 0$ if $\bx \notin \mathsf{L}$. Then, the authors show that if $\mathcal{F}_{\mathsf{Hard}}$ is learnable in the sense of Def.~\ref{append:def:eff-cond}, then there must exists a polynomial size classical circuit which evaluates $f_{\bm{\alpha}^*}=\Tr(\rho_{\mathsf{H}_{\mathsf{hard}}}(\bx^{\perp}) O(\bm{\alpha}^*))$ correctly on every $\bx^{\perp}\in \{0, 1\}^N$, implying $\BQP \subseteq \PPoly$.

\smallskip
We next generalize the results of Lemma~\ref{append:lem:complexity-low-bound} to show the computational hardness of learning the bounded-gate circuit with incoherent dynamics. To achieve this goal, it is sufficient to show the concept class $\mathcal{F}$ in Eq.~(3) of the main text can represent the concept class of $\mathsf{BQP}$ circuits discussed above. 

For $\BQP$ circuits, the observable $Z\otimes \mathbb{I}_2\otimes \cdots \otimes \mathbb{I}_2$ meets the requirement of the observable defined in $\mathcal{F}$, which is formed by Pauli operators with a bounded norm. As for the state $U_{\mathsf{BQP}}\ket{\bx^{\perp}}$ (or equivalently $U_{\mathsf{H}_{\mathsf{hard}}}\ket{\bx^{\perp}}$), it can also be expressed by the bounded-gate circuit $\{\RZ+\CI\}$ with the initial state $\ket{0}^{\otimes N}$. Specifically, the bounded-gate circuit is decomposed into two parts, where the first part is used to prepare the state $\ket{\bx^{\perp}}$ and the second part is to prepare the $\mathsf{BQP}$ circuit. For the first part, $N$ $\RZ$ gates and $\mathcal{O}(N)$ $\CI$ gates associated with a proper distribution are sufficient to prepare any input state $\ket{\bx^{\perp}}$ with $\bx^{\perp}\in \{0,1\}^N$. For the second part, since $U_{\mathsf{BQP}}$ contains at most $\mathcal{O}(poly (N))$ quantum gates and $\{\RZ+\CI\}$ is a universal basis gate set, $\mathcal{O}(poly (N))$ $\RZ$ gates with a proper distribution over the classical inputs and $\mathcal{O}(poly (N))$  $\CI$ gates are sufficient to synthesis $U_{\mathsf{BQP}}$. Taken together, the $\mathsf{BQP}$ circuit belongs to $\mathcal{F}$ when the number of $\{\RZ+\CI\}$ gates polynomially scale with $N$. 

The reformulation enables us to generalize the hardness of learning quantum observables from measured data to the broader challenge of learning bounded-gate circuits, which is the focus of our work. That is, according to the results of Lemma~\ref{append:lem:complexity-low-bound}, we can conclude that no algorithm can only use the measure-out data to learn the bounded-gate quantum circuit within a polynomial time unless $\BQP \subseteq \PPoly$. This proves Theorem~\ref{thm:lower-bound-runtime-comp}.

\section{Learnability of the proposed kernel-based ML model (Proof of Theorem 2)}\label{append:sec:proof-thm2}

This section provides the proof of Theorem 2, which analyzes how the prediction error of the proposed kernel-based ML model depends on the number of training examples $n$, the size of the quantum system $N$, and the dimension of classical inputs $d$. Recall that in the main text, the proposed  state prediction model is 
\begin{equation}\label{append:eqn:TriGeo-non-trunc-form}\hatsigma(\bx)=\frac{1}{n}\sum_{i=1}^n\hatkappa\left(\bx, \bxi\right)\tilderho_T(\bxi) ~\text{with}~\hatkappa\left(\bx, \bxi\right) = \sum_{\bomega, \|\bomega\|_0 \leq \Lambda} 2^{\|\bomega\|_0}\Phi_{\bomega}(\bx)\Phi_{\bomega}(\bxi) \in \mathbb{R}.
\end{equation} 
What we intend to prove is the average discrepancy between $\Tr(\hatsigma(\bx) O)$ and the ground truth $\Tr(\rho(\bx)O)$ when $\bx$ is uniformly and randomly sampled from $[-\pi, \pi]^d$ and the local observable $O$ is sampled from a prior distribution $\mathbb{D}_O$, i.e., $\mathbb{E}_{\bx\sim [-\pi, \pi]^d} | \Tr(O\hatsigma(\bx)) - \Tr(O\rho(\bx)) |^2$.

Before moving to proceed with the further analysis, let us exhibit the formal statement of Theorem 2.  
\begin{theorem-non}[Restatement of Theorem 2]\label{lem:Lowesa-no-trunc-prediction-error}
Following notations in the main text, consider a parametrized family of $N$-qubit states $\mathcal{Q}$ and a sum $O=\sum_{i=1} O_i$ of multiple local observables with $\sum_i\|O_i\|_{\infty}\leq B$ and the maximum locality of $\{O_i\}$ being $K$. Suppose $\mathbb{E}_{\bx\sim [-\pi, \pi]^d}\|\nabla_{\bx} \Tr(\rho(\bx)O)\|_2^2\leq C$. Then, let the dataset be $\mathcal{T}_{\mathsf{s}}=\{\bxi \rightarrow \tilderho(\bxi)\}_{i=1}^n$ with $\bxi \sim  \textnormal{Unif}[-\pi, \pi]^d$ and $n = |\mathfrak{C}(\Lambda)|  \frac{2  B^2 9^K}{\epsilon}  \log \left({2 \cdot |\mathfrak{C}(\Lambda)|}/{\delta}\right)$ with $\mathfrak{C}(\Lambda) =\{\bomega|\bomega \in \{0, \pm 1\}^d, ~s.t.~\|\bomega\|_0\leq \Lambda\}$. When the frequency is truncated to $  \Lambda=4C/\epsilon$, the state prediction model in Eq.~(\ref{append:eqn:TriGeo-non-trunc-form}) achieves 
\begin{equation}
	\mathbb{E}_{\mathcal{T}_{\mathsf{s}}}[\hatsigma(\bx)] = \rho_{\Lambda}(\bx)
\end{equation}
and with  probability at least $1-\delta$, 
\begin{equation}
	\mathbb{E}_{\bx\sim [-\pi, \pi]^d} \left| \Tr(O\hatsigma(\bx)) - \Tr(O\rho(\bx))  \right|^2 \leq \epsilon.
\end{equation}	
\end{theorem-non}

To reach Theorem 2, we first use the triangle inequality to decouple the difference between the prediction and ground truth into the truncation error and the estimation error, i.e.,
\begin{subequations}
\begin{eqnarray}\label{append:eqn:decouple-dis-model-truth}
	&& \mathbb{E}_{\bx\sim [-\pi ,\pi]^d} \left[ \left| \Tr \left(O \hatsigma(\bx)\right) - \Tr \left(O\rho(\bx)\right)  \right|^2 \right] \\
	\leq && \Bigg(\sqrt{\mathbb{E}_{\bx\sim [-\pi ,\pi]^d} \left[ \left| \Tr \left(O \rho_{\Lambda}(\bx) \right) - \Tr(O\rho(\bx))  \right|^2 \right]} +
	\sqrt{\mathbb{E}_{\bx\sim [-\pi ,\pi]^d} \left[ \left| \Tr \left(O \hatsigma(\bx)\right) - \Tr(O\rho_{\Lambda}(\bx))  \right|^2 \right]}
	  \Bigg)^2.
\end{eqnarray}
\end{subequations} 
After decoupling, we then separately derive the upper bound of these two terms, where the relevant results are encapsulated in the following two lemmas whose proofs are given in the subsequent two subsections.

\begin{lemma}[Truncation error of $\rho_{\Lambda}$]\label{lem:truncation-error-geo-kernel} 
 Following notations in Theorem 2, assuming $\mathbb{E}_{\bx\sim [-\pi, \pi]^d}\|\nabla_{\bx} \Tr(\rho(\bx)O)\|_2^2\leq C$, the truncation error induced by removing high-frequency terms of $\rho$ under the trigonometric  expansion with $\|\bomega\|_0\leq \Lambda$ is upper bounded by
 \begin{equation}\label{append:eqn:trigo-expan-traget-state}
 	\mathbb{E}_{\bx\sim [-\pi, \pi]^d} \left|\Tr(O\rho_{\Lambda}(\bx)- \Tr(O\rho(\bx)) \right|^2 \leq \frac{C}{\Lambda}.
 \end{equation}
\end{lemma}

\begin{lemma}[Estimation error of $\hatsigma$]\label{lem:estimation-error-geo-kernel} 
 Following notations in Theorem 2, with probability at least $1-\delta$, the estimation error induced by finite training examples $\mathcal{T}=\{\tilderho(\bxi)\}_{i=1}^n$ is upper bounded by
 \begin{equation}\label{append:eqn:trigo-estimation-error}
 	\mathbb{E}_{\bx\sim [-\pi, \pi]^d} \left[\left |\Tr(O  \hatsigma(\bx)) - \Tr(O \rho_{\Lambda}(\bx))\right|^2 \right] \leq    |\mathfrak{C}(\Lambda)|  \frac{1}{2n} B^2 9^K  \log \left(\frac{2 \cdot |\mathfrak{C}(\Lambda)| }{\delta}\right),
 \end{equation}
 where $\mathfrak{C}(\Lambda) =\{\bomega|\bomega \in \{0, \pm 1\}^d, ~s.t.~\|\bomega\|_0\leq \Lambda\}$ refers to the set of truncated frequencies.
\end{lemma}

We are now ready to present the proof of Theorem 2.
\begin{proof}[Proof of Theorem 2]
The difference between the prediction and ground truth can be obtained by integrating Lemmas~\ref{lem:truncation-error-geo-kernel} and \ref{lem:estimation-error-geo-kernel} into Eq.~(\ref{append:eqn:decouple-dis-model-truth}). Mathematically, with probability at least $1-\delta$, we have
\begin{eqnarray}
	&& \mathbb{E}_{\bx\sim [-\pi ,\pi]^d} \left[ \left| \Tr(O \hatsigma(\bx)) - \Tr(O\rho(\bx))  \right|^2 \right] \leq  \left(\sqrt{\frac{C}{\Lambda}} +  \sqrt{|\mathfrak{C}(\Lambda)|  \frac{1}{2n} B^2 9^K  \log \left(\frac{2 \cdot |\mathfrak{C}(\Lambda)|}{\delta}\right)} \right)^2.
\end{eqnarray}
To ensure the average prediction error is upper bounded by $\epsilon$, it is sufficient to showcase when the inner two terms are upper bounded by $\sqrt{\epsilon}/2$. For the first term, the condition is satisfied when
\begin{equation}
	\sqrt{\frac{C}{\Lambda}} \leq \frac{\sqrt{\epsilon}}{2} \Leftrightarrow \Lambda \geq \frac{4C}{\epsilon}. 
\end{equation}
For the second term, we have 
\begin{eqnarray}
&&	\sqrt{|\mathfrak{C}(\Lambda)|  \frac{1}{2n} B^2 9^K  \log \left(\frac{2 \cdot |\mathfrak{C}(\Lambda)|}{\delta}\right)} \leq \frac{\sqrt{\epsilon}}{2}   
\Leftrightarrow  n \geq  |\mathfrak{C}(\Lambda)|  \frac{2  B^2 9^K}{\epsilon}  \log \left(\frac{2 \cdot |\mathfrak{C}(\Lambda)|}{\delta}\right). 
\end{eqnarray}  

Taken together, with probability $1-\delta$, the prediction error is upper bounded by $\epsilon$ when the number of training examples satisfies
\begin{equation}
	n \geq |\mathfrak{C}(4C/\epsilon)|  \frac{2  B^2 9^K}{\epsilon}  \log \left(\frac{2 \cdot |\mathfrak{C}(4C/\epsilon)|}{\delta}\right).
\end{equation}

\end{proof}

\subsection{Truncation error of the classical learning model (Proof of Lemma \ref{lem:truncation-error-geo-kernel})}
Recall that under the trigonometric monomial expansion, the target state $\rho$ with truncation and without truncation takes the form as
\begin{equation}\label{append:eqn:trigo-expan-traget-truncated-state}
	\rho_{\Lambda}=  \sum_{\bomega \in \mathfrak{C}(\Lambda)}\Phi_{\bomega}(\bx)   \rho_{\bomega}~\textnormal{and}~
	\rho=  \sum_{\bomega \in \mathfrak{C}(d)}\Phi_{\bomega}(\bx)   \rho_{\bomega}, 
\end{equation}
respectively. 
The purpose of Lemma \ref{lem:truncation-error-geo-kernel} is to analyze the upper bound of the discrepancy between $\Tr(\rho_{\Lambda}O)$ and $\Tr(\rho O)$ induced by the truncation of high-frequency terms.

\begin{proof}[Proof of Lemma \ref{lem:truncation-error-geo-kernel}]
By adopting the explicit trigonometric monomial  expansion of $\rho$ and $\rho_{\Lambda}$  in Eq.~(\ref{append:eqn:trigo-expan-traget-truncated-state}), we have
\allowdisplaybreaks
\begin{subequations}
	\begin{eqnarray}
	&& \mathbb{E}_{\bx\sim [-\pi, \pi]^d} \left|\Tr(O\rho_{\Lambda}(\bx)- \Tr(O\rho(\bx)) \right|^2\\
	=	&&	\mathbb{E}_{\bx\sim [-\pi, \pi]^d} \left|   \sum_{\bomega, \|\bomega\|_0 >  \Lambda} \Phi_{\bomega}(\bx) \llangle \rho_{\bomega} |O \rrangle \right|^2 \\
		= && \mathbb{E}_{\bx\sim [-\pi, \pi]^d}   \sum_{\bomega, \|\bomega\|_0 >  \Lambda}  \sum_{\bomega', \|\bomega'\|_0 >  \Lambda}    \Phi_{\bomega}(\bx)  \Phi_{\bomega'}(\bx)\llangle \rho_{\bomega} |O \rrangle  \llangle \rho_{\bomega'} |O \rrangle   \\ 
		= &&  \sum_{\bomega, \|\bomega\|_0 >  \Lambda}  \sum_{\bomega', \|\bomega'\|_0 >  \Lambda}  \frac{1}{(2\pi)^d} \int_{[-\pi, \pi]^{d}}  \Phi_{\bomega}(\bx)  \Phi_{\bomega'}(\bx)\llangle \rho_{\bomega} |O \rrangle  \llangle \rho_{\bomega'} |O \rrangle  \mathsf{d}^d x \\
		= &&   \sum_{\bomega, \|\bomega\|_0 >  \Lambda}  \sum_{\bomega', \|\bomega'\|_0 >  \Lambda} 2^{-\|\bomega\|_0}\delta_{\bomega,\bomega'} \llangle \rho_{\bomega} |O \rrangle  \llangle \rho_{\bomega'} |O \rrangle  \label{append:eqn:trunc-geo-est-error-1}  \\
		= &&  \sum_{\bomega, \|\bomega\|_0 >  \Lambda}  2^{-\|\bomega\|_0} |\llangle \rho_{\bomega} |O \rrangle|^2 :=   \sum_{\bomega, \|\bomega\|_0 >  \Lambda}  2^{-\|\bomega\|_0} \alpha_{\bomega}^2,   \label{append:eqn:trunc-geo-est-error-2} 
	\end{eqnarray}
\end{subequations}
where the first three equalities follow a direct reformulation, Eq.~(\ref{append:eqn:trunc-geo-est-error-1}) employs the orthogonality of basis functions in the trigonometric  expansion, i.e., 
\begin{equation}\label{append:eqn:orthogonal-tri-expansion}
	\frac{1}{(2\pi)^d}\int_{[-\pi, \pi]^{d}}   \Phi_{\bomega}(\bx) \Phi_{\bomega'}(\bx) \mathsf{d}^d x =2^{-\|\bomega\|_0}\delta_{\bomega,\bomega'},
\end{equation}
and in the last equality we define $\alpha_{\bomega} \equiv\Tr(\rho_{\bomega}O)$ for clarification.

The remainder of the proof uses the assumption of the norm of the gradients of the expectation, i.e., $\|\nabla_{\bx} \Tr(\rho(\bx)O)\|_2^2\leq C$, to derive the upper bound of Eq.~(\ref{append:eqn:trunc-geo-est-error-2}).  To do so, we first derive the explicit form of the gradient under the trigonometric monomial expansion. That is, the gradient of $\Tr(\rho(\bx)O)$ with respect to $\bx$ is a $d$-dimensional vector, i.e., \begin{subequations}
\begin{eqnarray}
	&& \nabla_{\bx} \Tr(\rho(\bx)O) \\
	= && \nabla_{\bx} \sum_{\bomega} \Phi_{\bomega}(\bx)\Tr(\rho_{\bomega}O) \nonumber\\
	= && \sum_{\bomega} \nabla_{\bx} \Phi_{\bomega}(\bx) \alpha_{\bomega}   \\
	= && \Big[\sum_{\bomega} \Psi_{\bomega_1}(\bx_1)\Phi_{\bomega_{2:d}}(\bx_{2:d}) \alpha_{\bomega}, \cdots,  \sum_{\bomega} \Phi_{\bomega_{1:i-1}}(\bx_{1:i-1})\Psi_{\bomega_i}(\bx_i)\Phi_{\bomega_{i+1: d}}(\bx_{i+1:d}) \alpha_{\bomega},\\
	 && \cdots, \sum_{\bomega} \Psi_{\bomega_{1:d-1}}(\bx_{1:d-1}) \Psi_{\bomega_d}(\bx_d) \alpha_{\bomega} \Big]^{\top} \in \mathbb{R}^d, 
\end{eqnarray}
\end{subequations} 
where $\Phi_{\bomega_{a:b}}(\bx_{a:b}) \equiv \prod_{i=a}^b \Phi_{\bomega_i}(\bx_i)$, and the derivative with respect to the $i$-th entry is 
\begin{equation}
	\Psi_{\bomega_i}(\bx_i) \equiv \nabla_{\bx_i}\Phi_{\bomega_{i}}(\bx_{i}) = \begin{cases}
		 0 ~ & \textnormal{if}~ \bomega_i = 0 \\
		- \sin(\bx_i) & \textnormal{if}~\bomega_i = 1 \\
		 \cos(\bx_i) &  \textnormal{if}~ \bomega_i = -1
	\end{cases}.
\end{equation}
By making use of the above formula, the expectation value of $\|\nabla_{\bx} \Tr(\rho(\bx)O) \|_2^2 $ over $\bx\sim \textnormal{Unif}[-\pi, \pi]^d$ satisfies 
\begin{subequations}
\begin{eqnarray}
	&& \mathbb{E}_{\bx \sim [-\pi, \pi ]^d}\|\nabla_{\bx} \Tr(\rho(\bx)O) \|_2^2 \\
	= && \frac{1}{(2\pi)^d} \int_{[-\pi, \pi]^d} \Biggl[\left(\sum_{\bomega} \Psi_{\bomega_1}(\bx_1)\Phi_{\bomega_{2:d}}(\bx_{2:d}) \alpha_{\bomega} \sum_{\bomega'} \Psi_{\bomega'_1}(\bx_1)\Phi_{\bomega'_{2:d}}(\bx_{2:d}) \alpha_{\bomega'}  \right) + \cdots \\
	&& + \left(\sum_{\bomega} \Phi_{\bomega_{1:i-1}}(\bx_{1:i-1})\Psi_{\bomega_i}(\bx_{i})\Phi_{\bomega_{i+1:d}}(\bx_{i+1:d}) \alpha_{\bomega} \sum_{\bomega'} \Phi_{\bomega'_{1:i-1}}(\bx_{1:i-1})\Psi_{\bomega'_{i}}(\bx_i)\Phi_{\bomega'_{i+1:d}}(\bx_{i+1:d}) \alpha_{\bomega'}  \right) +  \cdots \\
	&&  + \left(\sum_{\bomega} \Phi_{\bomega_{1:d-1}}(\bx_{1:d-1}) \Psi_{\bomega_d}(\bx_d) \alpha_{\bomega} \sum_{\bomega'} \Phi_{\bomega'_{1:d-1}}(\bx_{1:d-1}) \Psi_{\bomega'_d}(\bx_d) \alpha_{\bomega'} \right)   \Biggl] \mathsf{d}^d x \\
	= && \frac{1}{(2\pi)^d} \Biggl[ \pi \int_{[-\pi, \pi]^{d-1}}  \left(\sum_{\bomega_{2:d}}  \Phi_{\bomega_{2:d}}(\bx_{2:d}) \alpha_{1,\bomega_{2:d}} \sum_{\bomega'_{2:d}}  (\bx_1)\Phi_{\bomega'_{2:d}}(\bx_{2:d}) \alpha_{1,\bomega'_{2:d}}  \right)\mathsf{d}^{d-1} x +  \cdots \label{append:eqn:grad-expect-0-0} \\
	&&  + \pi \int_{[-\pi, \pi]^{d-1}}  \left(\sum_{\bomega_{2:d}}  \Phi_{\bomega_{2:d}}(\bx_{2:d}) \alpha_{-1,\bomega_{2:d}} \sum_{\bomega'_{2:d}}  (\bx_1)\Phi_{\bomega'_{2:d}}(\bx_{2:d}) \alpha_{-1,\bomega'_{2:d}}  \right)\mathsf{d}^{d-1} x  + \cdots\\
	  &&   + \pi \int_{[-\pi, \pi]^{d-1}} \left(\sum_{\bomega_{1:d-1}} \Phi_{\bomega_{1:d-1}}(\bx_{1:d-1})   \alpha_{\bomega_{1:d-1},1} \sum_{\bomega'_{1:d-1}} \Phi_{\bomega'_{1:d-1}}(\bx_{1:d-1})   \alpha_{\bomega'_{1:d-1},1} \right) \mathsf{d}^{d-1} x \\
	  && + \pi \int_{[-\pi, \pi]^{d-1}} \left(\sum_{\bomega_{1:d-1}} \Phi_{\bomega_{1:d-1}}(\bx_{1:d-1})   \alpha_{\bomega_{1:d-1},-1} \sum_{\bomega'_{1:d-1}} \Phi_{\bomega'_{1:d-1}}(\bx_{1:d-1})   \alpha_{\bomega'_{1:d-1},-1} \right) \mathsf{d}^{d-1} x\Biggr] \label{append:eqn:grad-expect-0}\\
	= && \Biggl[\sum_{\bomega_{2:d}} 2^{-\|[1,\bomega_{2:d}]\|_0}  |\alpha_{1,\bomega_{2:d}}|^2   + \sum_{\bomega_{2:d}} 2^{-\|[-1,\bomega_{2:d}]\|_0}  |\alpha_{-1,\bomega_{2:d}}|^2  + \cdots  \\
 	&&  + \sum_{\bomega_{1:d-1}} 2^{-\|[\bomega_{1:d-1},1]\|_0}  |\alpha_{\bomega_{1:d-1,1}}|^2  + \sum_{\bomega_{1:d-1}} 2^{-\|[\bomega_{1:d-1},-1]\|_0}  |\alpha_{\bomega_{1:d-1,-1}}|^2  \Biggl] \label{append:eqn:grad-expect-1-1} := \bigstar, 
\end{eqnarray}
\end{subequations}
where Eqs.~(\ref{append:eqn:grad-expect-0-0})-(\ref{append:eqn:grad-expect-0}) are obtained by taking the expectation value on the gradient term $\Psi_{\bomega_i}$ for $\forall i\in [d]$, the quantity $\alpha_{\pm 1,\bomega_{2:d}}$ refers to $\Tr(\rho_{\pm 1,\bomega_{2:d}}O)$ with the frequency value at the $1$-st position being $\pm 1$ (the same rule applies to other terms such as $\alpha_{\bomega_{1:i}, i+1, \bomega_{i+2:d}}$), and Eq.~(\ref{append:eqn:grad-expect-1-1}) is obtained by taking integral over each term in the bracket and the orthogonality of basis functions in $\Phi_{\bomega}(\bx)$ and  $\Phi_{\bomega'}(\bx)$ as shown in Eq.~(\ref{append:eqn:orthogonal-tri-expansion}).

We now use the reformulated  $\mathbb{E}_{\bx \sim [-\pi, \pi ]}\|\nabla_{\bx} \Tr(\rho(\bx)O) \|_2^2$ to derive the upper bound of  the truncation error. Namely,   Eq.~(\ref{append:eqn:trunc-geo-est-error-2}) can be reformulated as 
\allowdisplaybreaks
\begin{subequations}
\begin{eqnarray}
	&& \mathbb{E}_{\bx\sim [-\pi, \pi]^d} \left|\Tr(O\rho_{\Lambda}(\bx)- \Tr(O\rho(\bx)) \right|^2  \\ 
= &&	\sum_{\bomega, \|\bomega\|_0 >  \Lambda}  2^{-\|\bomega\|_0} |\alpha_{\bomega}|^2 \\
= &&   \frac{1}{\Lambda} \left(\Lambda \sum_{\bomega, \|\bomega\|_0 >   \Lambda}      2^{-\|\bomega\|_0} | \alpha_{\bomega}|^2\right) \\
\leq && \frac{\mathbb{E}_{\bx \sim [-\pi, \pi ]}\|\nabla_{\bx} \Tr(\rho(\bx)O) \|_2^2 }{\Lambda}, \label{append:eqn:grad-expect-upper-1} \\ 
\leq && \frac{C}{\Lambda},
\end{eqnarray}
\end{subequations}
where the last second inequality is supported by the result indicated below and the last inequality is supported by the assumption $\mathbb{E}_{\bx \sim [-\pi, \pi ]}\|\nabla_{\bx} \Tr(\rho(\bx)O) \|_2^2\leq C$. 

\smallskip
\noindent\textit{\underline{The proof of Eq.~(\ref{append:eqn:grad-expect-upper-1})}}. According to Eq.~(\ref{append:eqn:grad-expect-1-1}), to reach Eq.~(\ref{append:eqn:grad-expect-upper-1}), it is equivalent to proving that for $\forall \Lambda\in [d]$, 
\begin{equation}
	\Lambda \sum_{\bomega, \|\bomega\|_0 >   \Lambda}      2^{-\|\bomega\|_0} | \alpha_{\bomega}|^2 \leq  \bigstar. \label{append:eqn:grad-expect-upper-2}  
\end{equation}

\noindent Note that depending on the number of non-zero entries of the frequency (i.e., $\|\bomega\|_0$),  the left and right-hand sides in Eq.~(\ref{append:eqn:grad-expect-upper-2}) can be decomposed into $d - \Lambda$ parts and $d-1$ parts, respectively. Mathematically, the left hand side of Eq.~(\ref{append:eqn:grad-expect-upper-2}) yields 
\begin{equation}
	\Lambda \left( \sum_{\bomega, \|\bomega\|_0 =   \Lambda+1}      2^{-\|\bomega\|_0} | \alpha_{\bomega}|^2 + \sum_{\bomega, \|\bomega\|_0 =   \Lambda+2}      2^{-\|\bomega\|_0} | \alpha_{\bomega}|^2 + \cdots + \sum_{\bomega, \|\bomega\|_0 =   d}      2^{-\|\bomega\|_0} | \alpha_{\bomega}|^2 \right). \label{append:eqn:grad-expect-upper-3}
\end{equation}
Besides, the right-hand side of Eq.~(\ref{append:eqn:grad-expect-upper-2}), i.e., $\bigstar$, can be rewritten as
\begin{eqnarray}
	\bigstar  = && \Biggl[\bigg(\sum_{\bomega_{2:d}, \|\bomega_{2:d}\|_0=0} 2^{-\|[1,\bomega_{2:d}]\|_0}  |\alpha_{1,\bomega_{2:d}}|^2    + \cdots + \sum_{\bomega_{2:d}, \|\bomega_{2:d}\|=d-1} 2^{-\|[1,\bomega_{2:d}]\|_0}  |\alpha_{1,\bomega_{2:d}}|^2  \bigg) \\
	&& + \cdots \\
	&&  + \bigg( \sum_{\bomega_{1:d-1},\|\bomega_{1:d-1}\|_0=1} 2^{-\|[\bomega_{1:d-1},-1]\|_0}  |\alpha_{\bomega_{1:d-1,-1}}|^2  + \cdots + \sum_{\bomega_{1:d-1},\|\bomega_{1:d-1}\|_0=d-1} 2^{-\|[\bomega_{1:d-1},-1]\|_0}  |\alpha_{\bomega_{1:d-1,-1}}|^2 \bigg) \Biggl] \\
	= && \bigg(\sum_{\bomega_{2:d}, \|\bomega_{2:d}\|_0=0} 2^{-\|[1,\bomega_{2:d}]\|_0}  |\alpha_{1,\bomega_{2:d}}|^2 + \cdots \sum_{\bomega_{1:d-1},\|\bomega_{1:d-1}\|_0=0} 2^{-\|[\bomega_{1:d-1},-1]\|_0}  |\alpha_{\bomega_{1:d-1,-1}}|^2 \bigg) \\
	&& + \cdots \\
	&& + \left(\sum_{\bomega_{2:d}, \|\bomega_{2:d}\|_0=d-1} 2^{-\|[1,\bomega_{2:d}]\|_0}  |\alpha_{1,\bomega_{2:d}}|^2 + \cdots + \sum_{\bomega_{1:d-1},\|\bomega_{1:d-1}\|_0=d-1} 2^{-\|[\bomega_{1:d-1},-1]\|_0}  |\alpha_{\bomega_{1:d-1,-1}}|^2 \right),  \label{append:eqn:grad-expect-1}
\end{eqnarray}
where the second equality is acquired by rearranging, as in every bracket, all frequencies have the same number of non-zero entries.

In conjunction with Eqs.~(\ref{append:eqn:grad-expect-upper-3}) and (\ref{append:eqn:grad-expect-1}), to achieve Eq.~(\ref{append:eqn:grad-expect-upper-2}), we need to demonstrate that for any $\Lambda <\Lambda'<d$,  the following relation is satisfied, i.e.,
\begin{equation}
	\Lambda \sum_{\bomega, \|\bomega\|_0 =   \Lambda'}      2^{-\|\bomega\|_0} | \alpha_{\bomega}|^2 \leq   \sum_{\bomega_{2:d}, \|\bomega_{2:d}\|_0=\Lambda'-1} 2^{-\|[1,\bomega_{2:d}]\|_0}  |\alpha_{1,\bomega_{2:d}}|^2 + \cdots + \sum_{\bomega_{1:d-1},\|\bomega_{1:d-1}\|_0=\Lambda'-1} 2^{-\|[\bomega_{1:d-1},-1]\|_0}  |\alpha_{\bomega_{1:d-1,-1}}|^2. \label{append:eqn:grad-expect-3}
\end{equation}

To achieve this goal, we next prove that the right-hand side in Eq.~(\ref{append:eqn:grad-expect-3}) equals to
\begin{equation}
	 \Lambda' \sum_{\bomega, \|\bomega\|_0 =   \Lambda'}      2^{-\|\bomega\|_0} | \alpha_{\bomega}|^2, \quad \text{with}~\Lambda < \Lambda' < d. \label{append:eqn:grad-expect-2}
\end{equation}
Recall that in Eq.~(\ref{append:eqn:grad-expect-2}), each frequency $\bomega\in  \{0, \pm 1\}^{d}$ with $\|\bomega\|_0 =   \Lambda' $ appears $\Lambda'$ times and the total number of frequencies is 
\begin{equation}
	\Lambda\cdot \left|\left\{\bomega\in \{0, \pm 1\}^{d}\Big|\|\bomega\|_0=\Lambda' \right \}\right| = \Lambda \cdot \left|\binom{d}{d-\Lambda'}\right| \cdot 2^{\Lambda'}=\Lambda'\frac{d!}{(d-\Lambda')!\Lambda'!} 2^{\Lambda'}=\frac{d!}{(d-\Lambda')!(\Lambda'-1)!} 2^{\Lambda'}.
\end{equation}
In addition, the total number of frequencies in the right-hand side of Eq.~(\ref{append:eqn:grad-expect-3}) is
\begin{equation}
	 \binom{d-1}{d-\Lambda'-1}2^{\Lambda'} + \cdots +  \binom{d-1}{d-\Lambda'-1}2^{\Lambda'} = d\cdot  \binom{d-1}{d-\Lambda'-1}2^{\Lambda'}=d\frac{(d-1)!}{(d-\Lambda')(\Lambda'-1)}2^{\Lambda'}= \frac{d!}{(d-\Lambda')(\Lambda'-1)}2^{\Lambda'}.
\end{equation}

\begin{figure}[h!]
\includegraphics[width=0.95\textwidth]{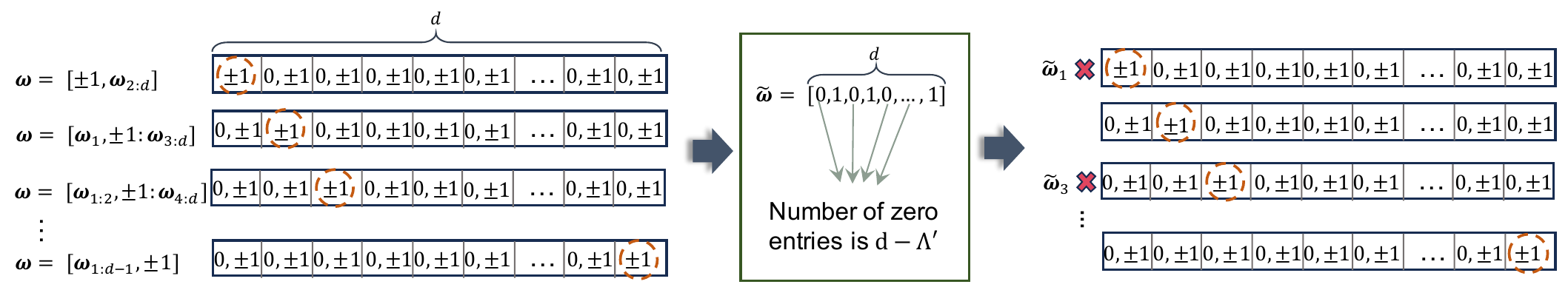}
\caption{\justifying\small{\textbf{A visual interpretation about the  appearance times of $\tilde{\bomega}$ in the right-hand side of Eq.~(\ref{append:eqn:grad-expect-3})}.}}
\label{append:fig:grad-freq}
\end{figure}
Combining with the above two equations, we know that the total number of frequencies in Eq.~(\ref{append:eqn:grad-expect-2}) and the right-hand side of Eq.~(\ref{append:eqn:grad-expect-3}) is the same. Accordingly, their equivalence can be reached if we can show that each frequency $\bomega\in  \{0, \pm 1\}^{d}$ with $\|\bomega\|_0 =   \Lambda' $ appears $\Lambda'$ times in the right-hand side of Eq.~(\ref{append:eqn:grad-expect-3}). This is indeed the case, as the visual interpretation is shown in Fig.~\ref{append:fig:grad-freq}. Denote that the specified frequency as $\tilde{\bomega}$ with $\|\tilde{\bomega}\|_0=\Lambda'$. In other words, there are $d-\Lambda'$  entries whose value is $0$. This property allows us to determine how many times $\tilde{\bomega}$ appears on the right side of Eq.~(\ref{append:eqn:grad-expect-3}). In particular, as shown in the left panel of Fig.~\ref{append:fig:grad-freq}, the right-hand side of Eq.~(\ref{append:eqn:grad-expect-3}) can be divided into $d$ groups, depending on the location of $\pm 1$. Note that among these $d$ groups, every entry with the value $0$ in $\tilde{\bomega}$ precludes one group, and only $d-(d-\Lambda')=\Lambda'$ feasible groups are preserved, as shown in the right panel of Fig.~\ref{append:fig:grad-freq}. Moreover, for each feasible group, $\tilde{\bomega}$ can only appear once, because of the orthogonality of different frequencies. Consequently, we obtain that each frequency $\tilde{\bomega}\in  \{0, \pm 1\}^{d}$ with $\|\tilde{\bomega}\|_0 =   \Lambda' $ appears $\Lambda'$ times, indicating that the right-hand side in Eq.~(\ref{append:eqn:grad-expect-3}) equals to Eq.~ (\ref{append:eqn:grad-expect-2}). 

In other words, we achieve
\begin{eqnarray}
&&	\bigstar = \sum_{\Lambda'=1}^d \Lambda' \sum_{\bomega, \|\bomega\|_0 =   \Lambda'}      2^{-\|\bomega\|_0} | \alpha_{\bomega}|^2 \\
 \Rightarrow && \bigstar =   \sum_{\bomega, \|\bomega\|_0 >   \Lambda}  \|\bomega\|_0     2^{-\|\bomega\|_0} | \alpha_{\bomega}|^2 + \sum_{\bomega, \Lambda \geq  \|\bomega\|_0 \geq   1}  \|\bomega\|_0     2^{-\|\bomega\|_0} | \alpha_{\bomega}|^2 \\
  \Rightarrow && \bigstar \geq \sum_{\bomega, \|\bomega\|_0 >   \Lambda}  \|\bomega\|_0     2^{-\|\bomega\|_0} | \alpha_{\bomega}|^2  \\
  \Rightarrow &&  \bigstar \geq \Lambda \sum_{\bomega, \|\bomega\|_0 >   \Lambda}     2^{-\|\bomega\|_0} | \alpha_{\bomega}|^2.
\end{eqnarray}

Taken together, the truncation error can be upper bounded by the averaged gradient norm of the specified circuit. 
\end{proof}

\subsection{Estimation error of the classical learning model (Proof of Lemma~\ref{lem:estimation-error-geo-kernel})}\label{append:subsec:proof-lemma2} 

The core of the proof is to show that the state prediction model $\hatsigma(\bx)$ in Eq.~(\ref{append:eqn:TriGeo-non-trunc-form}) is equal to the trigonometric expansion of the truncated target quantum state $\rho_{\Lambda}(\bx)$ if we take the expectation over the training data, which includes the randomness from the sampled inputs $\bxi$ and classical shadow. With this relation in mind, we can quantify the estimation error of the proposed state prediction model $\hatsigma(\bx)$ when calculating the expectation value of an unseen state under the specified observable $O$,  i.e.,  $\mathbb{E}_{\bx\sim [-\pi ,\pi]^d} [  | \Tr(O \hatsigma(\bx)) - \Tr(O\rho_{\Lambda}(\bx)) |^2 ].$ 
 
\begin{proof}[Proof of Lemma~\ref{lem:estimation-error-geo-kernel}]
	
We first prove the equivalence between the expectation of the classical representations and the target quantum state, i.e., $\mathbb{E}_{\mathcal{T}}[\hatsigma(\bx)] = \rho_{\Lambda}(\bx)$. Following the explicit form of $\hatsigma(\bx)$ in Eq.~(\ref{append:eqn:TriGeo-non-trunc-form}), we obtain  
\allowdisplaybreaks
\begin{subequations}
	\begin{eqnarray}\label{eqn:exp-model-recover-true-state-trigeo}
	\mathbb{E}_{\mathcal{T}}[\hatsigma(\bx)] = &&  \frac{1}{n}\sum_{i=1}^n \mathbb{E}_{\bxi \sim [-\pi, \pi]^d} \left[\kappa_{\Lambda}\left(\bx, \bxi\right) \right]\mathbb{E}_{s_1^{(\bxi)},...,s_N^{(\bxi)}}\tilderho_1(\bxi) \label{eqn:exp-model-recover-true-state-trigeo-1} \\
	= &&  \mathbb{E}_{\bx^{(1)}\sim [-\pi, \pi]^{d}} \left[\kappa_{\Lambda}\left(\bx, \bx^{(1)}\right) \right]\rho(\bx^{(1)}) \label{eqn:exp-model-recover-true-state-trigeo-2} \\
	= &&  \mathbb{E}_{\bx^{(1)}\sim [-\pi, \pi]^{d}}  \sum_{\bomega, \|\bomega\|_0\leq \Lambda} 2^{\|\bomega\|_0} \Phi_{\bomega}(\bx)\Phi_{\bomega}(\bx^{(1)}) \rho(\bx^{(1)})  \label{eqn:exp-model-recover-true-state-trigeo-3} \\
	= && \sum_{\bomega, \|\bomega\|_0\leq \Lambda} \Phi_{\omega}(\bx) 2^{\|\bomega\|_0} \mathbb{E}_{\bx^{(1)} \sim [-\pi, \pi]^{d}}   \Phi_{\bomega}(\bx^{(1)}) \rho(\bx^{(1)})  \\
    = && \sum_{\bomega, \|\bomega\|_0\leq \Lambda} \Phi_{\bomega}(\bx) 2^{\|\bomega\|_0} \frac{1}{(2\pi)^d} \int_{[-\pi, \pi]^{d}}   \Phi_{\bomega}(\bx^{(1)}) \rho(\bx^{(1)}) \mathsf{d}^d x^{(1)} \\
    = && \sum_{\bomega, \|\bomega\|_0\leq \Lambda} \Phi_{\bomega}(\bx) 2^{\|\bomega\|_0} \frac{1}{(2\pi)^d} \int_{[-\pi, \pi]^{d}}   \Phi_{\bomega}(\bx^{(1)}) \sum_{\bomega'}\Phi_{\bomega'}(\bx) \rho_{\bomega} \mathsf{d}^d x^{(1)}  \label{eqn:exp-model-recover-true-state-trigeo-6} \\
    = && \sum_{\bomega, \|\bomega\|_0\leq \Lambda} \Phi_{\bomega}(\bx)  2^{\|\bomega\|_0} \sum_{\bomega'} \frac{1}{(2\pi)^d} \int_{[-\pi, \pi]^{d}}   \Phi_{\bomega}(\bx^{(1)}) \Phi_{\bomega'}(\bx^{(1)}) \rho_{\bomega} \mathsf{d}^m x^{(1)}  \\
    =   && \sum_{\bomega, \|\bomega\|_0\leq \Lambda} \Phi_{\bomega}(\bx) \rho_{\bomega} \label{append:eqn:exp-model-recover-true-state-trigeo-8}  \\
   = && \rho_{\Lambda}(\bx),
\end{eqnarray}
\end{subequations}
where $s_j^{(\bxi)}$ denotes the randomized measurement outcome for the $j$-th qubit for the state $\rho(\bxi)$ with $s_j^{(\bxi)}\in \{\ket{0}, \ket{1}, \ket{\pm}, \ket{\pm \imath}\}$,  Eq.~(\ref{eqn:exp-model-recover-true-state-trigeo-2}) uses the fact that each $\bxi$ is sampled independently and uniformly from $[-\pi, \pi]^d$, Eq.~(\ref{eqn:exp-model-recover-true-state-trigeo-3}) employs the explicit formula of the truncated trigonometric monomial kernel $\kappa_{\Lambda}$, Eq.~(\ref{eqn:exp-model-recover-true-state-trigeo-6}) adopts trigonometric  expansion of the quantum state, i.e., $\rho(\bxi)=\sum_{\bomega}\Phi_{\bomega}(\bxi)\rho_{\bomega}$,   and Eq.~(\ref{append:eqn:exp-model-recover-true-state-trigeo-8}) comes from the orthogonality of basis functions in Eq.~(\ref{append:eqn:orthogonal-tri-expansion}).

To complete the proof, we next move to analyze the estimation error  $\mathbb{E}_{\bx\sim [-\pi ,\pi]^d} [ | \Tr(O \hatsigma(\bx)) - \Tr(O\rho_{\Lambda}(\bx)) |^2 ]$.   Define 
\begin{equation}\label{append:eqn:def-A-omega}
\tilde{A}_{\bomega}= \frac{1}{n}\sum_{i=1}^n 2^{\|\bomega\|_0}  \Phi_{\bomega}(\bxi)\Tr(\tilderho_1(\bxi)O) -   \Tr(\rho_{\bomega} O).	
\end{equation}
By making use of the trigonometric  expansion of $\rho_{\Lambda}$ and the explicit formalism of $\hatsigma(\bx)$, the estimation error can be reformulated as 
\allowdisplaybreaks
\begin{subequations}
\begin{eqnarray}
	&& \mathbb{E}_{\bx\sim [-\pi, \pi]^d} \left[\left |\Tr(O  \hatsigma(\bx)) - \Tr(O \rho_{\Lambda}(\bx))\right|^2 \right] \\
	= && \mathbb{E}_{\bx\sim [-\pi, \pi]^d} \left[ \left|\frac{1}{n}\sum_{i=1}^n \kappa_{\Lambda}(\bx, \bxi)\Tr(\tilderho_1(\bxi)O) -  \sum_{\bomega, \|\bomega\|\leq \Lambda}\Phi_{\bomega}(\bx) \Tr(\rho_{\bomega} O)\right|^2 \right] \\
	= && \mathbb{E}_{\bx\sim [-\pi, \pi]^d} \left[ \left| \sum_{\bomega, \|\bomega\|\leq \Lambda} \Phi_{\bomega}(\bx) \left(\frac{1}{n}\sum_{i=1}^n 2^{\|\bomega\|_0}  \Phi_{\bomega}(\bxi)\Tr(\tilderho_1(\bxi)O) -   \Tr(\rho_{\bomega} O)  \right) \right|^2 \right] \\ 
	= && \mathbb{E}_{\bx\sim [-\pi, \pi]^d}   \left[ \sum_{\bomega, \|\bomega\|\leq \Lambda}\sum_{\bomega', \|\bomega'\|\leq \Lambda} \Phi_{\bomega}(\bx)\Phi_{\bomega'}(\bx) \tilde{A}_{\bomega}\tilde{A}_{\bomega'} \right] \label{append:eqn:exp-model-dist-true-state-shadow-3} \\
    = && \sum_{\bomega, \|\bomega\|\leq \Lambda}\sum_{\bomega', \|\bomega'\|\leq \Lambda} \frac{1}{(2\pi)^d}\int_{[-\pi, \pi]^{d}}   \Phi_{\bomega}(\bx) \Phi_{\bomega'}(\bx) \mathsf{d}^d x \tilde{A}_{\bomega}\tilde{A}_{\bomega'} \label{append:eqn:exp-model-dist-true-state-shadow-4}  \\ 
    = && \sum_{\bomega, \|\bomega\|\leq \Lambda}  2^{-\|\bomega\|_0}        \tilde{A}_{\bomega}^2 \label{append:eqn:exp-model-dist-true-state-lowesa-6} \\ 
	= &&   \sum_{\bomega, \|\bomega\|\leq \Lambda}  2^{-\|\bomega\|_0}    \left| \frac{1}{n}\sum_{i=1}^n 2^{\|\bomega\|_0}  \Phi_{\bomega}(\bxi)\Tr(\tilderho_1(\bxi)O) -   \Tr(\rho_{\bomega} O)   \right|^2 \label{subeqn:exp-kernel-Lowesa-1} \\
	\equiv && \sum_{\bomega, \|\bomega\|\leq \Lambda}    |\tilde{D}_{\bomega}(\mathcal{T})|^2,
	\end{eqnarray}
\end{subequations}
where Eq.~(\ref{append:eqn:exp-model-dist-true-state-shadow-3}) adopts the explicit form of the trigonometric monomial kernel $\kappa_{\Lambda}$, Eq.~(\ref{append:eqn:exp-model-dist-true-state-shadow-3}) employs the definition $\tilde{A}_{\bomega}$ in Eq.~(\ref{append:eqn:def-A-omega}),  and Eq.~(\ref{append:eqn:exp-model-dist-true-state-lowesa-6}) comes from the evaluation of the orthogonality of basis functions in Eq.~(\ref{append:eqn:orthogonal-tri-expansion}).

Such reformulation suggests that the derivation of the upper bound of $\mathbb{E}_{\bx\sim [-\pi, \pi]^d} [ |\Tr(O  \hatsigma(\bx)) - \Tr(O \rho(\bx))|^2]$ is reduced to deriving the upper bound of  $\tilde{D}_{\bomega}(\mathcal{T})$ for $\forall \bomega \in \mathfrak{C}(\Lambda)$. To do so,  we rewrite the term $\Tr(\rho_{\bomega} O) $ in $\tilde{D}_{\bomega}(\mathcal{T})$ as the trigonometric monomial   expansion of $\rho(\bx)$, i.e.,
\begin{subequations}
	\begin{eqnarray}\label{append:eqn:foureir-trans-rho}
		\Tr(\rho_{\bomega}O) 
		 =  && 2^{\|\bomega\|_0}   \frac{1}{(2\pi)^d}\int_{[-\pi, \pi]^{d}}   \Phi_{\bomega}(\bx) \Tr(\rho(\bx)O) \mathsf{d}^d x \\
		 = && 2^{\|\bomega\|_0}   \mathbb{E}_{\bx\sim [-\pi,\pi]^d} \Phi_{\bomega}(\bx) \Tr(\rho(\bx)O) \\
		 = && 2^{\|\bomega\|_0}   \mathbb{E}_{\bx\sim [-\pi,\pi]^d} \Phi_{\bomega}(\bx) \mathbb{E}_{s_1^{\bx},...,s_N^{\bx}}\Tr(\tilderho_1(\bx)O).
	\end{eqnarray}
\end{subequations}
Accordingly, the quantity $|\tilde{D}_{\bomega}(\mathcal{T})|^2$ yields 
\begin{eqnarray}
	|\tilde{D}_{\bomega}(\mathcal{T})|^2 && =  \left|  \frac{1}{n}\sum_{i=1}^n \Phi_{\bomega}(\bxi)\Tr(\tilderho_1(\bxi)O) -  \mathbb{E}_{\bx\sim [-\pi,\pi]^d}  \Phi_{\bomega}(\bx) \mathbb{E}_{s_1^{\bx},...,s_N^{\bx}}\Tr(\tilderho_1(\bx)O)\right|^2. 
\end{eqnarray}
This formalism hints that we can use Hoeffding's inequality to bound $\tilde{D}_{\bomega}(\mathcal{T})$. Recall that the requirement of applying Hoeffding's inequality is ensuring the expectation value is bounded. In our case, we have
\begin{subequations}
	 \begin{eqnarray}
 	  |\Phi_{\bomega}(\bx) \Tr(\tilderho_1(\bx)O)| 
 	\leq &&   |  \Tr(\tilderho_1(\bx)O)| \\
 	\leq &&   \|O\|_{\infty} \|\tilderho_1(\bx)\|_1 \\
 	= && 3^K B, 
 \end{eqnarray}
\end{subequations}
where the first inequality uses the Cauchy–Schwarz inequality and $|\Phi_{\bomega}(\bx)|\leq 1$,  the second inequality adopts the von Neumann's trace inequality with H\"older's inequality, and the last inequality exploits $\|\tilderho_1(\bx)\|_1\leq 3^K$ \cite[Eq.~(F37)]{huang2022provably} and the condition $\|O\|_{\infty}\leq B$.

The bounded expectation term enables us to use Hoeffding's inequality to attain the following result, i.e., 
\begin{equation}
		\Pr\left[\tilde{D}_{\bomega}(\mathcal{T})^2 \geq \tau^2 \right] = \Pr\left[\tilde{D}_{\bomega}(\mathcal{T}) \geq \tau \right] \leq 2 \exp\left(- \frac{2 n \tau^2}{B^2 9^K} \right).
\end{equation}
Denote the set of truncated frequencies as $\mathfrak{C}(\Lambda) =\{\bomega|\bomega \in \{0, \pm 1\}^d, ~s.t.~\|\bomega\|_0\leq \Lambda\}$. This leads to 
\begin{equation}
	\Pr\left[\sum_{\bomega \in \mathfrak{C}(\Lambda)} |D_{\bomega}(\mathcal{T})|^2     \geq |\mathfrak{C}(\Lambda)|\tau^2 \right] 
	\leq \sum_{\bomega \in \mathfrak{C}(\Lambda)} \Pr\left[|D_{\bomega}(\mathcal{T})|^2 \geq \tau^2 \right] \leq |\mathfrak{C}(\Lambda)| \cdot 2 \exp\left(- \frac{2n\tau^2}{B^2 9^K} \right).
\end{equation}
Let the right-hand side be $\delta$. We have
\begin{equation}
\tau =	\sqrt{\frac{1}{2n}B^2 9^K \log \left(\frac{2 \cdot |\mathfrak{C}(\Lambda)|}{\delta}\right)}.
\end{equation}
This concludes the proof as with probability at least $1-\delta$, the mean-square error between the prediction and the ground truth taken over the randomness of the sampled inputs and classical shadow is upper bounded by    
\begin{eqnarray}
\mathbb{E}_{\bx\sim [-\pi ,\pi]^d} \left[ \left| \Tr(O \sigma_n^{(1)}(\bx)) - \Tr(O\rho(\bx))  \right|^2 \right] \leq  |\mathfrak{C}(\Lambda)|  \frac{1}{2n} B^2 9^K  \log \left(\frac{2 \cdot |\mathfrak{C}(\Lambda)|}{\delta}\right).   
\end{eqnarray}
\end{proof}

\section{Computational time for training and prediction}\label{append:sec:complexity-analysis}

Here we analyze the computational cost of the proposed ML model. For clarity, here we separately analyze the computational time of our proposal required in the training and inference.   
 
 \noindent\underline{Training time}. Recall that the training  procedure of our model amounts to loading the collected training dataset $\mathcal{T}_{\mathsf{s}}=\{\bxi, \tilde{\rho}_T(\bxi)\}_{i=1}^n$ to the classical memory. According to the explanation of classical shadow in SM~\ref{append:subsec:classical-shadows}, the required computation cost to store and load the $N$-qubit state $\tilde{\rho}_T(\bxi)$ with $T$ snapshots is $\mathcal{O}(NT)$. Simultaneously, the computation cost to store and load the classical input $\bxi$ is $\mathcal{O}(d)$. Combining these facts with the result of Theorem~2 such that the required total number of training examples $n$ of our model, the computation cost to load the dataset $\mathcal{T}_{\mathsf{s}}$ is 
 \begin{eqnarray}
 	\mathcal{O}(nNT)= \mathcal{O}\left(NT |\mathfrak{C}(\Lambda)|  \frac{B^2 9^K}{\epsilon}  \log \left(\frac{2 \cdot |\mathfrak{C}(\Lambda)|}{\delta}\right) \right) =\widetilde{\mathcal{O}}\left(  \frac{B^2 9^K NT |\mathfrak{C}(4C/\epsilon)|}{\epsilon} \right) .
 \end{eqnarray}
 As discussed in the main text, when $C$ is bounded (i.e., $C\sim \mathcal{O}(1/\text{poly}(N))$ or $C\sim \mathcal{O}(1/\text{exp}(N))$) or $d$ is a small constant in many practical scenarios, the cardinality $|\mathfrak{C}(4C/\epsilon)|$ polynomially scales with $N$ and $d$, and thus training our model is computationally efficient. 
 
 \medskip
  \noindent\underline{Inference (prediction) time}. Suppose that the observable $O$ is constituted by multiple local observables with a bounded norm, i.e., $O=\sum_{i=1}^q O_i$ and $\sum_l \|O_i\|\leq B$, and the maximum locality of $\{O_i\}$ is $K$. Following the definition of our model in Eq.~(\ref{append:eqn:TriGeo-non-trunc-form}), given a new input $\bx$ and the observable $O$, the prediction yields 
  \begin{equation}
  h_{\mathsf{s}} \equiv\Tr(O \hatsigma(\bx))=\frac{1}{n}\sum_{i=1}^n\hatkappa\left(\bx, \bxi\right)	\Tr\left(O \tilderho_T(\bxi)\right).
  	\end{equation}  
  In this regard, the evaluation involves summing over the assessment of each training example  $(\bxi, \tilderho_T(\bxi))$ for $\forall i \in [n]$, and the evaluation of each training example can be further decomposed into two components. That is, the first component is classically computing the shadow estimation $\Tr(O\tilde{\rho}_T(\bxi))$; the second component is calculating the kernel function $\kappa_{\Lambda}(\bx, \bxi)$ for $\forall i \in [n]$. According to the runtime complexity of shadow estimation elucidated in SM~\ref{append:subsec:classical-shadows}, when $K\sim \mathcal{O}(1)$, the computation of each $\Tr(O_i\tilde{\rho}_T(\bxi))$ can be completed in $\mathcal{O}(T)$ time after storing the classical shadow in the classical memory. In this regard, the computation cost to complete the first part is
  \begin{equation}
  	\mathcal{O}(Tq).
  \end{equation}
 Moreover, based on the explicit definition of the kernel $\kappa_{\Lambda}(\cdot, \cdot)$, the computation cost of evaluating $\kappa_{\Lambda}(\bx, \bxi)$ is 
 \begin{equation}
 	\mathcal{O}(|\mathfrak{C}(4C/\epsilon)|).
 \end{equation}
 In conjunction with the computation cost of each example and the total number of training examples $n$, the required predicting time of our proposal is 
    \begin{equation}
  	\mathcal{O}(n(Tq+|\mathfrak{C}(4C/\epsilon)|)) \leq \widetilde{\mathcal{O}}\left( \frac{TqB^2 9^K |\mathfrak{C}(4C/\epsilon)|^2}{\epsilon} \right).
  \end{equation}
    
\medskip    
In conclusion, when $C$ is bounded, the whole procedure of our proposal (encompassing the training and predicting) is both computational and memory efficient, which is upper bounded by 
\begin{equation}
\widetilde{\mathcal{O}}\left( \frac{TNqB^2 9^K |\mathfrak{C}(4C/\epsilon)|^2}{\epsilon} \right).
\end{equation}

\section{Classical prediction model with the full expansion (Proof of Corollary 1)}
This section comprises two parts. In SM~\ref{append:subsec:clc-pred-model-impl}, we elaborate on the implementation details of the protocol introduced in Corollary 1. Then, in SM~\ref{append:subsec:proof-coro-1}, we present the proof details of Corollary 1.

\subsection{Implementation of the classical prediction model} \label{append:subsec:clc-pred-model-impl}

\begin{figure}[h!]
	\centering	\includegraphics[width=0.98\textwidth]{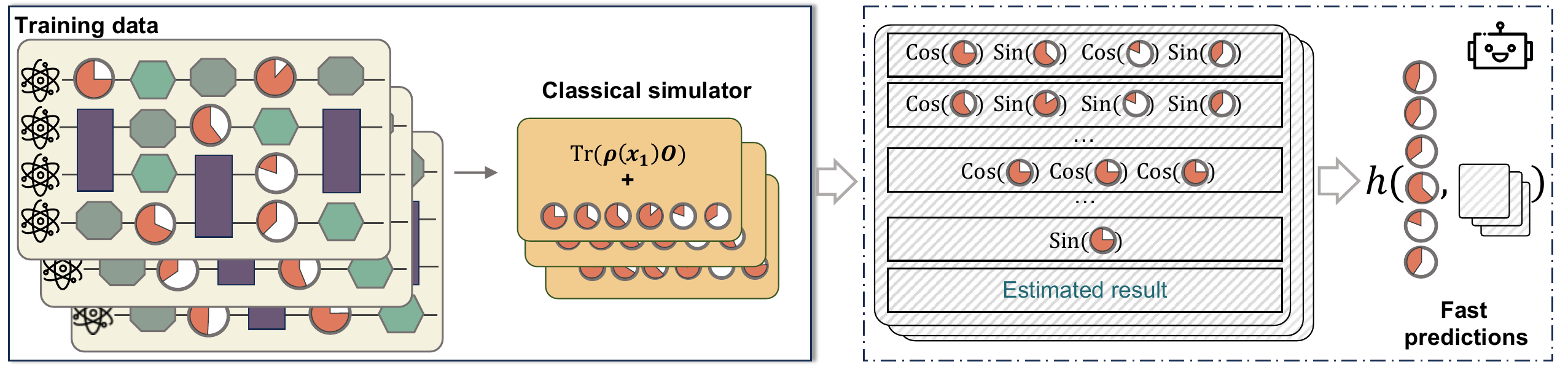}
	\caption{\justifying\small{\textbf{Purely classical learning model with the full expansion.} The pipeline is similar to the one presented in Fig.~1 of the main text. The only difference is the way of collecting training data, where the circuit layout is known by the learner and the expectation value $\Tr(\rho(bx)O)$ should be efficiently calculated by the employed classical simulators. }}
	\label{supp:fig:pure-clc-learning-model}
\end{figure}

The implementation of the protocol introduced in Corollary 1 is visualized in Fig.~\ref{supp:fig:pure-clc-learning-model}. The only difference with the one introduced in the main text is the way of collecting training data. Specifically, when the tunable quantum circuit architecture formed by $\RZ + \CI$ is known to the learner, the classical simulators (e.g.,  tensor network simulators, LOWESA, and near Clifford circuits simulators) are employed to collect the training data 
\begin{equation}
	\mathcal{T}_{\mathsf{c}}=\left\{\bxi \rightarrow \Tr\left(\rho(\bxi)O\right) \right\}_{i=1}^n.
\end{equation}
 
Once the training dataset is prepared, the learner uses it to form the kernel-based prediction model, i.e., given a new input $\bx$ and an observable $O$, it takes the form as 
\begin{equation}\label{append:eqn:TriGeo-classical-backend}
	h_{\mathsf{c}}(\bx,O)=\frac{1}{n}\sum_{i=1}^n\kappa\left(\bx, \bxi \right)\Tr(\rho(\bxi)O)~\text{with}~\kappa(\bx, \bxi) = \sum_{\bomega, \|\bomega\|_0\leq d} 2^{\|\bomega\|_0}\Phi_{\bomega}(\bx)\Phi_{\bomega}(\bxi) \in \mathbb{R}.
\end{equation}
Note that this protocol is relatively more restrictive than the one introduced in the main text, as the dimension $d$ is required to be sufficiently low with the full expansion and the expectation value under the observable $O$ should be efficiently calculated by the exploited classical simulators. 

\subsection{Proof of Corollary 1}\label{append:subsec:proof-coro-1}

The formal statement of Corollary 1 is as follows.
\begin{corollary-non}[Restatement of Corollary 1] Following notations in the main text, consider a parametrized family of $N$-qubit states $\mathcal{Q}$ and   $\|O\|_{\infty}\leq B$ Then, let the dataset be $\mathcal{T}_{\mathsf{c}}=\{\bxi \rightarrow \Tr(\rho(\bxi)O) \}_{i=1}^n$ with $\bxi \sim  \textnormal{Unif}[-\pi, \pi]^d$ and $n = {3^d B^2 \log\left(\frac{2\cdot 3^d}{\delta}\right)}{(2\epsilon)^{-1}}$. Then  with  probability at least $1-\delta$, the average prediction error of the prediction model in Eq.~(\ref{append:eqn:TriGeo-classical-backend}) satisfies 
\begin{equation} 
	\mathbb{E}_{\bx\sim [-\pi, \pi]^d} \left| h_{\mathsf{c}}(\bx,O) - \Tr(O\rho(\bx))  \right|^2 \leq \epsilon.
\end{equation}	
\end{corollary-non}

\begin{proof}[Proof of Corollary 1]
  	 
Since the full expansion is adopted, the truncation error is eliminated. We only need to quantify the estimation error induced by the finite training examples. The discrepancy can be reformulated as 
\allowdisplaybreaks
\begin{subequations}
\begin{eqnarray}
	&& \mathbb{E}_{\bx\sim [-\pi, \pi]^d} \left[\left |h_{\mathsf{c}}(\bx,O) - \Tr(O \rho(\bx))\right|^2 \right] \\
	= && \mathbb{E}_{\bx\sim [-\pi, \pi]^d} \left[ \left| \sum_{\bomega} \Phi_{\bomega}(\bx) \left(\frac{1}{n}\sum_{i=1}^n 2^{\|\bomega\|_0}  \Phi_{\bomega}(\bxi)\Tr(\rho(\bxi)O) -   \Tr(\rho_{\bomega}O)  \right) \right|^2 \right] \\ 
	= && \mathbb{E}_{\bx\sim [-\pi, \pi]^d}   \left[ \sum_{\bomega}\sum_{\bomega'} \Phi_{\bomega}(\bx)\Phi_{\bomega'}(\bx) A_{\bomega}A_{\bomega'} \right] \\
    = && \sum_{\bomega}\sum_{\bomega'} \frac{1}{(2\pi)^d}\int_{[-\pi, \pi]^{d}}   \Phi_{\bomega}(\bx) \Phi_{\bomega'}(\bx) \mathsf{d}^d x A_{\bomega}A_{\bomega'}   \\ 
    = && \sum_{\bomega}  2^{-\|\bomega\|_0}     \Phi_{\bomega}(\bx)  A_{\bomega}^2   \\ 
	= &&   \sum_{\bomega}  2^{-\|\bomega\|_0}    \left| \frac{1}{n}\sum_{i=1}^n 2^{\|\bomega\|_0}  \Phi_{\bomega}(\bxi)\Tr(\rho(\bxi)O) -   \Tr(\rho_{\bomega} O)   \right|^2 \label{subeqn:exp-kernel-Lowesa-clc-1} \\
	\equiv && \sum_{\bomega}    |D_{\bomega}(\mathcal{T})|^2,
	\end{eqnarray}
\end{subequations}
where Eq.~(\ref{subeqn:exp-kernel-Lowesa-clc-1}) comes from the evaluation of the orthogonality of basis functions in Eq.~(\ref{append:eqn:orthogonal-tri-expansion}). Recall the analysis in the proof of Theorem 1, the quantity $|D_{\bomega}(\mathcal{T})|^2$ yields 
\begin{eqnarray}
	|D_{\bomega}(\mathcal{T})|^2 && =  \left|  \frac{1}{n}\sum_{i=1}^n \Phi_{\bomega}(\bxi)\Tr\left(\rho(\bxi)O \right) -  \mathbb{E}_{\bx\sim [-\pi,\pi]^d}  \Phi_{\bomega}(\bx) \Tr(\rho(\bx)O) \right|^2. 
\end{eqnarray}
This formalism hints that we can use Hoeffding's inequality to bound $D_{\bomega}(\mathcal{T})$.   In this case, we have
 \begin{eqnarray}
 	  |\Phi_{\bomega}(\bx) \Tr(\rho(\bx)O)| 
 	\leq     |  \Tr(\rho(\bx)O)|  
 	\leq   \|O\|_{\infty} \|\rho(\bx)\|_1  
 	=  B, 
 \end{eqnarray}
where the first inequality uses the Cauchy–Schwarz inequality and $|\Phi_{\bomega}(\bx)|\leq 1$,  the second inequality adopts von Neumann's trace inequality with Hölder's inequality, and the last inequality exploits $\|\rho(\bx)\|_1=1$ and $\|O\|_{\infty}\leq B$.

The bounded expectation term enables us to use Hoeffding's inequality to attain the following result, i.e., 
\begin{equation}
		\Pr[D_{\bomega}(\mathcal{T})^2 \geq \tau^2] = \Pr[D_{\bomega}(\mathcal{T}) \geq \tau] \leq 2 \exp\left(- \frac{2 n \tau^2}{B^2} \right).
\end{equation}
 This leads to 
\begin{equation}
	\Pr\left[\sum_{\bomega \in \mathfrak{C}(d)} |D_{\bomega}(\mathcal{T})|^2     \geq |\mathfrak{C}(d)|\tau^2 \right] 
	\leq \sum_{\bomega \in \mathfrak{C}(d)} \Pr\left[|D_{\bomega}(\mathcal{T})|^2 \geq \tau^2 \right] \leq 3^d \cdot 2 \exp\left(- \frac{2n\tau^2}{B^2} \right).
\end{equation}
Let the right-hand side be $\delta$. We have
\begin{equation}
\tau =	\sqrt{\frac{1}{2n}B^2  \log \left(\frac{2 \cdot 3^d}{\delta}\right)}.
\end{equation}
This concludes the proof as with probability at least $1-\delta$, the mean-square error between the prediction and the ground truth taken over the randomness of sampled inputs and classical shadow is upper bounded by    
\begin{eqnarray}
\mathbb{E}_{\bx\sim [-\pi ,\pi]^d} \left[ \left| \Tr(O \hatsigmaTriGeoClc(\bx)) - \Tr(O\rho(\bx))  \right|^2 \right] \leq 3^d  \frac{1}{2n} B^2 \log \left(\frac{2 \cdot 3^d}{\delta}\right).   
\end{eqnarray}
Let the right-hand side be equal to the tolerant error $\epsilon$. We have
\begin{equation}
	n = \frac{3^d B^2 \log\left(\frac{2\cdot 3^d}{\delta}\right)}{2\epsilon}.
\end{equation} 	
\end{proof}

\section{Learning bounded-gate quantum circuit with $\CI$ gates and parameterized multi-qubit gates}

In the main text, we primarily elaborate on the application of the proposed ML model in predicting bounded-gate quantum circuits consisting of $\RZ$ and $\CI$ gates. However, our proposal is adaptable and can be readily extended to predict bounded-gate quantum circuits composed of alternative basis gate sets. To exemplify this flexibility, in this section, we illustrate how our proposal and the associated theoretical findings (Theorem 2) can be effectively expanded to encompass a wider context, specifically bounded-gate circuits incorporating CI gates alongside parameterized multi-qubit gates generated by arbitrary Pauli strings.

For ease of notations, given an $N$-qubit quantum circuit, we define the $j$-th parameterized  gate generated by the Pauli string $P_{\bm{a}_j}\in \{\mathbb{I},X,Y,Z\}^N$ with $\bm{a}_j\in \{0, 1, 2, 3\}^N$ as 
\begin{equation}\label{append:eqn:bound-gate-circ-arb-P}
	\RP(\bx_j) = \exp\left(-\imath \frac{\bx_j}{2}P \right) \equiv \cos\left(\frac{\bx_j}{2}\right)\mathbb{I}_{2^N} + \imath \sin\left(\frac{\bx_j}{2}\right)P. 
\end{equation}
Then the bounded-gate quantum circuit takes the form as  
\begin{equation}
	U(\bx) = \prod_{j=1}^{d}(\RP_{\bm{a}_j}(\bx_j)u_e),
\end{equation}
where $u_e\in \CI$ refers to an arbitrary Clifford gate or their combinations. In what follows, we demonstrate how the proposed ML model effectively predicts the incoherent dynamics of bounded-gate circuits in Eq.~(\ref{append:eqn:bound-gate-circ-arb-P}), with the resulting prediction error aligning with the findings outlined in Theorem~2.

Recall the implementation of the proposed ML model and the proof of Theorem~2 presented in SM~\ref{append:sec:proof-thm2}. In brief, if the pre-measured state under the Pauli-basis expansion can be expressed as the form 
\begin{equation} 
	\rho(\bx) = U(\bx)(\ket{0}\bra{0})^{\otimes N} U(\bx)^{\dagger} =  \sum_{\bomega}\Phi_{\bomega}(\bx) \llangle 0| \bUnitary_{\bomega}^{\dagger} \equiv  \sum_{\bomega}\Phi_{\bomega}(\bx) \rho_{\bomega},
\end{equation}
where the trigonometric monomial basis is $
	\Phi_{\bomega}(\bx) = \prod_{i=1}^d \begin{cases}
		 1 ~ & \textnormal{if}~ \bomega_i = 0 \\
		 \cos(\bx_i) & \textnormal{if}~\bomega_i = 1 \\
		 \sin(\bx_i) &  \textnormal{if}~ \bomega_i = -1
	\end{cases}$, then we can form the state prediction model in Eq.~(\ref{append:eqn:TriGeo-non-trunc-form}), i.e., 
\begin{equation}
	\hat{\sigma}_n(\bx) = \frac{1}{n}\sum_{i=1}^n\kappa\left(\bx, \bxi\right)\tilderho_T(\bxi), \quad \text{with}~ \kappa(\bx, \bxi) = \sum_{\bm{k}\in \mathbb{Z}^d, \|\bm{k}\|_2\leq \Lambda} \cos\left(\pi \bm{k}\cdot\left(\bx-\bxi\right)\right)\in \mathbb{R},  
\end{equation}
and the analysis in  Theorem~2 applies. In this respect, the key aspect of extending our proposal from $\{\RZ + \CI\}$ gate set to  $\{\RP_{\bm{a}_j} + \CI\}$ gate set involves demonstrating that each parameterized gate $\RP_{\bm{a}_j}$ with $j\in [d]$ can be represented in the form specified by Eq.~(\ref{append:eqn:PTM_RZ}), i.e.,
\begin{equation}\label{append:eqn:PTM-RPauli-gate}
	\bRP_{\bm{a}_j}(\bx_j) = D_0' + \cos(\bx_j)D_1' + \sin(\bx_j)D_2', 
\end{equation} 
where $D_0'$, $D_1'$, and $D_2'$ are three constant matrices with the size $2^N\times 2^N$. 

We next prove that the form in Eq.~(\ref{append:eqn:PTM-RPauli-gate}) is satisfied for any Pauli string $P_{\ba_j}$. For clarity, in the following analysis, we denote $P_{\ba_j}$ and $\RP_{\ba_j}(\bx_j)$ as $P_{j}$ and $\RP_j(x)$, respectively.  According to the definition of PTM, we have 
\begin{eqnarray}\label{append:eqn:PTM-RPauli-gate-2}
	[\bRP_{\bm{a}_j}(\bx_j)]_{ik} =   \Tr\left(\RP_{j}(x) P_i \RP_{j}(-x) P_k \right).
	\end{eqnarray}   
Since $P_i$, $P_j$, and $P_k$ are Pauli strings, any two of them must commute or anti-commute to each other. By exploiting this property, we next separately analyze the value of $[\bRP_{\bm{a}_j}(\bx_j)]_{ik}$ for $\forall i, k \in [4^d]$.

\smallskip
\noindent\textit{Case I: $[P_i, P_j]=0$ or $[P_k, P_j]=0$.} In this case, Eq.~(\ref{append:eqn:PTM-RPauli-gate-2}) can be simplified to
\begin{equation}
	[\bRP_{\bm{a}_j}(\bx_j)]_{ik} = \Tr(P_iP_k) = \delta_{ij}.
\end{equation}  
In other words, for the indices $i,k \in [4^d]$ whose corresponding Pauli strings commute with $P_j$, the corresponding PTM entry is a constant, which is zero if $i\neq k$ and $1$ if $i=k$. According to the formula in Eq.~(\ref{append:eqn:PTM-RPauli-gate}), the relevant  entries, whose indices satisfy $[P_i, P_j]=0$ or $[P_k, P_j]=0$, belong to $D_0'$.

\smallskip
\noindent\textit{Case II: $\{P_i, P_j\}=0$ or $\{P_k, P_j\}=0$.}   In this case, Eq.~(\ref{append:eqn:PTM-RPauli-gate-2}) can be simplified to
\begin{equation}
	[\bRP_{\bm{a}_j}(\bx_j)]_{ik} = \Tr(\RP_{j}(2x) P_i  P_k ) = \cos(x)\Tr(P_iP_k) + \imath \sin(x)\Tr(P_jP_iP_k) = \cos(x)\delta_{ij} + \imath \sin(x)\Tr(P_jP_iP_k),
\end{equation}  
where the second equality employs $\RP_{j}(2x)=\cos(x)\mathbb{I}_{2^N} + \imath \sin(x)P_j$.  In other words, for the indices $i,k \in [4^d]$ whose corresponding Pauli strings anticommute with $P_j$, (i) when $i= k$, the corresponding PTM entry only contributes to $D_1'$ in Eq.~(\ref{append:eqn:PTM-RPauli-gate}); (ii) when $i\neq k$ and $P_jP_iP_k = \mathbb{I}_{2^N}$, the corresponding PTM entry only contributes to $D_2'$ in Eq.~(\ref{append:eqn:PTM-RPauli-gate}); (iii) else, the corresponding PTM entry is zero, which only contributes to $D_0'$.

Taken together, for any pair indices $(i,k)$, the corresponding entry only contributes to one of $D_0'$, $D_1'$, and $D_2'$, indicating the relation in Eq.~(\ref{append:eqn:PTM-RPauli-gate}). As a result, our proposal and Theorem~2 applies to the basis gate set $\{\RP_{\bm{a}_j}, \CI\}$.

\section{Connection and differnces with Ref.~\cite{zhao2023learning} and Pauli path simulators}

In this section, we detail how our work relates to and differs from Ref.~\cite{zhao2023learning} and Pauli path simulators. Specifically, in Subsection~\ref{append:sec:complement-zhaolearning}, we discuss the complementary relationship between our work and Ref.~\cite{zhao2023learning}, highlighting how the combination of these studies enriches quantum learning theory and deepens our understanding of the learnability of bounded-gate quantum circuits. In SM~\ref{append:sec:diff-Pauli-sim}, we clarify the distinctions between the proposed learning model and Pauli path simulators, addressing differences at both the technical and application levels.
 
\subsection{Complementation with Ref.~\cite{zhao2023learning}}\label{append:sec:complement-zhaolearning}

For comprehensive, we summarize the theoretical results achieved by Ref.~\cite{zhao2023learning} and our work in Table~\ref{tab:comparison}. Particularly, Ref.~\cite{zhao2023learning} investigates the computational complexities of quantum learners across three learning tasks: 
\begin{itemize}
	\item  Task~(i)  learning bounded-gate unitaries;
	\item  Task~(ii)  learning bounded-gate quantum states;
	\item  Task~(iii)  learning $\mathbb{R}$-valued classical functions with 1-Lipschitz continuity. 
\end{itemize}
Different from Ref.~\cite{zhao2023learning}, we explore the computational complexities of classical learners for a new task:
\begin{itemize}
	\item  Task~(iv) predicting the expectation values of bounded-gate quantum circuit for a given observable.
\end{itemize}
The achieved theoretical results under {different learning tasks} (i.e., Tasks (i)-(iv)) and {different learning paradigms} (i.e., quantum learners versus classical learners) are complementary, providing the following two crucial insights in the context of the learnability of bounded-gate quantum circuits.

\begin{table}[h!]
\caption{\justifying\small{\textbf{Summarization of theoretical results achieved by Ref.~\cite{zhao2023learning} and our work.} The terms `Q-learner' and `C-learner' refer to the  quantum and classical learning paradigms, respectively. Besides, the terms `Samp' and `Comp' refer to the abbreviations of `sample complexity' and `computational complexity', respectively.  C-function refers to the classical function with $\nu$ being the number of variables in the classical function (Task (iii)). The notations $N$ and $\epsilon$ denote the number of qubits and the prediction error. In Ref.~\cite{zhao2023learning}, the explored bounded-gate circuit is realized by $\hat{G}$ two-qubit gates, while in our work the explored bounded-gate circuit is composed of $d$ $\RZ$ gates and $G-d$ Clifford gates. The notations $B$ and $T$ refer to the norm of the observable and the snapshots when collecting each training example, respectively. The markers \cmark, \xmark,  $\mathbbold{\qm}$ separately represent `yes', `no', and `unknown'.    }}
\label{tab:comparison}
\centering
\footnotesize
\begin{adjustbox}{width=.85\columnwidth,center}
\begin{tabular}{c|c|c|c|c}
\toprule 

& \begin{tabular}[c]{@{}c@{}} Learn unitary \\  
(Task (i)) \end{tabular}  & \begin{tabular}[c]{@{}c@{}} Learn state \\  
(Task (ii)) \end{tabular}  &  \begin{tabular}[c]{@{}c@{}} Learn C-function \\  
(Task (iii)) \end{tabular}  &   \begin{tabular}[c]{@{}c@{}} Learn observable \\  
(Task (iv)) \end{tabular} \\ 
                  \midrule  
Q-learner   &   \begin{tabular}[c]{@{}c@{}} Theorems 4\&6, Ref.~\cite{zhao2023learning} \\ Samp: $\begin{cases}
	\tilde{\mathcal{O}}(\hat{G}\min\{\frac{1}{\epsilon^2},\frac{\sqrt{2^N}}{\epsilon}\})\\
	\Omega(\frac{\hat{G}}{\epsilon})
\end{cases}$ \\  Comp: $e^{\Omega(\min\{\hat{G},N\})}$    \end{tabular}                &    \begin{tabular}[c]{@{}c@{}} Theorems 1\&2, Ref.~\cite{zhao2023learning} \\ \smallskip  Samp: $\tilde{\Theta}(\frac{\hat{G}}{\epsilon^2})$ \\ \smallskip Comp: $e^{\Omega(\min\{\hat{G},N\})}$    \end{tabular}                    &  \begin{tabular}[c]{@{}c@{}} Theorem 7, Ref~\cite{zhao2023learning} \\  \smallskip Samp: $\Omega(\frac{1}{\epsilon^{\nu}})$ \\ \smallskip Comp: $\tilde{\Omega}(\frac{1}{\epsilon^{\nu}})$  \end{tabular}                        &         \large{$\mathbbold{\qm}$}               \\  &  Algorithm:  \cmark & Algorithm:  \cmark & Algorithm:  \xmark  &       \\ \midrule 
C-learner &     \large{$\mathbbold{\qm}$}             &     \large{$\mathbbold{\qm}$}                    &   \xmark                      &   \begin{tabular}[c]{@{}c@{}} Theorem 1, \textbf{this work}  \\ Samp: $\begin{cases}
	\tilde{O}(\frac{B^2d+B^2NG}{\epsilon}) \\ \tilde{\Omega}(\frac{(1-\epsilon)d}{\epsilon T}) 
\end{cases}$  \\
Comp: $O(\exp(N,d))$
 \end{tabular}  
\\ & & &  &   Algorithm:  \cmark   (Theorem 2)
 \\ \midrule       
\end{tabular}
\end{adjustbox}
\end{table}

Before moving on to discuss the two insights, let us emphasize the necessity of exploring both classical and quantum learners. Recall that Ref.~\cite{zhao2023learning} focuses on the quantum learner, which may implement the learning model on quantum computers (e.g., quantum neural networks) and require quantum resources during both the training and inference stages. In contrast, our work centers on the classical learner, which implements the learning model on classical computers, requiring quantum resources only during the data collection stage. This means that both the training and inference stages for classical learners are conducted entirely on classical hardware. A comprehensive analysis about the computational complexities of both quantum and classical learners can  {guide various practical applications} and {indicate the utility of quantum computers}. On the one side, given the foreseeable scarcity of quantum resources, if classical learners can achieve comparable performance to quantum learners for a specific task, they offer undeniable benefits in terms of resource efficiency. On the other side, if quantum learners demonstrate provable computational advantages over classical learners for a given task, this would establish the practical utility of quantum computing for that particular application. For example, Ref.~\cite{molteni2024exponential} provides a compelling example of how computational capabilities differ between classical and quantum learners, demonstrating that quantum learners can achieve exponential runtime speedups in predicting the properties of ground states.

\subsubsection{Insight I: Computational complexity analysis} 
Based on the achieved computational complexities presented in Table~\ref{tab:comparison}, we draw the following  key observations in the context of learning bounded-gate quantum circuits.

\smallskip
\noindent\underline{Observation I}. Let us revisit the achieved sample complexity in Tasks (i), (ii), and (iv) summarized in Table~\ref{tab:comparison}. That is, despite the decreasing levels of information gain—from learning an unknown unitary, to an unknown state, and finally to unknown expectation values, we observe a consistent scaling behavior for both classical and quantum learners. 
		
	In particular, Theorems 1 and 4 in Ref.~\cite{zhao2023learning} demonstrate that when learning a circuit or a pure state composed of $\hat{G}$ two-qubit gates, $\Theta(\hat{G})$ training examples are sufficient and necessary for a quantum learner to ensure a bounded prediction error. Similarly, Theorem~1 in the main text establishes that when predicting the expectation values of a bounded-gate circuit composed of $d$ $\RZ$ gates and $G - d$ Clifford gates, $\Theta(d)$ training examples are sufficient and necessary for a classical learner to achieve a bounded prediction error.
	
	We would like to note that although Ref.~\cite{zhao2023learning} and our work use different descriptions of bounded-gate circuits, the results achieved in our study can be directly extended to their framework. This is because our description of bounded-gate circuits can be efficiently transformed into their setting with only a constant overhead. More precisely, a bounded-gate quantum circuit composed of $\hat{G}$ two-qubit gates can be implemented using at most $d=24\hat{G}$ $\RZ$ gates and $G-d=35\hat{G}$ Clifford gates to implement this circuit. An intuition is provided in Fig.~\ref{fig:two-qubit-gate-decomposition}. Ref.~\cite{Vidal2004Universal} has proved that three CNOT gates, associated with eight single-qubit gates, are necessary and sufficient in order to implement an arbitrary unitary transformation of two qubits. Meanwhile, any single qubit gate can be realized by three $\RZ$ gates and four single-qubit Clifford gates \cite{nielsen2010quantum}. Taken together, $3\times 8=24$ $\RZ$ gates and  $4\times 8 + 3=35$ Clifford gates can implement an arbitrary unitary transformation of two qubits. Since there are $\hat{G}$ two-qubit gates, $d=24\hat{G}$ $\RZ$ gates and $G-d=35\hat{G}$ Clifford gates are sufficient to implement this circuit.    
		
	\begin{figure}[h!]
	\centering
	\includegraphics[width=0.95\textwidth]{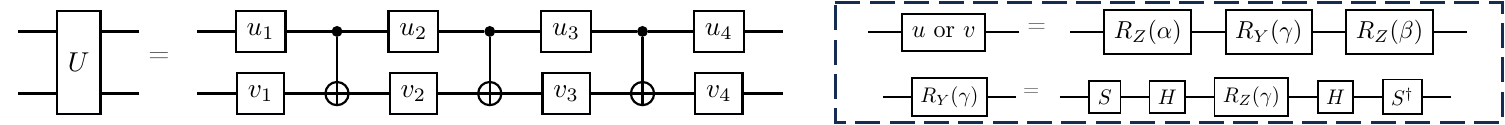}
	\caption{\justifying\small{\textbf{Decomposition of an arbitrary two-qubit gate into $\RZ$ and Clifford gates}. The dashed box highlights the implementation of the single-qubit gates $u_*$ and $v_*$ with $*\in \{1, 2, 3, 4\}$ by $\RZ$ gates and single-qubit Clifford gates. }}
	\label{fig:two-qubit-gate-decomposition}
	\end{figure}
	
	According to the result of Theorem 1 in our work, when a classical learner is applied to predict the expectation value of a bounded-gate quantum circuit composed of $\hat{G}$ two-qubit gates, $\tilde{O}((B^2\hat{G} + B^2N\hat{G})/\epsilon)$ and $\tilde{\Omega}((1-\epsilon)\hat{G}/(\epsilon T))$ training examples are sufficient and necessary to achieve a bounded prediction error. This result is consistent with the findings for Tasks (i) and (ii) achieved in Ref.~\cite{zhao2023learning}, both of which display a linear-in-$\hat{G}$ sample complexity.
		
	The linear scaling behavior with respect to the number of tunable gates across different learners (i.e., quantum versus classical) and tasks (i.e., Tasks (i), (ii), and (iv)) suggests a consistent linear-in-$d$ (or linear-in-$\hat{G}$)  sample complexity for classical learners in Tasks (i)\&(ii) and for quantum learners in Task (iv). This conjecture is supported by two key facts. First, covering and packing nets, as well as (quantum) hypothesis selection, are robust tools for analyzing the sample complexity of learning bounded-gate quantum circuits, often providing tight bounds \cite{huang2022provably,du2022efficient,caro2021generalization,du2023problem,wang2023transition}. Second, the cardinality of these nets is intrinsically dependent on the number of gates ($d$, $G$, or $\hat{G}$), which leads to similar results across tasks, despite variations in the intended information gain. 
	
\medskip	
\noindent\underline{Observation II}.  We next turn to rethink the achieved runtime complexities in Tasks (i), (ii), and (iv) summarized in Table~\ref{tab:comparison}. An immediate observation is that, for all tasks, there are no efficient \textit{quantum learners} for learning bounded-gate unitaries or quantum states \cite[Theorems 2\&6]{zhao2023learning}. In addition, there are no efficient \textit{classical learners} capable of accurately predicting the mean-value space of arbitrary bounded-gate quantum circuits (Theorem~1 in the main text). These results are complementary and offer valuable insights into the unresolved issues outlined in Table~\ref{tab:comparison}.

		First, the inherent limitations of quantum learners in learning bounded-gate unitaries or states strongly suggest that there are no efficient classical learners capable of accomplishing Tasks (i) and (ii). This is because from an information-theoretic perspective, classical learners cannot extract much more information than quantum learners. Consequently, classical learners are unlikely to outperform quantum learners in these tasks.
		
		Second, while neither quantum nor classical learners can efficiently learn arbitrary bounded-gate quantum circuits, it is natural to ask whether specific conditions on the explored bounded-gate circuits could lead to the computational separations, i.e., quantum learners can provide efficient solutions, whereas classical learners would still face computational hardness. Answering this question can enhance our understanding of the learnability of bounded-gate quantum circuits and help identify the practical utility of quantum learners.

\subsubsection{Insight II: Learning algorithm design} 
The computational hardness of learning bounded-gate quantum circuits in Tasks (i), (ii) and (iv) highlights the need to develop advanced learning algorithms with  provable computational efficiency under practical and realistic conditions. Both Ref.~\cite{zhao2023learning} and our work make initial strides in this direction. Below, we summarize the advancements of each method and outline potential directions for future research.

	 For Tasks (i) and (ii), Ref.~\cite{zhao2023learning} proposes efficient learning algorithms for relatively restrictive scenarios where the number of two-qubit gates $\hat{G}$ scales logarithmically with the number of qubits $N$, i.e., $\hat{G}\sim \mathcal{O}(\log(N))$. Conceptually, at the inference stage, given any new quantum circuit composed of $\mathcal{O}(\log(N))$ two-qubit gates, the proposed algorithm first identifies the non-trivial $\mathcal{O}(\log(N))$ qubits on which two-qubit gates are applied and then employs Pauli-based tomography on these qubits. The achieved result naturally raises two future research directions. The first direction is to explore whether there exist efficient learning algorithms beyond the logarithmic scaling scenario. The second research direction is to explore how to design effective classical learners that achieve performance comparable to quantum learners in Tasks (i) and (ii).

	\medskip	
	  For Task (iv), we devise an efficient classical-shadow predictor based on the kernel methods. While a direct comparison may not be entirely fair, it is worth highlighting that the conditions required for provable efficiency in our approach are much less restrictive than those in Ref.~\cite{zhao2023learning}. Specifically, our protocol can efficiently learn a broad class of quantum circuits composed of a polynomial number of gates, as opposed to being limited to logarithmic scenarios. The achieved results, combined with the unsolved issues in Table~\ref{tab:comparison}, lead to two interesting questions to be further explored.  The first one is to explore alternative efficient classical learning algorithms. The second one is to analyze the computational separation between quantum and classical learning algorithms for Task (iv).

In summary, the complementary roles of Ref.~\cite{zhao2023learning} and our work  expand the landscape of learning algorithms for bounded-gate quantum circuits, paving the way for deeper exploration of their utility and limitations in quantum computing.

\subsection{Connections and differences with Pauli path simulators}\label{append:sec:diff-Pauli-sim} 

The primary distinction between our work and existing classical simulators, especially for Pauli path simulators, lies in their respective purposes. Specifically, our work focuses on developing a shadow state predictor to predict the mean values of a bounded-gate quantum circuit for new rotation angles and observables. In contrast, existing classical algorithms aim to efficiently simulate the mean values of a bounded-gate quantum circuit for varying rotation angles while keeping the observables fixed. This fundamental difference—prediction versus estimation—necessitates distinct theoretical analyses (i.e., the way of using Pauli transfer matrix (PTM) and the truncation method) and underscores unique benefits and advancements our protocol offered in many practically relevant scenarios.  
 
\smallskip
\noindent \underline{Different roles of PTM}. PTM is a widely used tool in quantum computing, with applications ranging from gate set tomography \cite{greenbaum2015introduction} to quantum learning theory \cite{caro2024learning} and the analysis of quantum circuit dynamics \cite{rudolph2023classical,fontana2023classical,gonzalez2024pauli,schuster2024polynomial}. In the context of simulating the mean values of bounded-gate quantum circuits, PTM is typically combined with Pauli path simulations to compute the contribution of each path \cite{rudolph2023classical,fontana2023classical}. As explained in SM~\ref{append:subsec:trigo-monomial-exp-QC},  given an $N$-qubit quantum circuit $U(\bx)$ is composed of $d$ $\RZ$ gates and $G-d$ Clifford gates, the expectation value of the given observable $O$ under the PTM representation yields 
\[\Tr(\rho(\bx) O) = \Tr( U(\bx)(\ket{0}\bra{0})^{\otimes N} U(\bx)^{\dagger} O) =  \sum_{\bomega}\Phi_{\bomega}(\bx)  \Tr(\rho_{\bomega} O),\]
 where the notation $\Phi_{\bomega}(\bx)$ with $\bomega\in \{0, \pm 1\}^d$ refers to the trigonometric monomial basis defined in Eq.~(\ref{append:eqn:basis-tri-comp}) and $\Tr(\rho_{\bomega} O)$ refers to the expectation value of the \textit{purely-Clifford circuit} for the path indexed by $\bomega$. For Pauli path simulators, the expectation value for a portion of $\{\bomega\}$ should be {explicitly computed on the classical hardware}.
	
Unlike Pauli path simulators, our proposal adopts a fundamentally different approach to leveraging PTM. Instead of computing nontrivial $\Tr(\rho_{\bomega} O)$  like Pauli path simulators, we utilize PTM to derive the explicit formula of the bounded-gate quantum state with respect to the input $\bx$, i.e.,
	\[
		 \rho(\bx)  =  U(\bx)(\ket{0}\bra{0})^{\otimes N} U(\bx)^{\dagger}   =  \sum_{\bomega}\Phi_{\bomega}(\bx)   \rho_{\bomega},
	\] 
	where $\Phi_{\bomega}(\bx)$ and $\rho_{\bomega}$ follow the same definitions in Eq.~(\ref{append:eqn:PTM-mean-value}). This explicit form inspires us to design the shadow state predictor $\hatsigma(\bx)$ in Eq.~(\ref{append:eqn:TriGeo-non-trunc-form}).  Moreover, in Appendix~\ref{append:subsec:proof-lemma2} (proof of Lemma~\ref{lem:estimation-error-geo-kernel}), we prove that given any new input $\bx'\in [-\pi, \pi]^d$, the expectation value of the shadow state predictor $\hatsigma(\bx')$ matches the truncated target quantum state $\rho_{\Lambda}(\bx')$, i.e.,  $\mathbb{E}[\hatsigma(\bx')] = \rho_{\Lambda}(\bx') \equiv \sum_{\bomega, \|\bomega\|_0 \leq \Lambda}  \Phi_{\bomega}(\bx') \rho_{\bomega}$. 
	
	The above result sets our approach apart from the use of PTM in Pauli path simulators. Specifically, our work leverages PTM to derive the explicit formula of the concept class with respect to the input $\bx$ and subsequently design the shadow state prediction model $\hatsigma(\bx)$, whereas Pauli path simulators such as LOWESA utilize PTM to compute values for each nontrivial path.
	 
 \smallskip
\noindent \underline{Truncation strategy}. The low-weight truncation scheme, as with PTM, is also a widely utilized strategy in quantum computing. However, the approach to truncation and the elements being truncated differ substantially across applications, leading to distinct theoretical results and heuristic performances. For Pauli path simulators, low-weight truncation is applied to the frequencies $\bomega$ or the weights of Pauli operators to achieve computational efficiency while maintaining a low estimation error. In contrast, in our work, low-weight truncation is applied to the frequencies to ensure computational efficiency while minimizing the prediction error. This fundamental distinction in purpose underscores that prior results on Pauli path simulators offer no direct implications for learning protocols, necessitating an entirely new theoretical framework. 
  
Recall that Theorem~1 in the main text establishes that no efficient learning model can predict all bounded-gate quantum circuits. Consequently, the most critical challenge lies in identifying specific classes of bounded-gate quantum circuits that can be efficiently learned—a question that classical simulators do not need to address. To tackle this, we establish a connection between the gradient norm of bounded-gate circuits and the truncation error as indicated in Lemma~\ref{lem:truncation-error-geo-kernel}. This relation warrants that the proposed learning model can achieve efficient predictions with low prediction error. This novel result, which is absent in Pauli path simulators, highlights the uniqueness of our approach and opens new avenues for developing efficient learning models tailored to predicting the mean values of bounded-gate quantum circuits.

\smallskip
\noindent\underline{Unique benefits of our learning model in two scenarios}. As explained previously, our method focuses on the shadow state prediction, while LOWESA centers on Heisenberg simulations. These differences define the unique strengths of our model in the following two scenarios. 
\begin{itemize}
	\item [(i)] Circuits with many Clifford gates. Classical Pauli path simulators, such as LOWESA, must explicitly simulate Clifford gates, which, although theoretically efficient  ($o(N^2)$ complexity with $N$ being the qubit count), can be computationally expensive in terms of wall-clock time. In contrast, our learning model is insensitive to the number of Clifford gates, as these gates do not affect the learning process. We validate this advantage through numerical simulations in SM~\ref{append:sec:comp-lowesa}, demonstrating superior performance in this regime.  
 
	\item [(ii)] Circuits with unknown system noise. Classical simulators, including both tensor network simulators and Pauli path simulators, rely on assumptions about noise strength when simulating the dynamics of noisy quantum circuits. However, these assumptions generally do not align with the actual noise present in the system. In contrast, a learning-based approach (such as the extension of our model) can automatically capture and adapt to inherent system noise by leveraging the dataset collected from the target quantum system, offering a more flexible and accurate alternative. 
\end{itemize}

\section{More numerical simulations}\label{append:sec:num-res}
In this section, we demonstrate more simulation details omitted in the main text. Specifically, in SM~\ref{append:sec:num-scaling-Lambda}, we provide heuristic evidence about the required truncation value $\Lambda$ in practice.  In SM~\ref{append:subsec:sim-GHZ}, we provide more simulation details about predicting properties of rotational $N$-qubit GHZ states. Then, in SM~\ref{append:subsec:global-ham-res}, we illustrate more simulation results about the global Hamiltonian simulation task. Next, in SM~\ref{appendix:subsec:enhanced-VQA}, we demonstrate how to use the proposed ML model to enhance variational quantum algorithms, including variational quantum Eigen-solver and quantum neural networks. Last, in SM~\ref{append:sec:comp-lowesa}, we numerically explore the performance of the proposed ML model over and LOWESA when processing quantum circuits with many Clifford gates.

\subsection{Heuristic evidence about the truncation value $\Lambda$}\label{append:sec:num-scaling-Lambda}
In this subsection, we conduct a systematic analysis to heuristically understand the behavior of $\Lambda$, or equivalently $|\mathfrak{C}(4C/\epsilon)|$. The results provide insights into the conditions under which the proposed ML model can operate efficiently.

\smallskip
\noindent{\underline{The scaling behavior of $\mathfrak{C}(\Lambda)$}}. We begin by examining how the cardinality of the truncated frequency set $\mathfrak{C}(\Lambda)$ defined in Theorem~2 of the main text scales with the input dimensions $d$ (i.e., the number of $\RZ$ gates) and the truncation thresholds $\Lambda$. Mathematically, we have
\begin{equation}
	p \equiv |\mathfrak{C}(\Lambda)|=\sum_{k = 0}^{\Lambda} \binom{d}{k} 2^k.
\end{equation}

\begin{figure}[h!]
	\centering
\includegraphics[width=0.85\textwidth]{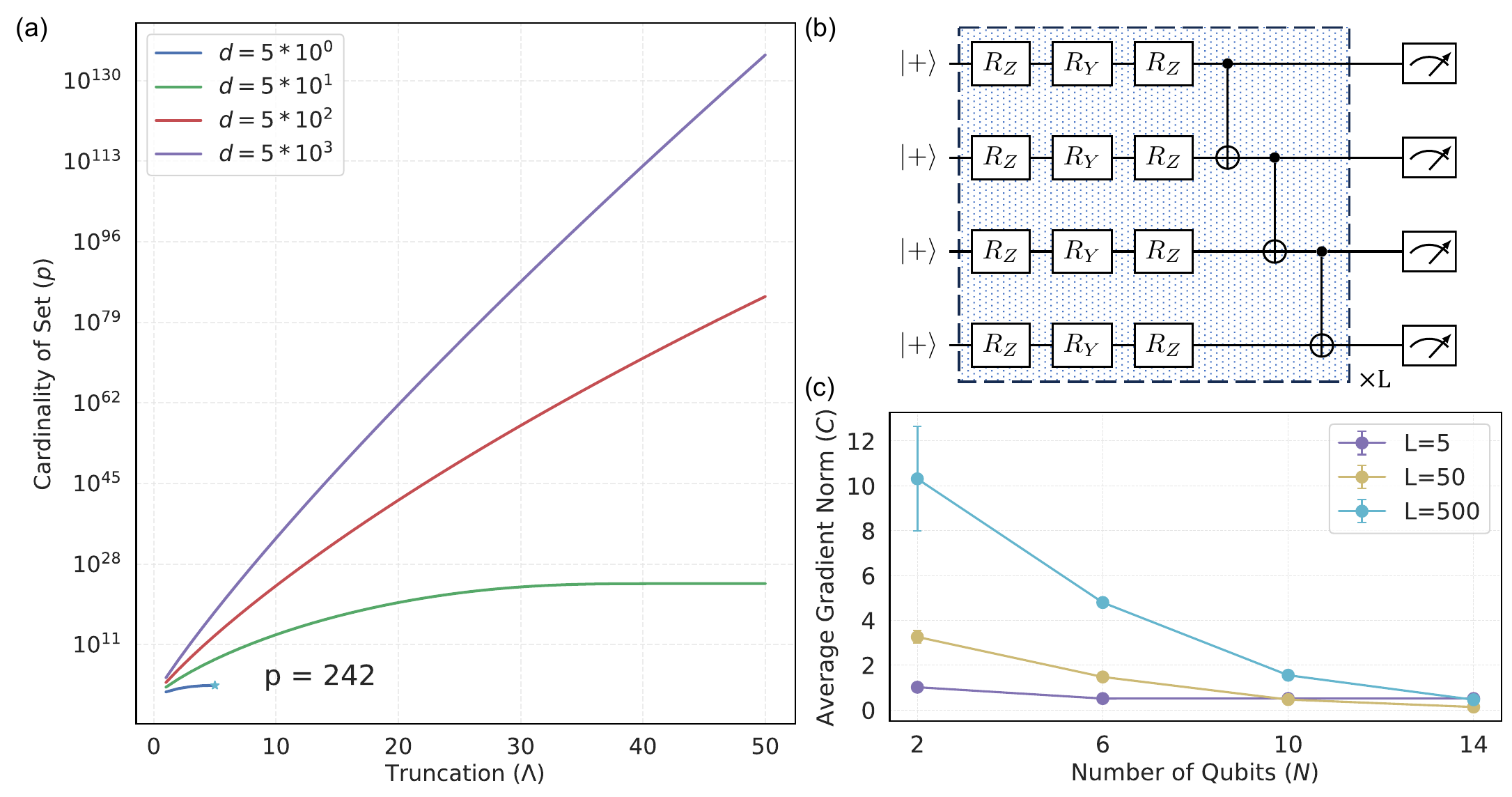}
\caption{\justifying\small{\textbf{The problem-dependent behavior of the truncation threshold $\Lambda$.} (a) \textsc{Scaling behavior of $p$}. The x-axis and y-axis represent the truncation thresholds and the cardinality of the truncation set $p\equiv|\mathfrak{C}(\Lambda)|$, respectively. The symbol `$\star$' marks that the cardinality of the truncation set is $p=242$ when $d=5$ and $\Lambda=5$. (b)  \textsc{Quantum circuit for average gradient norm computation}. The panel illustrates the quantum circuits used for computing the average gradient norm $C$ when the number of qubits is $N=4$.  The dashed box highlights the implementation of hardware-efficient ansatz and the notation `$L$' denotes the number of layers. (c) \textsc{Scaling behavior of $C$}. The simulation results of the average gradient norm $C$ with varied qubit count $N$ and layer number $L$. The vertical bar refers to the variance of the gradient norm.      }} 
\label{fig:scaling-trunc-set}
\end{figure}

The scaling behavior is visualized in Fig.~\ref{fig:scaling-trunc-set}(a), where $p$ exhibits two distinct phases: in Phase I, for small truncation values $\Lambda$, $p$ scales exponentially with both $\Lambda$ and the input dimension $d$; in Phase II, as $\Lambda$ increases further, the growth of $p$ begins to saturate and eventually converges to a fixed value. Notably, such a scaling imposes a substantial computational burden for large $d$ and $\Lambda$. As shown in the plot, when both $d$ and $\Lambda$ are large, the computational overhead becomes unaffordable, e.g., when $d=50$ and $\Lambda=10$, $p$ is above $10^{12}$.

\smallskip
\noindent{\underline{The relation between $\Lambda$ and the gradient norm $C$}}. 
The above analysis indicates that the practical efficiency of the proposed ML model lies in the small-$\Lambda$ regime, especially for the large $d$. According to the relation between $\Lambda$ and the gradient norm $C$ in Theorem~2, this is equivalent to requesting a small value about the ratio $4C/\epsilon$. In this regard, the problem reduces to determining how the average gradient norm $C$ scales in practice.

We perform the following numerical simulations to investigate the scaling behavior of the average gradient norm with $C=\mathbb{E}_{\bx\sim [-\pi, \pi]^d}\|\nabla_{\bx}\Tr(\rho(\bx)O)\|_2^2$. Specifically, the class of states $\{\rho(\bx)\}$ is prepared by applying an $N$-qubit hardware-efficient ansatz $U(\bx)$ to an $N$-qubit  product state $\ket{+}^{\otimes N}$ with $\ket{+}=(\ket{0}+\ket{1})/\sqrt{2}$, i.e., 
\begin{equation}
\left \{\rho(\bx) = U(\bx) (\ket{+}\bra{+})^{\otimes N} U(\bx)^{\dagger}\big| \bx\in [-\pi, \pi]^d\right\} \nonumber
\end{equation}
and 
\begin{equation}
		U(\bx) = \prod_{l=1}^L\Big(\bigotimes_{j=1}^N \RZ(\bx_{j*(l-1)})\RY(\bx_{j*(l-1)+1})\RZ(\bx_{j*(l-1)+2})\Big) u_e,
\end{equation}
where $L$ denotes the number of layers, $u_e$ refers to the entanglement layer consisting of $N$ CNOT gates applied sequentially to adjacent qubits, and the dimension of the input yields $$d=3 * N * L.$$ The observable interacted with $\rho(\bx)$ is fixed to be $O=(X_1 + Y_1 + Z_1)/3$, where $X_1$ (or $Y_1$, $Z_1$) refers to apply the Pauli operator $X$ (or $Y$, $Z$) to the first qubit. An illustration of the employed quantum circuit is depicted in Fig.~\ref{fig:scaling-trunc-set}(b).

The hyper-parameter settings are as follows. The number of qubits $N$ varies from 2 to 14, and for each specified qubit count, the number of layers $L$ is varied from 5 to 500. For each configuration, we uniformly and randomly sample the input 500 times to estimate the average gradient norm $C$ and its variance.

The scaling behavior of the average gradient norm is illustrated in Fig.~\ref{fig:scaling-trunc-set}(c). An important observation is that, for all layer numbers $L$, the average gradient norm $C$ converges to a small value with the negligible variance as the qubit count $N$ increases. When $N=14$, the value of C with L=50 is around $0.13 \pm 10^{-6}$. Although the concrete value of $C$ depends on the employed ansatz $U(\bx)$ and the qubit count $N$, the obtained simulation results support a reasonable conjecture that $C\ll 1$ as the number of qubits increases.  

The above findings validate the potential of the proposed ML model to predict interested properties of large-qubit circuits. According to the relation $ \Lambda=4C/\epsilon$  and the scaling behavior of $\Lambda$ exhibited in Fig.~\ref{fig:scaling-trunc-set}, we conclude that the primary factor influencing the applicability of the proposed ML model is the error threshold $\epsilon$. That is, when the given task requires high-precision predictions such that $\Lambda=4C/\epsilon$ yields a very large constant, the proposed ML model will encounter computational issues; otherwise, the proposed ML model can be efficient.

\subsection{Numerical simulations of $N$-qubit rotational GHZ states}\label{append:subsec:sim-GHZ}
\noindent \textbf{Dataset construction of  $N$-qubit rotational GHZ states}. The mathematical form of the $N$-qubit rotational GHZ states is
\begin{eqnarray}
&& \ket{\text{GHZ}(\bx)} = \left(\RY_1(\bx_1) \otimes \RY_{N/2}(\bx_2) \otimes  \RY_{N}(\bx_3) \right)\frac{\ket{0 \cdots 0} + \ket{1 \cdots 1}}{\sqrt{2}}.
\end{eqnarray}
The circuit implementation of this class of quantum states is visualized in Fig.~\ref{supple:fig:GHZ_sim}(a).

\begin{figure}[h!]
	\centering
	\includegraphics[width=0.98\textwidth]{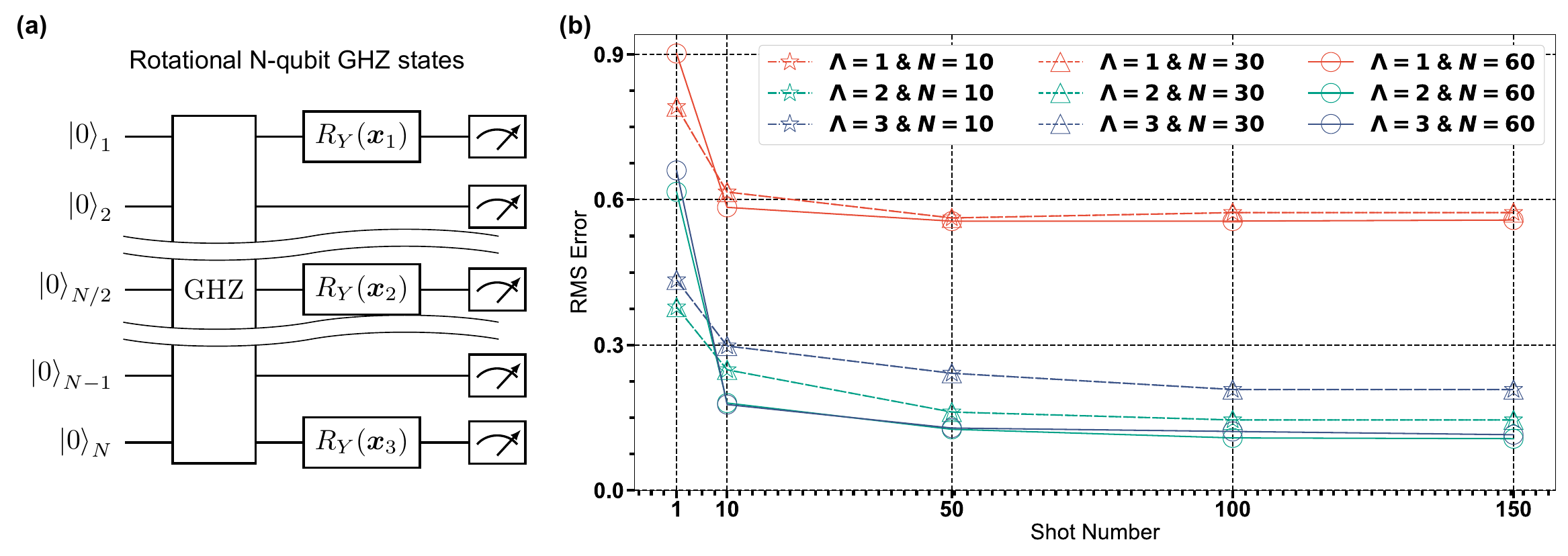}
	\caption{\justifying\small{\textbf{More simulation results of rotational N-qubit GHZ states}. (a) \textsc{Circuit implementation of the rotational $N$-qubit GHZ states}. The three $\RY$ gates apply to the first qubit, the $(N/2)$-th qubit, and the last qubit. (b) \textsc{Prediction error}. The root mean squared (RMS) error of the trained ML model with varied truncation $\Lambda$, the shot number $T$, and the qubit count $N$.}}
	\label{supple:fig:GHZ_sim}
\end{figure}
In both subtasks presented in the main text, the classical shadow of each example $(\bxi, \tilderho_T(\bxi))$ in the training dataset $\mathcal{T}$ are acquired by PastaQ library \cite{pastaq}.  In the first subtask of two-point correlation estimation, the exact value is obtained by matrix product operators (MPO) provided by PastaQ \cite{pastaq}. In the second subtask of predicting expectation values on $Z_1\otimes Z_{N}$, the accurate results yield
\begin{eqnarray}
 && \langle \text{GHZ}(\bx)|Z_1 \otimes  Z_{N}| \text{GHZ}(\bx)\rangle  \nonumber\\
 = &&   -\sin(\bx_1)\sin(\bx_2)\cos(\bx_3) -\cos(\bx_1)\sin(\bx_2)\sin(\bx_3) + \sin(\bx_1)\cos(\bx_2)\sin(\bx_3) + \cos(\bx_1)\cos(\bx_2)\cos(\bx_3).
\end{eqnarray}

\smallskip
\noindent\textbf{Hyper-parameter settings.} The random seed to collect training examples and test examples is set as $1234$ and $123$, respectively. The hyper-parameter of MPO used to calculate the exact values of the two-point correlation of rotational GHZ states is as follows. The cutoff value is set as $10^{-8}$ and the max dimension is set as $50$. 

\smallskip
\noindent\textbf{Cardinality of the frequency set $\mathfrak{C}(\Lambda)$}. In the main text, we adopt three settings of the maximum frequency length, i.e., $\Lambda=1, 2, 3$, to evaluate the performance of the proposed ML model. The corresponding cardinality of the frequency set is $|\mathfrak{C}(\Lambda=1)|=7$, $|\mathfrak{C}(\Lambda=2)|=19$, and $|\mathfrak{C}(\Lambda=3)|=27$, respectively. The similar performance between $\Lambda=2$ and $\Lambda=3$ (full expansion) in Fig.~2 indicates that truncating the high-frequency terms does not apparently affect the capability of the proposed ML model.

\smallskip
\noindent\textbf{Prediction error versus varied number of qubits $N$}. We append more simulation results about how the prediction error of the proposed ML model depends on the shot number $T$, the truncation  $\Lambda$, and the qubit counts $N$. To be specific, we fix the number of training examples to be $n=500$ and collect these training examples under different qubit count $N$, where the maximum shot number is set as $T=150$. Fig.~\ref{supple:fig:GHZ_sim}(b) visualizes the root mean squared (RMS) prediction error under different settings on $10$ test examples. The achieved results indicate that the performance of the proposed ML model is dominated by the truncation number $\Lambda$ (i.e., the dimension of classical inputs $d$) and not sensitive to the number of qubits, which echoes our theoretical analysis. More precisely, when $\Lambda\geq 2$ and $T\geq 100$, the prediction error attains a very low value for both $N=10, 30, 60$.    

\subsection{More details of synthetic global Hamiltonian simulation}\label{append:subsec:global-ham-res}

Here we conduct a much more difficult task compared to the one shown in the main text. In particular, we apply the proposed ML model to predict the magnetization with $\braket{\bar{Z}}$ when $d=30$. The number of training examples and the snapshot for each example are set to $n=50000$ and $T=50$, respectively. In this scenario, full expansion becomes computationally infeasible, where the cardinality of the frequency set is $3^{30}$.  The corresponding cardinality of the frequency set is $|\mathfrak{C}(\Lambda=1)|=61$, $\mathfrak{C}|(\Lambda=2)|=1801$,  $|\mathfrak{C}(\Lambda=3)|=472761$, and respectively. When evaluating the statistical performance of the proposed ML model with $n<50000$, we sample a subset from the whole dataset using different random seeds. The random seeds for all relevant simulations are set as $12345$, $22222$, $33333$, $44444$, and $55555$, respectively.

\smallskip
\noindent \textbf{RMS prediction error}. Fig.~\ref{fig:Ham-sim-res} demonstrates the RMS prediction error of our proposal with the varied number of training examples, i.e., $n\in \{10, 10^2, 10^3, 10^4\}$.  An immediate observation is that with the increased $n$, the averaged prediction error continuously decreases for all settings of $\Lambda$. In addition, the achieved results hint that the proposed model with a larger truncation value $\Lambda$ may require more training examples to surpass the one with a smaller truncation value. For instance, the proposed ML model with $\Lambda=2$ attains a better performance compared to $\Lambda=1$ when $n=10^4$. It is expected that the model with $\Lambda=3, 4$ can attain better performance by further increasing $n$.

\begin{figure}
	\centering
	\includegraphics[width=0.5\textwidth]{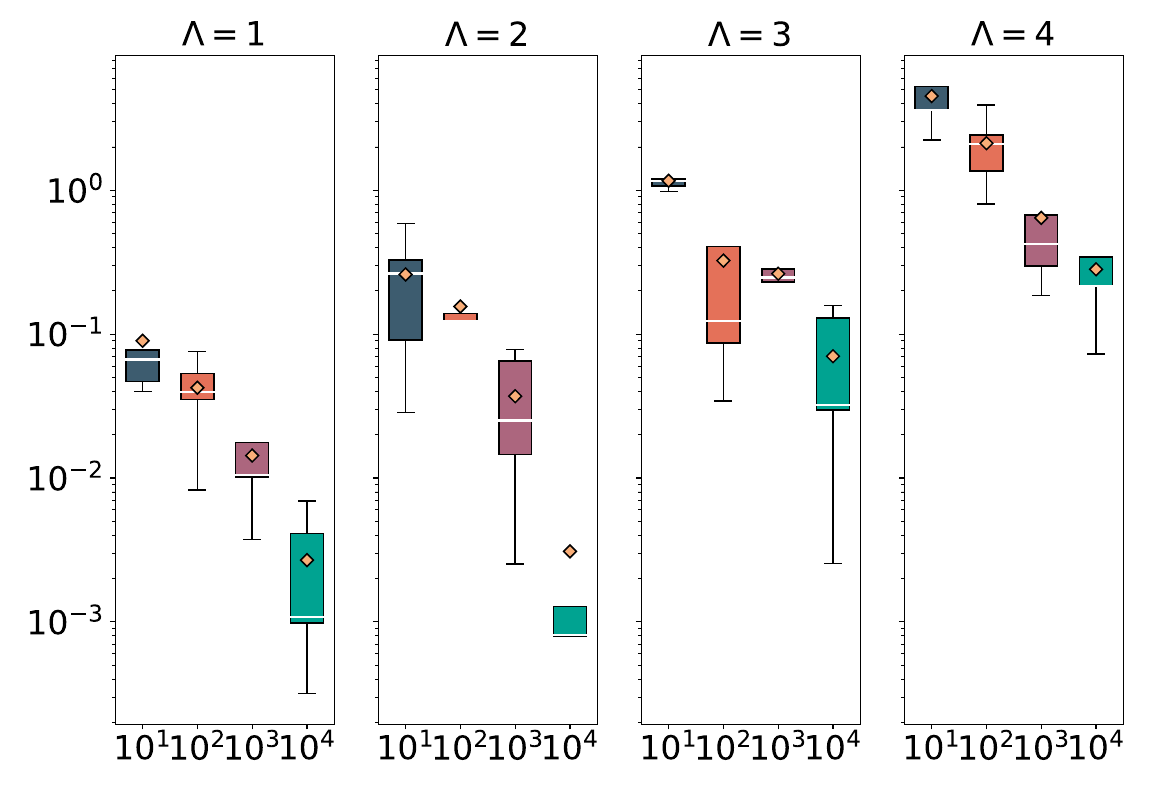}
	\caption{\justifying\small{\textbf{Prediction error on the state evolved  by $60$-qubit global Hamiltonians with $d=30$ with different sizes of the training dataset}.  The meaning of notations in the box plot is as follows. The label `$\Lambda=a$' stands for that the truncation value is set as $a$. The $x$-axis refers to the varied number of training examples $n$, ranging from $n=10$ to $n=10^4$. The $y$-axis refers to the achieved RMS prediction error. All results are obtained using the same random seeds.  }}
	\label{fig:Ham-sim-res}
\end{figure} 

 For completeness, we also append the corresponding standard error in Table~\ref{tab:avg-std-global-ham-sim}. Particularly, for all settings with   $\Lambda\in \{1,2,3,4\}$, the averaged RMS prediction error and the corresponding standard error continuously decrease with the increased number of training examples.

\begin{table}[h]
\centering
\caption{\justifying\textbf{Simulation results of predicting the magnetization of states evolved by the 60-qubit global Hamiltonian.}}
\label{tab:avg-std-global-ham-sim}
\resizebox{0.8\textwidth}{!}{%
\begin{tabular}{|c|c|c|c|c|c|}
\hline
            & $n=100$             & $n=250$             & $n=500$             & $n=750$             & $n=1000$            \\ \hline
$\Lambda=1$ & $0.0424 \pm 0.0222$ & $0.0210 \pm 0.0126$ & $0.0130 \pm 0.0109$ & $0.0152 \pm 0.0117$ & $0.0143  \pm 0.0088$ \\ \hline
$\Lambda=2$ & $0.1556 \pm 0.0959$ & $0.1862 \pm 0.0933$ & $0.0729 \pm 0.0472$ & $0.0729 \pm 0.0472$ & $0.037 \pm 0.0293$  \\ \hline
$\Lambda=3$ & $0.3244 \pm 0.3473$ & $0.3736 \pm 0.3249$ & $0.4012 \pm 0.2055$ & $0.3502 \pm 0.1357$ & $0.2631 \pm 0.1003$  \\ \hline
$\Lambda=4$ & $2.1227 \pm 1.061$  & $0.9519 \pm 0.7148$ & $0.9678 \pm 0.6590$ & $0.7926 \pm 0.4887$ & $0.6432 \pm 0.5222$ \\ \hline
\end{tabular}%
}
\end{table}

\smallskip
\noindent\textbf{The role of shot number $T$}. Here we conduct numerical simulations to explore how the shot number $T$ influences the prediction error. Specifically, we fix the number of training examples as $n=500$ and the truncation value as $\Lambda=1$, but vary the shot number from $T=50$ to $T=500$.  Each setting is repeated five times to collect the statistical results, where the randomness stems from distilling different training datasets from $600$ training examples.  The simulation results are summarized in Table.~\ref{tab:sim-global-vary-shots}, where the averaged prediction error does not decrease with the increased number of measurements. This phenomenon echoes our theoretical analysis, suggesting that once the shot number exceeds a certain threshold, it does not heavily affect the performance of the proposed model.

\begin{table}[h!]
\centering
\caption{\justifying\textbf{Simulation results of predicting the magnetization of states evolved by the 60-qubit global Hamiltonian with varied shot numbers}.}
\label{tab:sim-global-vary-shots}
\resizebox{0.45\textwidth}{!}{%
\begin{tabular}{|c|c|c|}
\hline
                     & $T=50$               & $T=500$             \\ \hline
$n=500$ \& $\Lambda=1$ & $0.0130  \pm 0.0109$ & $0.0251 \pm 0.0038$ \\ \hline
\end{tabular}%
}
\end{table}

\subsection{More simulation results of enhanced variational quantum algorithms by pretraining}\label{appendix:subsec:enhanced-VQA}

As shown in SM~\ref{append:subsec:A:generality}, a major application of the proposed ML model is enhancing variational quantum algorithms by substantially reducing the quantum resource demands. In this subsection, we provide a comprehensive explanation of the corresponding algorithmic implementation and then proceed to numerical simulations that demonstrate the efficacy of our approach. 

\medskip
\noindent\textbf{Algorithmic implementation}. Recall that contemporary quantum devices encounter limitations such as connectivity, gate fidelities, and coherence time. To overcome these constraints, experimentalists typically adopt tailored ansatzes  constructed from the native gate set from the specified quantum device to implement various variational quantum algorithms \cite{kandala2017hardware,havlicek2018supervised}. Essentially, these ansatzes maintain the consistent gate layouts while adjusting parameters for various computational tasks, as delineated by the framework formulated in Eq.~(3) of the main text. This intrinsic relation warrants the use of the proposed ML model to enhance plenty of variational quantum algorithms associated with device-specific or task-specific ansatzes. 

The learning framework is summarized in Alg.~\ref{alg:offline-VQA}. For a given quantum device and a specified ansatz, the learner first constructs the training dataset following the outlined procedure in the main text. Once the dataset is prepared, the proposed ML model serves as a surrogate to optimize variational quantum algorithms in which the optimization correspond to minimize the output value of the surrogate, obviating the need for direct implementation on the specified quantum device. Since the optimization process (i.e., the third step) is conducted entirely on classical processors, it significantly reduces the overhead associated with accessing sparse quantum devices in the contemporary era. 
\begin{algorithm}
\caption{Optimizing variational quantum algorithms in an offline manner}\label{alg:offline-VQA}
1. (Dataset construction) Randomly generate classical inputs $\bx\sim [-\pi, \pi]^{d}$, feed it to the specified $N$-qubit quantum device to obtain the pre-measured state $U(\bx)\ket{0}^{\otimes N}$, and apply Pauli-based classical shadow with $T$ shots to constitute a single training example $(\bx, \tilderho_T(\bx))$\;
\smallskip
2. (Dataset construction) Repeat the above procedure $n$ times to construct the training dataset $\mathcal{T}_{\mathsf{s}}$ \;
\smallskip
3. (Downstream task optimization) Formalize the prediction model $h_{\mathsf{s}}$ following Eq.~(6) and use it to optimize various downstream tasks without access to the quantum processor.
\end{algorithm}

In what follows, we first briefly explain how our approach supports two crucial classes of variational quantum algorithms: the variational quantum Eigen-solvers (VQEs) and quantum neural networks (QNNs). Then, in SM~\ref{append:subsubsec:3-qubit-VQA}, we present numerical simulations demonstrating the application of the proposed pretraining strategy to enhance a 3-qubit ansatz for both ground-state energy estimation and binary classification tasks. To demonstrate the potential of the proposed pretraining strategy at the use-case level, we employ it to enhance VQEs when processing a class of transversefield Ising models (TFIMs) up to 50 qubits in SM~\ref{append:sec:pretrain-TFIM}. In addition,  we conduct numerical simulations to illustrate how the proposed pretraining strategy can advance VQEs when handling a class of $\mathsf{H}_5$ molecules in SM~\ref{append:sec:pretrain-H5}. Note that the learning scheme presented below can be readily extended to wide applications covered by variational quantum algorithms.

\noindent\underline{Remark}. The focus on TFIMs and $\mathsf{H}_5$ molecules stems from the fact that they are widely studied in the literature, garnering  attention in quantum many-body physics and quantum chemistry, spanning algorithmic developments to experimental investigations. In particular, TFIM serves as a prototypical model for understanding quantum phase transitions, non-equilibrium dynamics, and critical phenomena in condensed matter systems~\cite{Dutta2015Quantum}. Its simplicity, coupled with its rich physics, makes it a cornerstone for testing quantum algorithms~\cite{dborin2022simulating,pelofske2024short} and benchmarking quantum simulators~\cite{miessen2024benchmarking,kim2023evidence}. In addition, the hydrogen molecule (e.g., $\mathsf{H}_2$ and $\mathsf{H}_5$) represents a small yet nontrivial system that captures essential features of electron correlation and quantum entanglement in molecular systems~\cite{mcardle2020quantum}. It is widely used to benchmark quantum computational methods in solving quantum chemistry problems with near-term quantum devices \cite{colless2018computation,liu2023performing}.

\medskip
\noindent\underline{Ground state energy estimation by VQE}. Here we first briefly recap the mechanism of VQE when applied to estimate the ground state energy of a specified Hamiltonian. Given an $N$-qubit Hamiltonian $
	\mathfrak{H} \in \mathbb{C}^{2^N \times 2^N}$, the ground state energy estimation aims to find its minimum eigenvalue, i.e., 
 \begin{equation}
 	\mathfrak{E}^* = \min_{\ket{\psi}\in \mathbb{C}^{2^N}}  \braket{\psi|\mathfrak{H}|\psi}.
 \end{equation}
 
To estimate the ground state energy $\mathfrak{E}^*$, VQE adopts an ansatz $W(\btheta)$ to prepare a variational quantum state $\ket{\psi(\btheta)}=W(\btheta)\ket{\psi_0}^{\otimes N}$ with a fixed input state $\ket{\psi_0}$. The trainable parameters $\btheta$ are optimized by minimizing the loss function		$\mathcal{L}(\btheta, \mathfrak{H}) = \Tr\left( (\ket{\psi_0}\bra{\psi_0})  W(\btheta)^{\dagger} \mathfrak{H} W(\btheta)\right)$.      
	The optimization of VQE follows an iterative manner, i.e., the classical optimizer continuously leverages the output of the quantum circuits to update $\btheta$ and the update rule is 
	\begin{equation}\label{eqn:VQE-loss}
		\btheta^{(t+1)}=\btheta^{(t)}-\eta \frac{\partial \mathcal{L}(\btheta^{(t)}, \mathfrak{H})}{\partial \btheta},
	\end{equation}
	 where $\eta$ refers to the learning rate. The first-order gradient in the equation can be obtained by parameter shift rule \cite{schuld2019evaluating}.  Mathematically, the derivative with respect to the $k$-th parameter $\bx_k$ for $\forall k\in[d]$ is 
	\begin{equation}
	\label{append:eqn:para_shift}
			\frac{\partial \mathcal{L}(\btheta^{(t)}, \mathfrak{H})}{\partial \btheta_k} =  \frac{1}{2\sin \alpha} \left[\Tr\left((\ket{\psi_0}\bra{\psi_0})W(\btheta^{(t,+)}) \mathfrak{H} W(\btheta^{(t,+)})^{\dagger}\right)  
			- \Tr\left( (\ket{\psi_0}\bra{\psi_0}) W(\btheta^{(t,-)}) \mathfrak{H} W(\btheta^{(t,-)})^{\dagger}\right)\right],
	\end{equation}
	where $ \btheta^{(t,\pm )} = \btheta^{(t)} \pm  \alpha \bm{e}_{k}$, $\bm{e}_{k}$ is the unit vector along the $\btheta_{k}$ axis, and $\alpha$ can be any real number but the multiple of $\pi$ because of the diverging denominator.

Supported by the proposed ML model, the optimization of $\btheta$ can be entirely carried out on classical processors. More precisely, define $W(\btheta)=U(\bx)$ (see Fig.~\ref{fig:append:visual-vqa-offline} for the visual interpretation), the two terms of the derivatives in  Eq.~(\ref{append:eqn:para_shift}) can be predicted by $h_{\mathsf{s}}$ in Alg.~\ref{alg:offline-VQA} and the optimized parameters can be obtained in an offline manner. For convenience, in the subsequent context, we dub the VQEs optimized by the proposed ML model $h(\bx)$ as \textit{Offline-VQE}. An alternative scenario involves utilizing the optimized parameters of Offline-VQE as effective initial parameters for running VQEs on real quantum devices. In either case, the offline optimization enabled by the proposed ML model dramatically reduces the demand for quantum resources, a critical advantage given the scarcity of available quantum processors in the current landscape.

\begin{figure*}
	\centering
\includegraphics[width=0.98\textwidth]{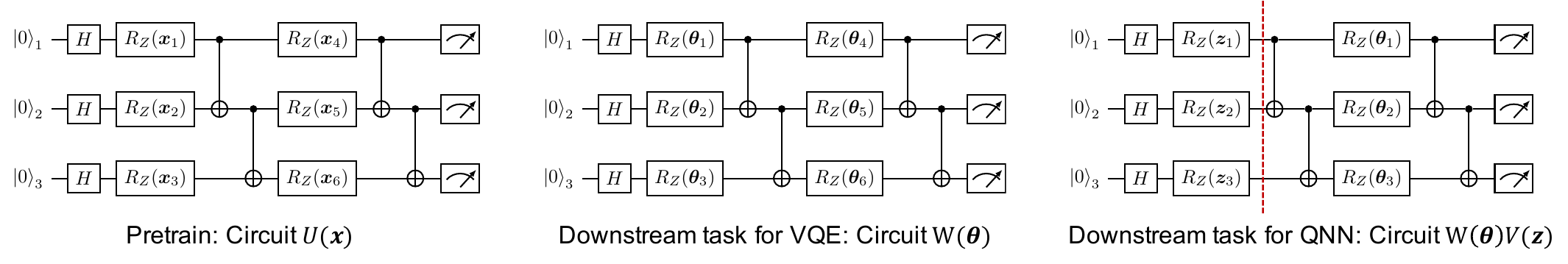}
\caption{\justifying\small{\textbf{A toy model of optimizing variational quantum algorithms in an offline manner}. Left panel: In the pre-training state, the learner collects training data from a $3$-qubit device with a hardware-efficient ansatz  (HEA) $U(\bx)$, which contains $d=6$ tunable parameters. Then the learner can use the collected data to train an ML model introduced in the main text to enhance various downstream tasks. Middle panel: The optimization of variational quantum Eigen-solvers (VQEs) associated with the ansatz $W(\bm{\theta})$ can be accomplished by the proposed ML model without access to the quantum resources. Right panel: The optimization of quantum neural networks (QNNs) associated with the encoding circuit $V(\bz)$ and the trainable circuit $W(\bm{\theta})$ can be accomplished by the proposed ML model without the access to the quantum resources.   }}
\label{fig:append:visual-vqa-offline}
\end{figure*}

\medskip
\noindent\underline{Binary classification  by QNN}. Denote the input space as $\mathcal{Z}$, the binary label (class) space as $\mathcal{Y} = \{0, 1\}$, and the training dataset as $\mathcal{D}= \{(\bzi, y^{(i)})\}_{i=1}^{n}$ with  samples  drawn i.i.d. from an unknown probability distribution $\mathbb{D}$ on $\mathcal{Z}\times \mathcal{Y}$.  The purpose of a binary classification algorithm is using $\mathcal{D}$ to infer a hypothesis (a.k.a., a binary classifier) $g_{\mathcal{D}}:\mathcal{Z} \rightarrow \mathbb{R}$ from the hypothesis space to separate training examples from binary classes. This is equivalent to identifying an optimal hypothesis minimizing the expected risk $\ROPT (g)=\mathbb{E}_{(\bz, \by)\sim \mathbb{D}}[\ell(g(\bz), y)]$, where $\ell(\cdot, \cdot)$ is the per-sample loss and for clarity we specify it as the square error \cite{bishop2006pattern}. Unfortunately, the inaccessible distribution $\mathbb{D}$ forbids us to assess the expected risk directly. In practice, the binary classification algorithm  alternatively learns an empirical classifier $\hath$ as the global minimizer of the (regularized) loss function 
 \begin{equation}\label{eqn:gene_obj_func}
  \mathcal{L}(g, \mathcal{D}) =  \frac{1}{n}\sum_{i=1}^{n}  \frac{1}{2} \left( g(\bzi) -  y^{(i)} \right)^2  + \mathfrak{E}(g),
\end{equation}
where $\mathfrak{E}(g)$ is an optional regularizer. Given an unseen example $\bz'$, its predicted label is $0$ if $g(\bz')<0.5$; otherwise, its predicted label is $1$.  
 
When QNN is employed to implement the binary classifier, the hypothesis $g$ is realized by variational quantum circuits followed by a predefined measurement operator $\Pi_0$ \cite{du2023problem,gil2024understanding}. The mathematical expression of the  hypothesis space for the binary quantum  classifier (BiQC) is 
\begin{equation}
	\mathcal{G}= \left\{g(\bz, \btheta)=\Tr\left(V(\bz) (\ket{0}\bra{0})^{\otimes N} V(\bz)^{\dagger} W(\btheta)^{\dagger} \Pi_0 W(\btheta) \right) \Big| \btheta \in \Theta \right\},
\end{equation} 
where $W(\btheta)$ is the trainable circuit and $V(\bz)$ is another variational circuit that encodes the training example $\bzi$ into the quantum state $\rho(\bzi)$. As with VQE, the optimization of trainable parameters $\btheta$ can be completed by the gradient descent methods (e.g., stochastic gradient descent or batch gradient descent) and the derivatives can be acquired by the parameter shift rule given in Eq.~(\ref{append:eqn:para_shift}). 

\smallskip	
Supported by the proposed ML model,  the optimization of $\btheta$ for BiQC can be entirely carried out on classical processors. According to Alg.~\ref{alg:offline-VQA}, the classical input $\bx$ should be divided into parts, where the first part refers to the training example $\bz$ and the second part refers to the trainable parameters $\btheta$ in BiQCs. In addition, the circuit architecture of $U(\bx)$ amounts to the combination of $W(\btheta)$ and $V(\bz)$ (see Fig.~\ref{fig:append:visual-vqa-offline} for the visual interpretation). In this way, for an arbitrary input $\bz$ and trainable parameters $\btheta$, the derivatives $\partial \mathcal{L}/\partial \btheta_k$, which is formed by $g(\bz, \btheta^{(t)})$ and $g(\bz, \btheta^{(t, \pm)})$, can be estimated by the proposed ML model without accessing quantum resources. Due to this offline property, we call BiQCs optimized by the proposed ML model as \textit{Offline-BiQCs}.

\smallskip
\noindent\underline{Remark}. The approach of leveraging the proposed ML model to propel BiQCs can be readily extended to more complex scenarios. Specifically, the proposed ML model can enhance the optimization of multi-class quantum classifiers \cite{du2023problem} and quantum regression models \cite{mitarai2018quantum}, significantly reducing the quantum resource overhead. Moreover, it can also improve QNNs with diverse architectures, such as data re-uploading strategy \cite{perez2020data} and quantum convolutional neural networks \cite{cong2019quantum,herrmann2022realizing}.

\subsubsection{Pretraining a 3-qubit HEA for both VQE and QNN}\label{append:subsubsec:3-qubit-VQA}

Recall that contemporary quantum devices encounter limitations such as connectivity, gate fidelities, and coherence time. To overcome these constraints, experimentalists typically adopt hardware-efficient ansatzes (HEAs) constructed from the native gate set from the specified quantum device to implement various variational quantum algorithms \cite{kandala2017hardware,havlicek2018supervised}. Essentially, these HEAs  are tailored to individual quantum devices, maintaining consistent gate layouts while adjusting parameters for various computational tasks, as delineated by the framework formulated in Eq.~(3) of the main text. This intrinsic relation warrants the use of the proposed ML model to enhance plenty of variational quantum algorithms associated with device-specific HEAs.

 To exhibit the effectiveness of our proposal for enhancing VQEs and QNNs, we employ it to pretrain a 3-qubit hardware-efficient ansatzes (HEAs)~\cite{kandala2017hardware,havlicek2018supervised}, followed by estimating the ground state of the Transverse-field Ising model (TFIM) and classifying a synthetic binary dataset, respectively. The explicit form of the adopted HEA is 
\begin{equation}\label{append:eqn:def-hea}
	U(\bx) = \prod_{l=1}^{3}\left(\CNOT_{2,3}\CNOT_{1,2} \bigotimes_{i=1}^3\RY(\bx_{i+3(l-1)})\right),
\end{equation}
where $\CNOT_{a,b}$ denotes applying $\CNOT$ gate to $a$-th and $b$-th qubits. This ansatz contains in total $d=9$ classical parameters.

 \smallskip
\noindent\underline{Pre-training}. At the stage of dataset construction, we collect in total  $90000$ training examples and the shot number of each training example is set as $T=1000$. The random seed to generate such a dataset is set as $123$. Once the data collection is completed, we form the kernel-based learning model to accomplish the following two downstream tasks on the classical side. 
 
 \smallskip
 \noindent\underline{Downstream task I: ground state energy estimation of TFIM}. The mathematical form of the exploited $3$-qubit one-dimensional (1D) TFIM is $
	\mathsf{H}_{\text{TFIM}} = -0.1 (Z_1Z_2 + Z_2Z_3) + 0.5(X_1 + X_2 + X_3)$. The ground state energy of $\mathsf{H}_{\text{TFIM}}$ is $\mathfrak{E}^*_{\text{TFIM}}=-1.51 \text{Ha}$.
	 
\begin{figure*}
	\centering
\includegraphics[width=0.95\textwidth]{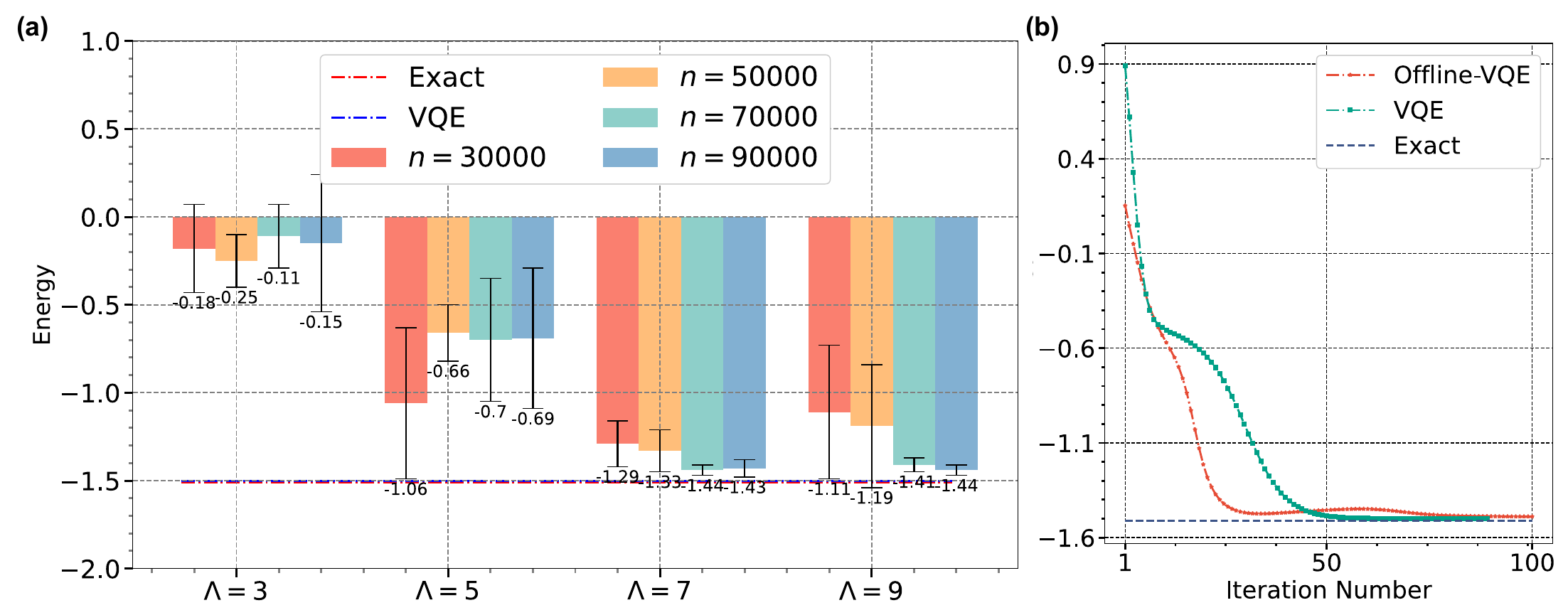}
\caption{\justifying\small{\textbf{Simulation results for ground state energy estimation of TFIM}. (a) \textsc{The estimated energy of Offline-VQE under different hyper-parameter settings}. The label $n=a$ refers that the number of training data is $a$. The label `Exact' denotes the accurate ground state energy of the employed TFIM. The blue dotted line highlights the estimated result of traditional VQE with infinite shots. The vertical bar reflects the standard error of Offline-VQE in each setting. (b) \textsc{The optimization process of Top-1 Offline-VQE and conventional VQE}.}}
\label{append:fig:sim-VQE}
\end{figure*}	
	
	In numerical simulations, we adopt different hyper-parameter settings to evaluate the performance of Offline-VQE. In particular, we vary the number of training examples as $n \in \{ 30000, 50000, 70000, 90000\}$ and set the truncation frequency as $\Lambda\in \{3, 5, 7, 9\}$. The maximum iteration is set as $200$. The initial parameters are uniformly sampled from the range $[-\pi, \pi]^9$. Each setting is repeated $5$ times to obtain the statistical results. For comprehensive, we employ the conventional VQE with the same ansatz, the same initial parameters, the same optimizer, and infinite shots as the benchmark.

The achieved results are depicted in Fig.~\ref{append:fig:sim-VQE}. Fig.~\ref{append:fig:sim-VQE}(a) demonstrates the performance of Offline-VQE under various settings. A key observation is that with the increased frequency truncation and number of training examples, the estimated ground state energy of Offline-VQE converges to the exact result. Namely, when $\Lambda\geq 7$ and $n\geq 70000$, the estimation error is less than $0.1\text{Ha}$. Besides, Fig.~\ref{append:fig:sim-VQE}(b) compares the top-1 Offline-VQE (i.e., $\Lambda=7$ and $n=90000$) with conventional VQE. The achieved results exhibit a similar convergence rate during the optimization, where both of them converge to the near-optimal value after $50$ iterations. These results validate the potential of Offline-VQE in advancing conventional VQEs.

  \smallskip
\noindent\underline{Downstream task II: binary synthetic data classification}. The construction rule of the binary synthetic dataset follows the study \cite{havlicek2018supervised}. That is, we randomly and uniformly sample classical data $\bz$ in the range $[-\pi, \pi]^3$ and then embed them into a $3$-qubit quantum circuit $V(\bz)=\CNOT_{2,3}\CNOT_{1,2} \bigotimes_{j=1}^3\RY(\bz_j)$. Then, we evolve the quantum state $V(\bz)\ket{0}^{\otimes 3}$ by the unitary $W(\btheta^*)=\prod_{l=1}^{2}\left(\CNOT_{2,3}\CNOT_{1,2} \bigotimes_{j=1}^3\RY(\btheta^*_{j+3(l-1)})\right)$ with the fixed parameters $\btheta^*$  and measuring the evolved state under the observable $X\otimes \mathbb{I}_4$. For the input $\bzi$, its label is defined as
\begin{equation}
	y^{(i)} = \text{sign}\left(\Tr\left( W(\btheta^*)V(\bzi)(\ket{0}\bra{0})^{\otimes 3}V(\bzi)^{\dagger} W(\btheta^*)^{\dagger}\left( X\otimes \mathbb{I}_4 \right)\right) \right),
\end{equation}
where $\text{sign}(\cdot)$ denotes the sign function. We collect $500$ positive samples and $500$ negative samples and split them into the training set and the test set with a ratio of $0.2$.

	In numerical simulations, we adopt different hyper-parameter settings to evaluate the performance of Offline-BiQC. The varied settings are analogous to those employed in the offline-VQE, where the number of training examples is set as $n \in \{ 30000, 50000, 70000\}$ and the truncation frequency is set as $\Lambda\in \{3, 5, 7, 9\}$.  The initial parameters are uniformly sampled from the range $[-\pi, \pi]^6$. The learning rate and the maximum iteration are set as $0.5$ and $140$, respectively. For comprehensive, we employ the conventional BiQC with the same encoding and training circuits, the same initial parameters, the same optimizer, and infinite shots as the benchmark.

\begin{figure}
	\centering
\includegraphics[width=0.77\textwidth]{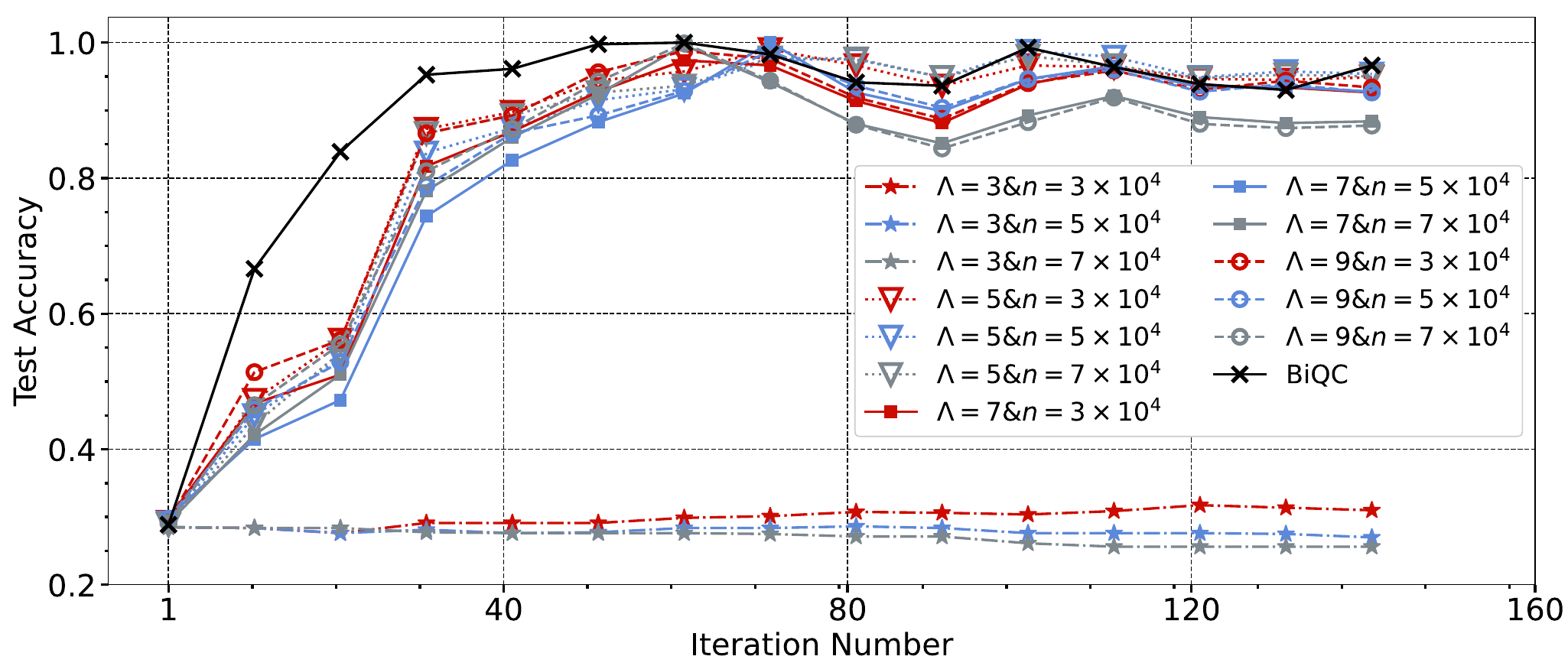}
\caption{\justifying\small{\textbf{Simulation results of Offline-BiQC}. The figure depicts the test accuracy of Offline-BiQC under varied settings. The label `$\Lambda=a$\&$n=b$' denotes that the truncation value is $a$ and the number of training examples used to form the classical representations is $b$. The label `BiQC' refers to the conventional BiQC optimized in an online manner. }}
\label{append:fig:sim-Offline-BiQC}
\end{figure}

The simulation results are shown in Fig.~\ref{append:fig:sim-Offline-BiQC}. When $\Lambda=3$, the test accuracy of Offline-BiQC is lower than $40\%$ no matter how the number of training examples $n$ is. These results reflect the crucial role of the truncation value in warranting the performance of Offline-BiQC. In addition, Offline-BiQC attains a superior performance with an increased truncation value. Namely, when $\Lambda\geq 5$ and $n\geq 30000$, its performance is comparable with conventional BiQC, where both of them attain a test accuracy above $85\%$. These results validate the capability of Offline-BiQC to advance the study of quantum neural networks by greatly reducing the demand for quantum resources. Another phenomenon is that Offline-BiQC and BiQC encounter the oscillated test accuracy after $80$ iterations. This is mainly caused by the employed large learning rate. Adopting a smaller learning rate, associated with the learning rate decay strategy, can mitigate this issue.

\subsubsection{Pretraining Hamiltonian-variational ansatz towards a class of Transverse-field Ising model}\label{append:sec:pretrain-TFIM}

Following the definition of 3-qubit TFIM in SM~\ref{append:subsubsec:3-qubit-VQA}, the mathematical expression of a family of $N$-qubit  (one-dimensional) TFIMs is 
\begin{equation}\label{eqn:TFIM-family}
	\mathcal{H}_{\text{TFIM}}=\Big\{\mathsf{H}_{\text{TFIM}}(J;h)= \sum_{i=1}^{N-1} J Z_iZ_{i+1} + h\sum_{i=1}^N X_i \Big| J , h\in \mathbb{R} \Big\},
\end{equation} 
where $J$ and $h$ refer to the coupling strength and the external field applied in a transverse direction, respectively. In the following, we evaluate the performance of the proposed ML model in pretraining a Hamiltonian-variant ansatz to minimize the ground state energies in $\mathcal{H}_{\text{TFIM}}$.

\medskip
\noindent\underline{Benchmark methods}. To attain a comprehensive understanding about the capabilities of the proposed ML model, we compare its performance with three other types of classical learning surrogates in pretraining VQEs. Additionally, we use conventional VQEs as a benchmark to evaluate the discrepancy between the optimized results and the exact ground-state energy. Below, we outline the fundamental mechanisms of each approach.

\smallskip 
\noindent\textit{Method I: Conventional VQE}. The implementation of VQE mainly follows the rules introduced in Eq.~(\ref{eqn:VQE-loss}). The only difference is that we set the number of measurements to infinity (i.e., $T\rightarrow\infty$) and compute the gradients using the \textit{automatic differentiation method}, i.e., the update of trainable parameters is completed by the Adam optimizer~\cite{kingma2014adam}. The infinite measurements can eliminate the influence of estimation errors, allowing us to isolate and examine how the employed ansatz $U(\bx)$ affects the performance, i.e., the discrepancy between the  optimized results of VQE and the exact ground-state energy. 

\smallskip
\noindent\textit{Method II: The proposed ML model}. The first surrogate applied to the pre-training VQE is the proposed ML model, i.e.,
\begin{equation}\label{eqn:generic-learner}
	h_{\mathsf{s}}(\bx, \mathsf{H}) = \frac{1}{n}\sum_{i=1}^n\kappa_{\Lambda}\left(\bx, \bxi \right)g\left(\bxi, \mathsf{H} \right), \nonumber
\end{equation}
where $\kappa_{\Lambda}(\bx, \bxi)$ is the truncated trigonometric monomial kernel   and $g(\bxi,\mathsf{H})=\Tr(\tilderho_T(\bxi)\mathsf{H})$ refers to the shadow estimation of $\Tr(\rho(\bxi) \mathsf{H})$ with $T$ snapshots for each training example in $\mathcal{T}_{\mathsf{s}}$ (Step 2 in Alg.~\ref{alg:offline-VQA}). 

Given access to the proposed ML model $h_{\mathsf{s}}(\bx, \mathsf{H})$, the initialized parameters of $U(\bx)$ take the form as
\begin{equation}\label{eqn:surrogate-loss-proposed-ML}
	\hat{\bx} = \arg\min_{\bx} h_{\mathsf{s}}(\bx, \mathsf{H}).
\end{equation}
To find the parameters minimizing the surrogate loss $h_{\mathsf{s}}(\bx, \mathsf{H})$, we employ Adam optimizer, an advanced gradient-descent-based method.

\smallskip
\noindent\textit{Method III: kernel ridge regression}. We identify that when the employed ansatz $U(\bx)$ consists of $d$ $\RZ$ gates and an arbitrary number of Clifford gates, the estimated energy in VQE can be reformulated into an inner-product form under the PTM representations, i.e.,
	$\Tr(\mathsf{H} U(\bx) \rho_0 U(\bx)^{\dagger}) = \langle \bmw^*, \vec{\Phi}(\bx) \rangle$, where $\vec{\Phi}(\bx)=\{\Phi_{\bomega}(\bx)\}_{\bomega}$ is the $3^d$-dimensional vector with $\bomega\in\{0,\pm 1\}^d$ and $\Phi_{\bomega}(\bx)$ being the trigonometric monomial basis, and the optimal parameters $\bmw^*=\{2^{\|\bomega\|_0}\Tr(\rho_{\bomega}\mathsf{H})\}_{\bomega}$ is also a $3^d$-dimensional vector with $\rho_{\bomega}$ being parameter-independent matrix.  

This inner-product reformulation  enables us to utilize kernel ridge regression (KRR) \cite{mohri2018foundations} to construct a learning surrogate. Specifically, given the training dataset $\{(\bxi, \yi\}_{i=1}^n$ with $\yi = \Tr(\mathsf{H} U(\bxi) \rho_0 U(\bxi)^{\dagger})$ in the ideal case or $\yi = \Tr(\tilderho_T(\bxi)\mathsf{H})$ in the shadow case for $\forall i\in[n]$, we can use a KRR model $h_{\text{KRR}}(\bx; \bmw)=\langle \bmw, \vec{\Phi}(\bx)  \rangle$ to minimize the empirical risk, i.e.,
	$\min_{\bm{\mathsf{w}}} \frac{1}{2} \sum_{i=1}^n \left(h_{\text{KRR}}(\bxi; \bmw) - \yi \right)^2 + \frac{\alpha}{2} \|\bmw\|_2,$ where $\alpha$ is a hyper-parameter. Under the \textit{dual representation}~\cite{mohri2018foundations}, the closed-form solution of the above minimization problem is
\begin{equation}\label{eqn:KRR-explit-form}
	h_{\text{KRR}}(\bx; \hat{\bmw}) = \bm{k}(\bx)\cdot (K+\alpha\mathbb{I}_n)^{-1}\bm{y},
\end{equation}
where $\hat{\bmw}$ refers to the optimized parameters, $\bm{y}=[y^{(1)},...,y^{(n)}]$ is the label vector, $K\in \mathbb{R}^{n\times n}$ refers to the kernel matrix with 
\begin{equation}\label{eqn:KRR-kernel}
	K_{ij} \equiv k(\bxi, \bx^{(j)}) = \langle \vec{\Phi}(\bxi),  \vec{\Phi}(\bx^{(j)}) \rangle = \prod_{m=1}^d\left(1 + 2\cos(\bxi_m - \bx^{(j)}_m)\right),~\forall i,j\in [n], 
\end{equation}
and $\bm{k}(\bx)$ is an $n$-dimensional vector whose $i$-th entry is $\bm{k}(\bx)_i=\langle \vec{\Phi}(\bx),  \vec{\Phi}(\bx^{(i)}) \rangle = \prod_{m=1}^d(1 + 2\cos(\bxi_m - \bx_m))$.  Given access to the optimized learning model $h_{\text{KRR}}(\bx; \hat{\bmw})$ in Eq.~(\ref{eqn:KRR-explit-form}), the pre-training VQE amounts to solving the following minimization problem, 
\begin{equation}\label{eqn:surrogate-loss-KRR}
	\hat{\bx} = \arg\min_{\bx} h_{\text{KRR}}(\bx, \hat{\bmw}).
\end{equation}
The optimization is completed by the Adam optimizer. After that, the optimized result serves as the initialized parameters of $U(\bx)$ for the corresponding VQE task. 

\smallskip
\noindent\textit{Method IV: Random Fourier features}. The computational efficiency of KRR highly depends on the number of training examples $n$. When the number of training examples $n$ becomes large, the space and time cost of KRR scales with $\mathcal{O}(n^2)$ and $\mathcal{O}(n^3)$,  respectively. To address the computational bottleneck, an alternative is employing random Fourier features (RFF) techniques to learn the given dataset~\cite{rahimi2007random}. More specifically, when the target kernel is \textit{shift-invariant}, i.e., $k(\bx,\bx')=k(\bx-\bx')$, the KRR model $h_{\text{KRR}}(\bx;\bmw)$ can be estimated by RFF-based regression model, i.e.,
\begin{equation}\label{eqn:surrogate-RFF}
	h_{\text{RFF}}(\bx, \bmw) = \langle \bmw, \vec{\Psi}(\bx) \rangle,
\end{equation} 
where $\vec{\Psi}(\bx)= [\cos(\frac{\gamma}{d'} \langle \btheta_1, \bx + b_1\rangle), \sin(\frac{\gamma}{d'} \langle \btheta_1, \bx + b_1\rangle), \cdots, \cos(\frac{\gamma}{d'} \langle \btheta_{d'}, \bx + b_{d'}\rangle), \sin(\frac{\gamma}{d'} \langle \btheta_{d'}, \bx + b_{d'}\rangle)]^{\top}$ is the RFF with $d'$ dimensions, $\gamma$ and $d'$ are tunable hyper-parameters,  $\btheta_i$ for $\forall i\in  [d']$ are $d'$-dimensional vectors sampled from a multivariate standard normal distribution, and $b_i$ for $\forall i \in  [d']$ is scalar sampled from the uniform distribution ranging from $-\pi$ to $\pi$. 

Following Eq.~(\ref{eqn:KRR-kernel}), an immediate observation is that the kernel under consideration is \textit{shift-invariant}, justifying the use of an RFF-based regression model as a surrogate. Specifically, given the training dataset $\{(\bxi, \yi\}_{i=1}^n$, we train the RFF-based regression model by minimizing the empirical risk 
\begin{equation}\label{eqn:RFF}
	\min_{\bm{\mathsf{w}}} \frac{1}{2} \sum_{i=1}^n \left(h_{\text{RFF}}(\bxi; \bmw) - \yi \right)^2 + \frac{\alpha}{2} \|\bmw\|_2,
\end{equation}
where $\alpha$ is a hyper-parameter. Denote the optimized parameters as $\hat{\bmw}$. As before, we use Adam optimizer to locate the solution that minimizes the surrogate loss, i.e., 
\begin{equation}\label{eqn:surrogate-loss-RFF}
	\hat{\bx} = \arg\min_{\bx} h_{\text{RFF}}(\bx, \hat{\bmw}).
\end{equation}
The obtained estimate serves as the initialization parameters of $U(\bx)$ for the corresponding VQEs. 

\smallskip
\noindent\textit{Method V: Multi-layer Perceptron}. We last exploit the performance of deep learning protocols for VQE's parameter initialization. For this purpose, we adopt multi-layer perceptron (MLP) as the deep learning protocol \cite{goodfellow2016deep}, whose mathematical expression is

\begin{equation}\label{eqn:MLP-explit-form}
	h_{\text{MLP}}(\bx; \bm{W}) = W_L\text{ReLU}(...\text{ReLU}(W_2\text{ReLU}(W_1 \bx))),
\end{equation}
where $\bm{W}=\{W_i\}_{i=1}^L$ refers to $L$ trainable weight matrices and the size of $W_i$ for $\forall i \in [L]$ is adjustable. Given the training dataset $\mathcal{T}_{\mathsf{s}}=\{(\bxi, \yi\}_{i=1}^n$, we train MLP by minimizing the empirical risk 
\begin{equation}\label{eqn:train-loss-MLP}
	\min_{\bm{W}} \frac{1}{2} \sum_{i=1}^n \left(h_{\text{MLP}}(\bxi; \bmw) - \yi \right)^2 + \frac{\alpha}{2} \sum_i\|W_i\|_F,
\end{equation}
where $\alpha$ is a hyper-parameter. Given access to the optimized MLP $h_{\text{MLP}}(\bx, \hat{\bm{W}})$, the initialization parameters for $U(\bx)$ are acquired via solving 
\begin{equation}\label{eqn:surrogate-loss-MLP}
	\hat{\bx} = \arg\min_{\bx} h_{\text{MLP}}(\bx, \hat{\bm{W}}).
\end{equation}
For a fair comparison, the optimization of MLP is also completed by the Adam optimizer. 

Note that while the employed MLP lacks provable guarantees, benchmarking its performance is crucial given the growing success of deep learning in quantum system learning tasks.  

\medskip
\noindent\textit{\underline{Metrics}}. For comprehensive, we employ two metrics to evaluate the performance of different benchmarking methods.
 
 \textcircled{\raisebox{-0.9pt}{1}} \textit{Prediction error}. Our first metric concerns the ability of the trained learning models to predict unseen examples. Denote the test dataset as $\{(\bxi, \yi)\}_{i=1}^{n_{\text{test}}}$ such that the test examples are sampled from the same distribution with the training examples. The prediction error refers to the mean square error on the test dataset, i.e., 
\begin{equation}\label{eqn:metric-pred-error}
	\mathsf{R}_{\text{predict}} = \frac{1}{n_{\text{test}}}\sum_{i=1}^{n_{\text{test}}} \left( h_*(\bx) - \yi \right)^2,
\end{equation}   
where $h_*\in \{h_{\mathsf{s}}, h_{\text{KRR}}, h_{\text{RFF}}, h_{\text{MLP}}\}$.

\textcircled{\raisebox{-0.9pt}{2}} \textit{Absolute energy difference}. The second metric compares the difference between the estimated energy $\hat{E}_{\text{VQE}}$ returned by VQE and the exact minimum ground state energy $E_{\text{Exact}}$ of the given Hamiltonian, i.e.,
\begin{equation}
	\mathsf{R}_{\text{abs}} = \left|\hat{E}_{\text{VQE}} - E_{\text{Exact}}\right|.
\end{equation} 
The estimated energy is defined by $\hat{E}_{\text{VQE}}=\Tr(\mathsf{H} U(\hat{\bx}) \rho_0 U(\hat{\bx})^{\dagger})$. Given different parameters $\hat{\bx}$ returned by varied learning surrogates with $h_*\in \{h_{\mathsf{s}}, h_{\text{KRR}}, h_{\text{RFF}}, h_{\text{MLP}}\}$, we can evaluate how these surrogates advance the parameter initialization by comparing $\mathsf{R}_{\text{abs}}$. 
 
\medskip
\noindent\underline{\textit{Problem setup}}. The detailed settings of TFIMs are as follows. We investigate six TFIMs in $\mathcal{H}_{\text{TFIM}}$ in Eq.~(\ref{eqn:TFIM-family}). In particular, these Hamiltonians are constructed by fixing $N=10$, $h=-0.5$, and varying the coupling strength with $J\in \{-0.5, -0.3, -0.1, 0.1, 0.3, 0.5\}$.

\begin{figure}[h!]
	\centering
	\includegraphics[width=0.99\textwidth]{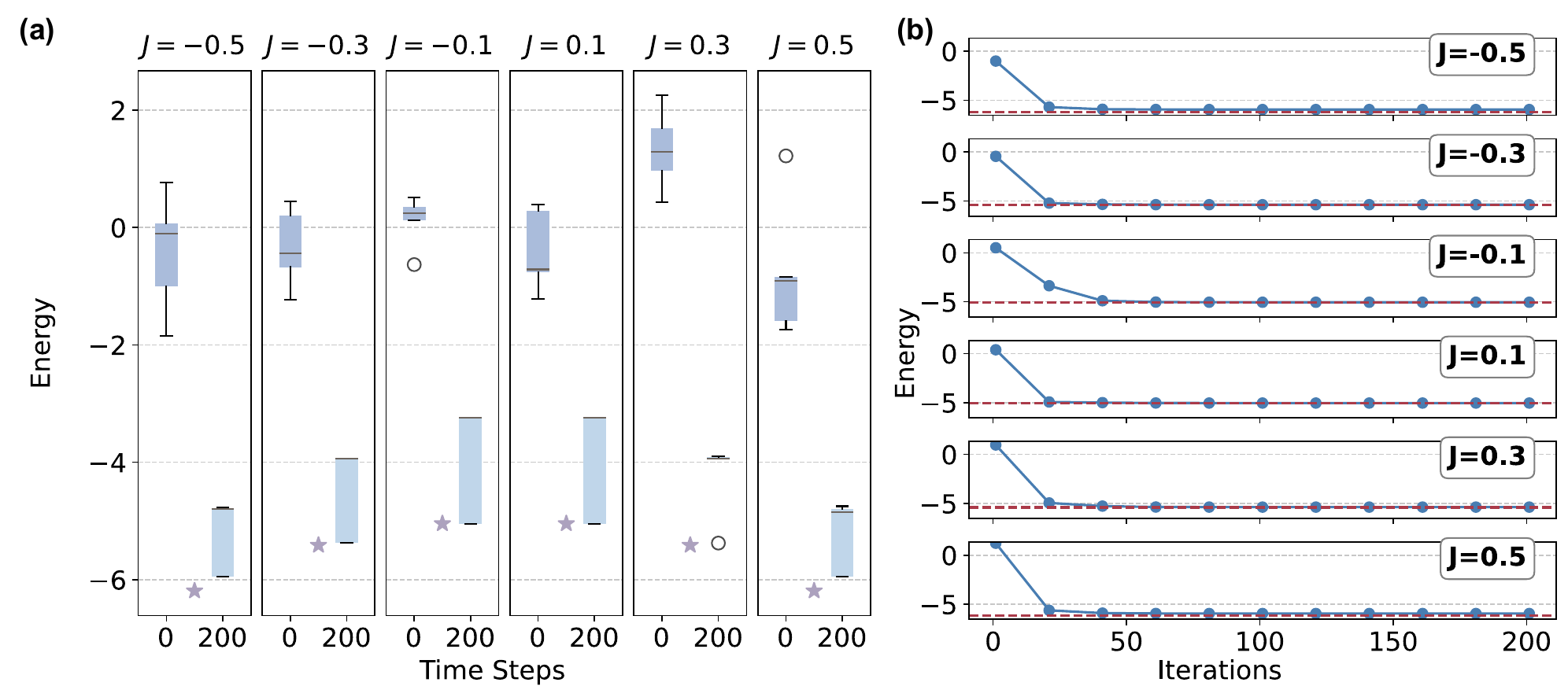}	
	\caption{\justifying\small{\textbf{Performance of conventional VQEs on the family of TFIMs with $N=10$}. (a) \textsc{Statistical performance of conventional VQEs}. The box plots show the distribution of the estimated energies at the first iteration and after 200 iterations for each setting, with $J$ values ranging from $-0.5$ to $0.5$. Each box represents the interquartile range (IQR), with the horizontal line inside indicating the median energy. The whiskers extend to the most extreme data points within 1.5 times the IQR, and outliers beyond this range are shown as individual circles. The purple stars represent the exact ground-state energies. (b) \textsc{Training dynamics of the best VQEs}. The plots show the convergence of the estimated energies toward the exact ground-state energies (dashed red lines) over 200 iterations for the top-performing instance in each setting.}}
	\label{fig:VQE-TFIM}
\end{figure}

Moreover, Hamiltonian-informed ansatz is adopted to estimate their ground state energies, i.e.,
\begin{equation}\label{eqn:TFIM-ansatz}
	U(\bx)= \prod_{i=1}^{N-1} \exp\left(-\imath \frac{\bx_i}{2} (Z_i\otimes Z_{i+1})\right) \prod_{i=N}^{2N-1} \left(-\imath \frac{\bx_i}{2} (X_i)\right) \text{H}^{\otimes N},
\end{equation}
where $\imath$ refers to the imaginary unit and $ \text{H}$ refers to the Hadamard gate. As illustrated in Fig.~\ref{fig:VQE-TFIM}, optimal parameters for this ansatz exist to minimize the six Hamiltonians under investigation. In particular, for each setting,   we perform five repetitions with different initial parameters to collect statistical results, using a learning rate of $0.8$ for the Adam optimizer. Fig.~\ref{fig:VQE-TFIM}(a) illustrates the statistical performance of VQEs. That is, for each configuration, the exact ground-state energy can be accurately estimated after $200$ iterations with \textit{well-initialized parameters}. For clarity, we visualize the training dynamics of the VQE for the top-performing setting in Fig.~\ref{fig:VQE-TFIM}(b). An immediate observation is that across all settings, the absolute error between the estimated energy and the exact value remains small, which are $0.24$, $0.04$, $5.06\times 10^{-4}$, $4.92\times 10^{-4}$, $0.04$, and $0.24$ for $J=-0.5, -0.3, -0.1, 0.1, 0.3, 0.5$, respectively.

The process of pretraining VQEs by classical learning surrogates follows the steps outlined in Alg.~\ref{alg:offline-VQA}. Specifically, in Steps 1\&2, we construct the training dataset $\mathcal{T}_{\mathsf{s}}$ by randomly sampling $\bx\in [-\pi, \pi]^{19}$ from a uniform distribution and inputting these values into the ansatz $U(\bx)$ in Eq.~(\ref{eqn:TFIM-ansatz}) to generate classical shadows (when $N=10$, the dimension of the input $\bx$ is $d=19$). The snapshot for each quantum state $\rho(\bx)$ is fixed to be $T=100$ and the number of trainable examples is set $n=7500$.  

Once the dataset $\mathcal{T}_{\mathsf{s}}$ is collected, we move to Step 3 of Alg.~\ref{alg:offline-VQA}. That is, the dataset $\mathcal{T}_{\mathsf{s}}$ is used to pretrain VQEs by the four classical surrogates to estimate the ground-state energies of the six explored TFIMs. The hyper-parameter settings of each learning surrogate are as follows.
\begin{itemize}
	\item Proposed ML model. For all configurations, we fix the truncation value as $\Lambda=2$ and the number of training examples as $n=7500$ to obtain the corresponding learning surrogate $h_{\mathsf{s}}(\bx;\mathsf{H})$ following Eq.~(\ref{eqn:generic-learner}). In the procedure of minimizing the surrogate loss as defined in Eq.~(\ref{eqn:surrogate-loss-proposed-ML}), we set the learning rate of Adam optimizer as $0.01$ and the total number of iterations as $500$.
		\item Kernel regression model. For all configurations, we set the number of training examples as $n=100$ to avoid the memory issue. The learning models $h_{\text{KRR}}$ are trained following Eq.~(\ref{eqn:KRR-explit-form}). Once the trained models are available under different $J$ values, we follow Eq.~(\ref{eqn:surrogate-loss-KRR}) to minimize their corresponding surrogate losses. As with the proposed ML model, we set the learning rate of Adam optimizer as $0.01$ and the total number of iterations as $500$.
	\item Random Fourier Feature. For all configurations, we set the dimension of random Fourier features as $d'=800$ in Eq.~(\ref{eqn:surrogate-RFF}), to match the feature dimension of the proposed learning models $h_{\mathsf{s}}(\bx;\mathsf{H})$. We follow Eq.~(\ref{eqn:RFF}) to $h_{\text{RFF}}$ under different $J$ values. The number of training examples is set as $n=7500$ and the hyper-parameter is fixed to be $\alpha=1$. In the procedure of minimizing the surrogate loss as defined in Eq.~(\ref{eqn:surrogate-loss-RFF}), we follow the same routine with $h_{\mathsf{s}}(\bx;\mathsf{H})$, where the learning rate of Adam optimizer is $0.01$ and the total number of iterations is $500$.
	\item MLP. For all configurations, the exploited $h_{\text{MLP}}(\bx, \bm{W})$ in Eq.~(\ref{eqn:MLP-explit-form}) refers to a 2-hidden-layer MLP, where the dimension of the input layer, the first hidden layer, the second hidden layer, and the output layer is $19$, $32$, $8$, and $1$, respectively. The total number of trainable parameters of MLP is $872$, comparable to those used in the proposed learning model and KRR. Moreover, each hidden layer is followed by a ReLU activation function. We use $n=7500$ training examples to optimize MLPs under varied $J$ values are as follows. The learning rate of the Adam optimizer is set as $0.001$. The batch size is set as $200$. The hyper-parameter $\alpha$ is set as $0.0001$. Once $h_{\text{MLP}}(\bx, \hat{\bm{W}})$ is accessible, we follow the surrogate loss in Eq.~(\ref{eqn:surrogate-loss-MLP}) to pretrain VQEs. As with the proposed ML model, we set the learning rate of Adam optimizer as $0.01$ and the total number of iterations as $500$. 
\end{itemize} 
\noindent Remark. To assess the impact of noisy labels caused by the finite snapshots of classical shadows in each training example, we also benchmark the performance of the employed learning surrogate under the assumption of an infinite number of shots for each configuration, i.e., $T\rightarrow \infty$, $\tilde{\rho}_T(\bx)=\rho(\bx)$, and $\yi=\Tr(\rho(\bxi)\mathsf{H})$. We repeat each configuration described above five times to collect statistical results. 

\medskip
\noindent\underline{\textit{Numerical results}}. The predictor error of the employed four learning surrogates (i.e., the first metric $\mathsf{R}_{\text{predict}}$ in Eq.~(\ref{eqn:metric-pred-error})) is demonstrated in Fig.~\ref{fig:pred-error-TFIM}, which is evaluated by $500$ test examples. For all subplots, the proposed ML models outperform the rest three learning surrogates under both ideal (highlighted by the solid bar) and noisy (highlighted by the shadowed bar) scenarios. Particularly, the prediction error of the proposed ML model under the ideal (noisy) setting achieves $0.36$ $(0.68)$, $0.13$ $(0.30)$, $0.01$ $(0.10)$, $0.01$ $(0.10)$, $0.13$ $(0.28)$, and $0.36$ $(0.63)$ when $J=-0.5, -0.3, -0.1, 0.1, 0.3, 0.5$, respectively. Another observation is that  while the noisy labels caused by finite snapshots (i.e., $T=100$) can slightly degrade prediction performance of the proposed ML model, the impact is not substantial, with a maximum degradation of no more than $0.32$ across all settings.  

\begin{figure}[h!]
	\centering
	\includegraphics[width=0.99\textwidth]{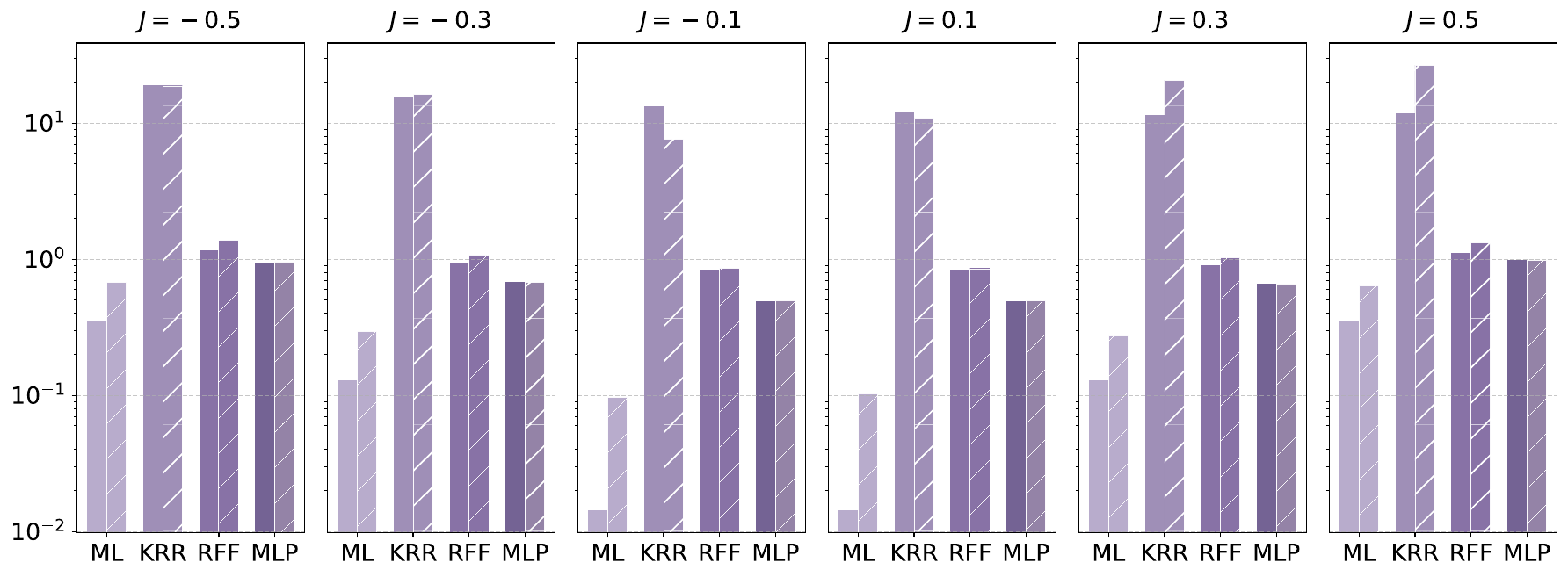}	
	\caption{\justifying\small{\textbf{Prediction error of four learning surrogates for the family of TFIMs with $N=10$}. Each subplot corresponds to the simulation results for a specific value of $J$ in the TFIM. Solid bars represent the prediction error under ideal conditions ($T\rightarrow \infty$ in the training dataset $\mathcal{T}_{\mathsf{s}}$), while shadowed bars with a “\(/\)” pattern denote the noisy case ($T = 100$). The four surrogates are labeled as ML (the proposed ML model), KRR (Kernel ridge regression), RFF (random Fourier features), and MLP (Multi-Layer Perceptron), respectively. }}
	\label{fig:pred-error-TFIM}
\end{figure}

 \begin{figure}[h!]
	\centering
	\includegraphics[width=0.98\textwidth]{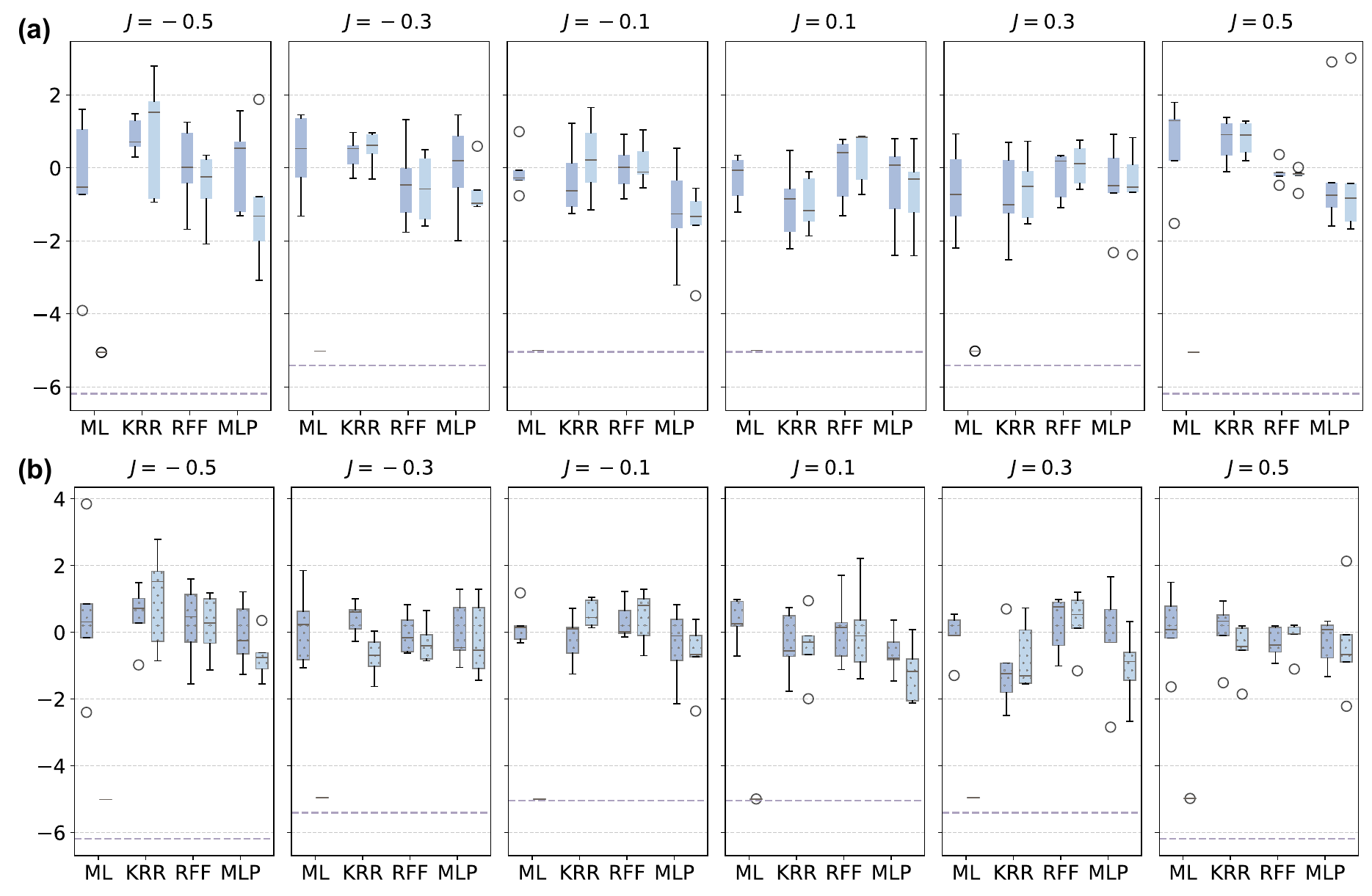}
	\caption{\justifying\small{\textbf{Pretraining performance of different learning surrogates on the family of TFIMs with $N=10$}. The upper panel (a) and the lower panel (b) separately exhibit the simulation results under the ideal and noisy settings ($T\rightarrow \infty$ and $T=100$). Each subplot corresponds to a specific $J$ value in the TFIM, and the box plots depict the estimated ground-state energy distributions for four learning surrogates. For each surrogate, the estimated energies without pre-training and after pretraining are visualized by the deep and light blue colors, respectively. Dashed horizontal lines represent the exact ground-state energies.}}
	\label{fig:pre-train-VQE-TFIM}
\end{figure}

Next, we use the optimized learning surrogates to pretrain VQEs for the six explored TFIMs. The estimated energies, both with and without pretraining by the four learning surrogates, are illustrated in Fig.~\ref{fig:pre-train-VQE-TFIM}. An immediate observation is that  compared to other learning surrogates, the estimated energy obtained from the proposed ML model is greatly closer to the exact result in both ideal and noisy settings. Notably, unlike the proposed ML model, the performance of VQEs does not improve after pretraining with other learning surrogates and, in some cases, may even degrade. These contrastive outcomes further validate the effectiveness of our proposed learning protocol.

\begin{table}[h!]
\centering
\caption{\justifying\small{\textbf{Comparison between conventional VQEs and the proposed ML model on a family of TFIMs with $N=10$}. The table presents the estimated ground state energies for various coupling strengths $J$, as obtained from the conventional VQE and the proposed ML model under both ideal and noisy conditions (labeled by `Pretrain (ideal)' and `Pretrain (shadow)', respectively). The notation `$a \pm b$' indicates that the mean value is $a$ and the variance is $b$. }}
\label{tab:TFIM-vqe-ML}
\resizebox{\textwidth}{!}{%
\begin{tabular}{@{}c|c|c|c|c|c|c@{}}
\toprule
\multicolumn{1}{l|}{} & $J=-0.5$ & $J=-0.3$ & $J=-0.1$ & $J=0.1$  & $J=0.3$  & $J=0.5$  \\ \midrule
VQE & $-5.252\pm 0.327$ & $-4.512 \pm 0.498$ & $-3.960 \pm 0.784$ & $-3.961 \pm 0.784$ & $-4.213 \pm 0.339$ & $-5.26 \pm0.32$ \\ \midrule
Pretrain (ideal)      & $-5.056 \pm 0.000$ & $-5.020 \pm 0.000$  & $-5.002 \pm 0.000$ & $-5.002 \pm 0.000$ & $-5.021 \pm 0.000$  & $-5.057 \pm 0.000$ \\ \midrule
Pretrain (shadow)     & $-5.012 \pm 0.000$ & $-4.959 \pm 0.000$ & $-5.002 \pm 0.000$ & $-4.996 \pm 0.000$ & $-4.951 \pm 0.000$ & $-4.981 \pm 0.000$ \\ \bottomrule
\end{tabular}%
}
\end{table}

Another key observation from the simulation results is that the performance of the proposed ML model remains comparable between the ideal and noisy settings, as illustrated in Figs.~\ref{fig:pre-train-VQE-TFIM}(a) and (b), respectively. This indicates that the noisy labels introduced by the finite snapshots of classical shadows only subtly effect the performance of the proposed ML model. Furthermore, the results suggest that a constant number of snapshots (e.g., $T=100$) is sufficient to achieve strong performance, aligning with the findings of Theorem~2 in the main text.

Finally, we compare the performance of the conventional VQE (using infinite measurements) with the proposed ML model. The statistical results are summarized in Table~\ref{tab:TFIM-vqe-ML}. The first key observation is that the proposed ML model, in both ideal and noisy settings, generally achieves energy estimates that are comparable to or better than those of the conventional VQE, particularly at coupling strengths $J \in \{ -0.3,  -0.1, 0.1,  0.3\}$. The second key observation is the consistency and stability of the proposed ML model's performance across both settings, especially at intermediate coupling strengths, where the VQE exhibits significant variance and large estimation errors. These findings confirm the potential of the proposed ML-based approach to enhance conventional VQE methods, particularly in practical scenarios with constrained quantum resources.

\medskip
\noindent\textbf{Additional information about pretraining VQE on a 50-qubit TFIM.} In the main text, we apply the proposed ML model to pretrain VQE on a 50-qubit TFIM. The explored Hamiltonian, ansatz, optimizers, and metrics are identical to those introduced earlier, except that the qubit size $N$ is increased from 10 to 50. Recall the main plot in Fig.~3(b) of the main text, which illustrates the optimization dynamics of the proposed ML model during the minimization of the surrogate loss in Eq.~(\ref{eqn:surrogate-loss-proposed-ML}). The estimated energies after pretraining are comparable and closely approximate the optimal results, which are $-24.95$, $-24.90$, and $-25.99$ for the ideal dataset, shadow dataset, and the exact value, respectively. Furthermore, for both the noisy and ideal cases,  the variance diminishes as the number of iterations increases. These results highlight  the effectiveness of the proposed ML model in large-scale scenarios, showcasing its potential to facilitate a wide range of practical use-case tasks.

\subsubsection{Pretraining Hamiltonian-variational ansatz towards a class of $\mathsf{H}_5$ molecules}\label{append:sec:pretrain-H5}

Here we apply the proposed method to pretain Hamiltonian-variational ansatz towards a class of $\mathsf{H}_5$ molecules. The detailed settings of $\mathsf{H}_5$ molecules are as follows. We center on six $\mathsf{H}_5$ molecules in $\mathcal{H}_{\mathsf{H}_5}$ with varying bond lengths, i.e., $R\in \{0.5 \mathring{A}, 0.6 \mathring{A}, 0.7 \mathring{A}, 0.8 \mathring{A}, 0.9 \mathring{A}, 1.0 \mathring{A}\}$. Moreover, a simplified ansatz inspired by the unitary coupled-cluster ansatz is exploited to estimate their ground state energies, i.e.,
\begin{equation}\label{eqn:H5-ansatz}
	 \includegraphics[width=0.65\textwidth]{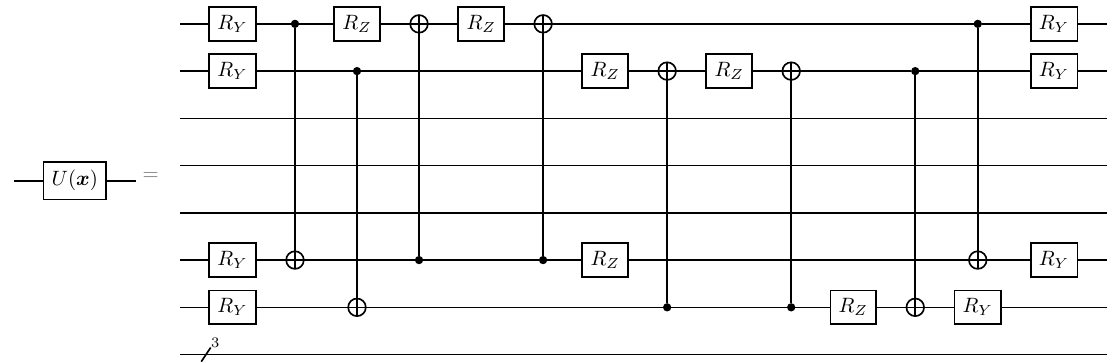}.
\end{equation}
The simplifications are twofold. First, we remove all quantum gates interacting with the last three qubit wires. Second, we decompose the double-excitation gates into $\RZ$ and Clifford gates, leveraging their universal properties. Consequently, by initializing the state as $\rho_0 = \ket{1111100000}\bra{1111100000}$, the corresponding VQE in Eq.~(\ref{eqn:VQE-loss}) achieves good performance in estimating the ground-state energies of the six explored $\H5(R)$ molecules.
 
\begin{figure}[h!]
	\centering
	\includegraphics[width=0.99\textwidth]{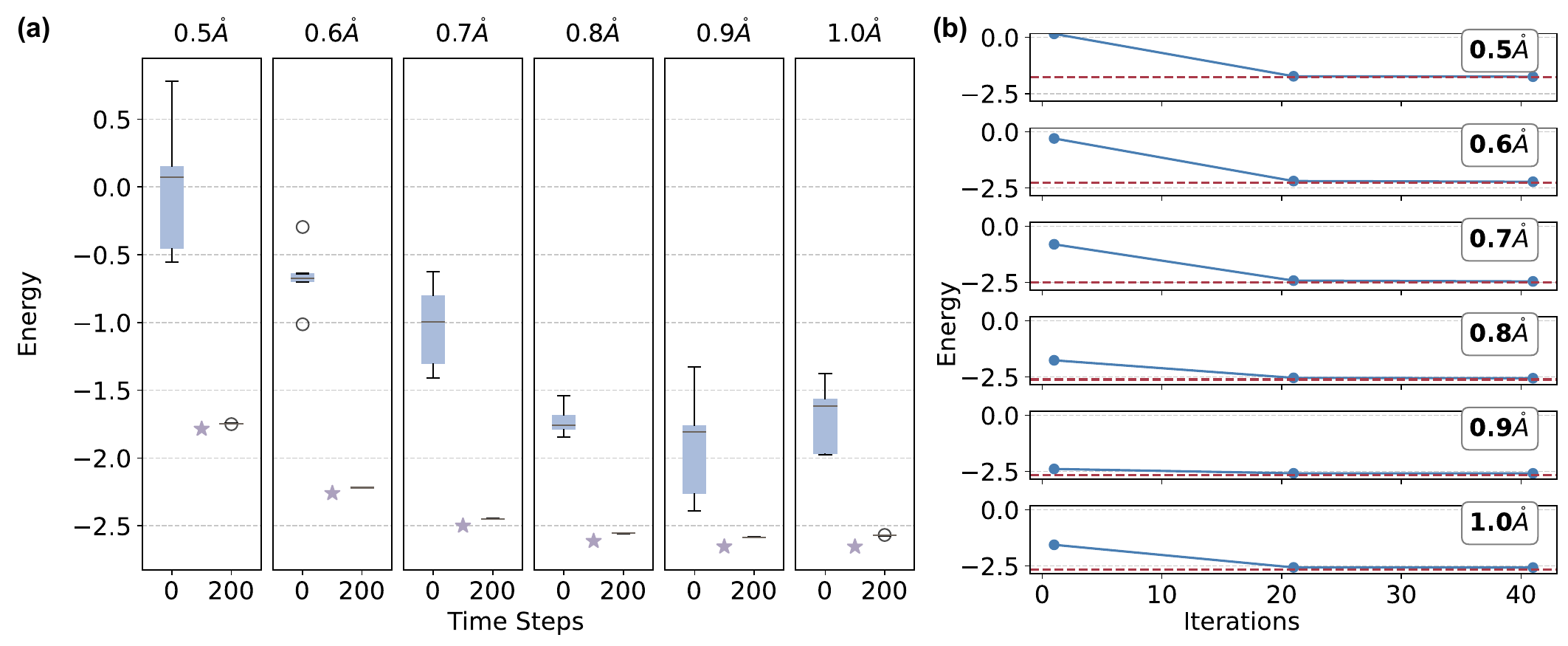}
	\caption{\justifying\small{\textbf{Performance of conventional VQEs on the family of $\mathsf{H}_5$ molecules with $N=10$}. (a) \textsc{Statistical performance of conventional VQEs}. The box plots show the distribution of the estimated energies at the first iteration and after 200 iterations for each setting, with the bond length $R$ ranging from $0.5\mathring{A}$ to $1.0\mathring{A}$. All components of the box follow the same explanations in Fig.~\ref{fig:VQE-TFIM}. (b) \textsc{Training dynamics of the best VQEs}. The plots show the convergence of the estimated energies toward the exact ground-state energies (dashed red lines) over 200 iterations for the top-performing instance in each setting.}}
	\label{fig:VQE-H5}
\end{figure}

\smallskip
\noindent\underline{Benchmark methods and metrics}. The benchmark models and metrics are exactly identical to those in SM~\ref{append:sec:pretrain-TFIM}.

\smallskip
\noindent\underline{Problem setup}. The simulation results in Fig.~\ref{fig:VQE-H5} validate that optimal parameters for the ansatz in Eq.~(\ref{eqn:H5-ansatz}) exist to minimize the six Hamiltonians under investigation. In particular, for each setting,  we perform five repetitions with different initial parameters to collect statistical results, using a learning rate of $0.8$ for the Adam optimizer. Fig.~\ref{fig:VQE-H5}(a) illustrates the statistical performance of VQEs. That is, for each configuration, the exact ground-state energy can be accurately estimated after $200$ iterations, especially for the small bond lengths. For clarity, we visualize the training dynamics of the VQE for the top-performing setting in Fig.~\ref{fig:VQE-H5}(b). Specially, across all settings, the absolute error between the estimated energy and the exact value remains small, which are $0.04$, $0.04$, $0.05$, $0.06$, $0.07$, and $0.08$ for $R=0.5 \mathring{A}, 0.6 \mathring{A}, 0.7 \mathring{A}, 0.8 \mathring{A}, 0.9 \mathring{A}, 1.0 \mathring{A}$, respectively.

Similar to the task of TFIMs, the process of pretraining VQEs by classical learning surrogates follows the steps outlined in Alg.~\ref{alg:offline-VQA}. In Steps 1\&2, we construct the training dataset $\mathcal{T}_{\mathsf{s}}$ by randomly sampling $\bx\in [-\pi, \pi]^{14}$ from a uniform distribution and inputting these values into the ansatz $U(\bx)$ in Eq.~(\ref{eqn:H5-ansatz}) to generate classical shadows. The snapshot for each quantum state $\rho(\bx)$ is fixed to be $T=100$ and the number of trainable examples is set $n=7500$. Once the dataset $\mathcal{T}_{\mathsf{s}}$ is collected, we move to Step 3 of Alg.~\ref{alg:offline-VQA}, where the four classical surrogates employ   $\mathcal{T}_{\mathsf{s}}$ to pretrain VQEs. The hyper-parameter settings for each learning surrogate are identical to those used in the TFIM tasks.  

\begin{figure}[h!]
	\centering
	\includegraphics[width=0.99\textwidth]{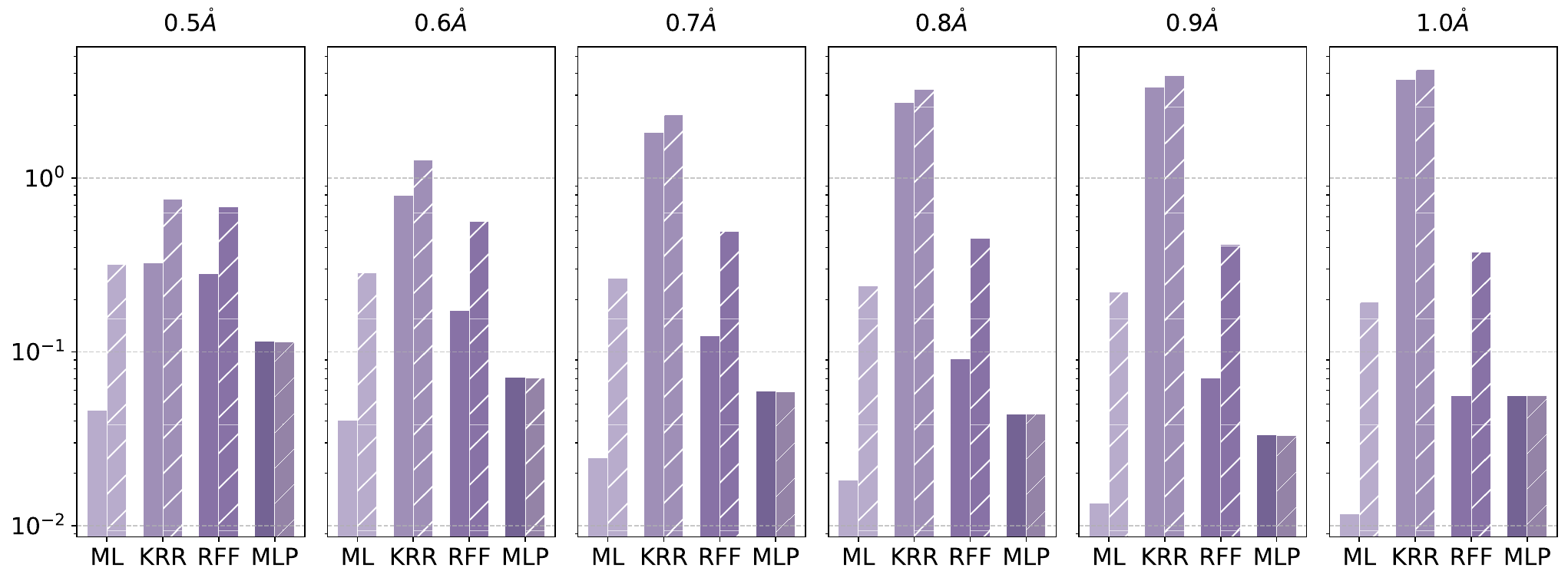}
	\caption{\justifying  \small{\textbf{Prediction error of four learning surrogates for the family of $\mathsf{H}_5$ molecules with $N=10$}. Each subplot corresponds to the simulation results for a specific bond length $R$. The solid bars, shaded bars, and x-axis labels follow the same representations as those in Fig.~\ref{fig:pred-error-TFIM}. }}
	\label{fig:pred-error-H5}
\end{figure}

\smallskip
 
\noindent\underline{Numerical results}. The predictor error $\mathsf{R}_{\text{predict}}$ of the employed four learning surrogates is shown in Fig.~\ref{fig:pred-error-H5}, which is evaluated by $500$ test examples. For all explored bond lengths, the proposed ML model outperforms the other three learning surrogates under the ideal setting. Additionally, under the noisy setting, the proposed ML model performs slightly worse than MLP but remains superior to the other two learning surrogates. More specifically, the prediction error of the proposed ML model under the ideal (noisy) setting achieves $0.046$ $(0.319)$, $0.040$ $(0.286)$, $0.025$ $(0.266)$, $0.018$ $(0.239)$, $0.013$ $(0.222)$, and $0.012$ $(0.193)$ when $R=0.5 \mathring{A}, 0.6 \mathring{A}, 0.7 \mathring{A}, 0.8 \mathring{A}, 0.9 \mathring{A}, 1.0 \mathring{A}$, respectively.  The increased discrepancy between the ideal and shadow results, compared to the TFIM task, arises from the presence of numerous global Pauli terms in the $\mathsf{H}_5$ molecules, which substantially amplify the Pauli-based shadow estimation error.

We proceed to use the optimized four learning surrogates to pretrain VQEs for the six explored $\mathsf{H}_5$ molecules. The estimated energies, both with and without pretraining by the four learning surrogates, are illustrated in Fig.~\ref{fig:pre-train-VQE-H5}. The first key observation is that, under both ideal and noisy conditions, the parameters returned by the proposed ML model substantially improve estimation accuracy after pretraining. In contrast, other learning surrogates fail to reduce the estimation error through pretraining and, in some cases, even lead to degraded performance. Furthermore, while the shadowed dataset increases the test prediction error of the proposed ML model discussed in Fig.~\ref{fig:pred-error-H5}, it notably improves the estimation of the ground-state energy compared to both the ideal case and MLP, particularly when the bond length is $R=0.8\mathring{A}$. These observations align with the results of Theorem~2 in the main text.

\begin{figure}[h!]
	\centering
	\includegraphics[width=0.98\textwidth]{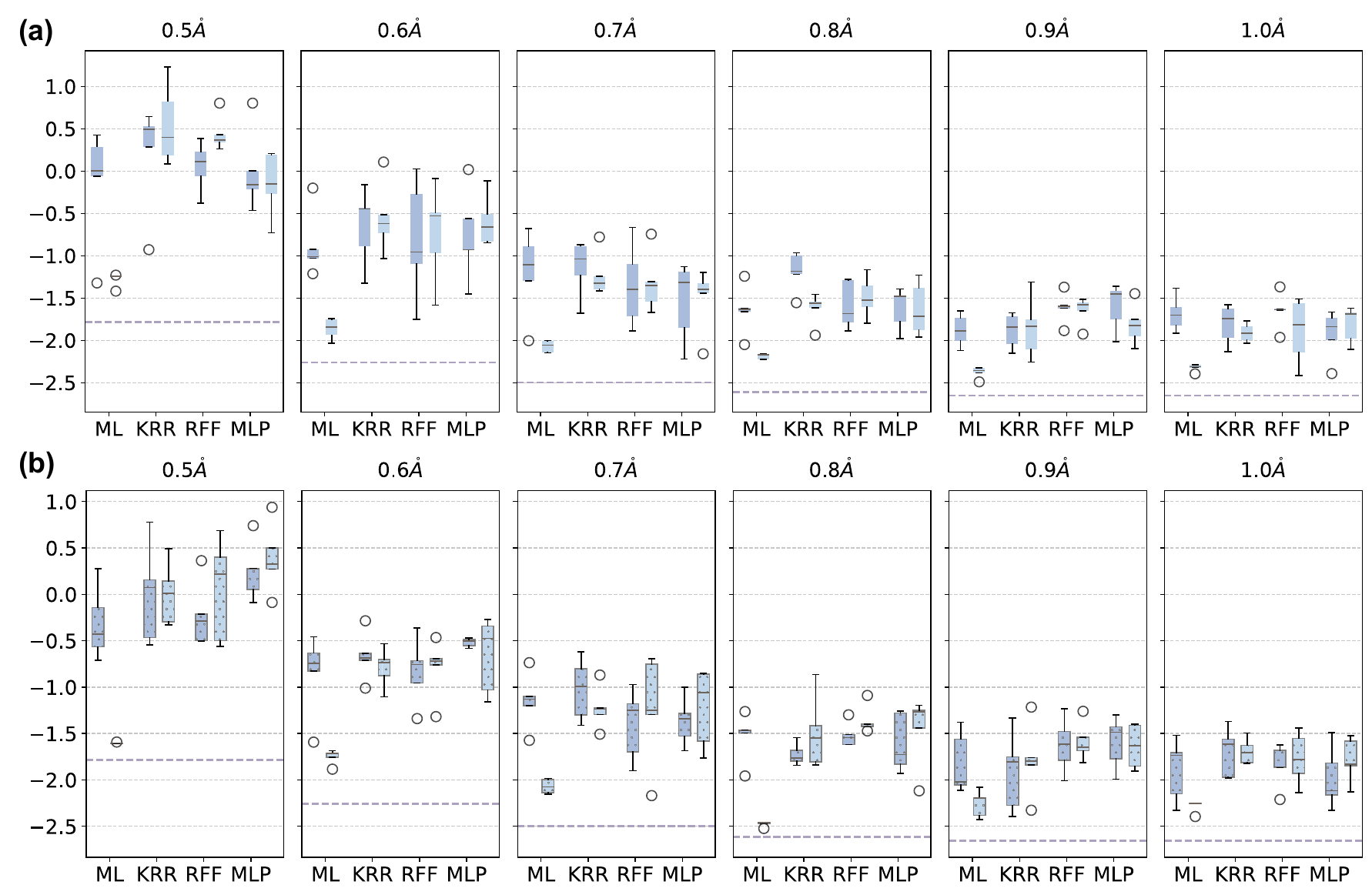}
	\caption{\justifying  \small{\textbf{Pretraining performance of four learning surrogates on the family of $\mathsf{H}_5$ molecules with $N=10$}. The upper panel (a) and the lower panel (b) separately exhibit the simulation results under the ideal and noisy settings ($T\rightarrow \infty$ and $T=100$). Each subplot corresponds to a specific bond length $R$ in $\H5(R)$, and the box plots depict the estimated ground-state energy distributions for four learning surrogates. For each surrogate, the estimated energies without pre-training and after pretraining are visualized by the deep and light blue colors, respectively. Dashed horizontal lines represent the exact ground-state energies.}}
	\label{fig:pre-train-VQE-H5}
\end{figure}

\begin{table}[h!]
\centering
\caption{\justifying  \small{\textbf{Comparison between conventional VQEs and the proposed ML model on a family of $\mathsf{H}_5$ with $N=10$}. The table presents the estimated ground state energies for various coupling strengths $J$, as obtained from the conventional VQE and the proposed ML model under both ideal and noisy conditions (labeled by `Pretrain (ideal)' and `Pretrain (shadow)', respectively). The notation `$a \pm b$' indicates that the mean value is $a$ and the variance is $b$. }}
\label{tab:H5-vqe-ML}
\resizebox{0.88\textwidth}{!}{%
\begin{tabular}{@{}c|c|c|c|c|c|c@{}}
\toprule
\multicolumn{1}{l|}{} & $R=0.5 \mathring{A}$ & $R=0.6 \mathring{A}$ & $R=0.7 \mathring{A}$ & $R=0.8 \mathring{A}$  & $R=0.9 \mathring{A}$  & $R=1.0 \mathring{A}$  \\ \midrule
VQE & $-1.747 \pm 0.000$ & $-2.219 \pm 0.000$ & $-2.45 \pm 0.000$ & $-2.555 \pm 0.000$ & $-2.584  \pm 0.000$ & $-2.571 \pm 0.000$ \\ \midrule
Pretrain (ideal)      & $-1.276 \pm 0.005$ & $-1.862 \pm 0.013$  & $-2.072 \pm 0.004$ & $-2.189  \pm 0.001$ & $-2.376 \pm 0.004$  & $-2.322 \pm 0.002$ \\ \midrule
Pretrain (shadow)     & $-1.601 \pm 0.000$ & $-1.752 \pm 0.005$ & $-2.069 \pm 0.001$ & $-2.475 \pm 0.001$ & $-2.257 \pm 0.017$ & $-2.283 \pm 0.003$ \\ \bottomrule
\end{tabular}%
}
\end{table}
Finally, we compare the performance of the conventional VQE (using infinite measurements) with the proposed ML model. The statistical results are summarized in Table~\ref{tab:TFIM-vqe-ML}. A key observation is that the proposed ML model, under both ideal and noisy settings, generally produces energy estimates that are comparable to but slightly less accurate than those of the conventional VQE. This slight inferiority may be attributed to the increased gradient norm when the Hamiltonian involves more than $400$ Pauli terms, which necessitates more training examples and larger truncation values to achieve improved performance.

\subsection{Comparison with LOWESA}\label{append:sec:comp-lowesa}
We conduct the numerical simulations to compare the wall-clock runtime of our method and LOWESA using open-source implementations. 

The problem setup follows the first scenario explained in SM~\ref{append:sec:diff-Pauli-sim}. The number of qubits is fixed to be $N=32$. The employed ansatz $U(x)$ is an adaptation of the Trotterized transverse-field Ising Hamiltonian circuit, designed with a corresponding bricklayer topology. A visualization of this circuit is provided in Fig.~\ref{fig:lowesa-ansatz}. Specifically, the ansatz is designed with 50,000 layers. All angles of the $\rm{R}_{\text{ZZ}}$ gates and $\RX$ gates are set to zero, except for the 16th $\rm{R}_{\text{ZZ}}$ gate in the first layer, which is assigned arbitrary angles. As such, the explored ansatz is tailored to include a single tunable parameter. The initialized state is $\ket{0}^{\otimes N}$ and the observable employed takes the form as $O=\mathbb{I}^{\otimes 15} \otimes Z \otimes \mathbb{I}^{\otimes 16}$.

\begin{figure}[h!] 
	\centering
	\includegraphics[width = 0.5\textwidth]{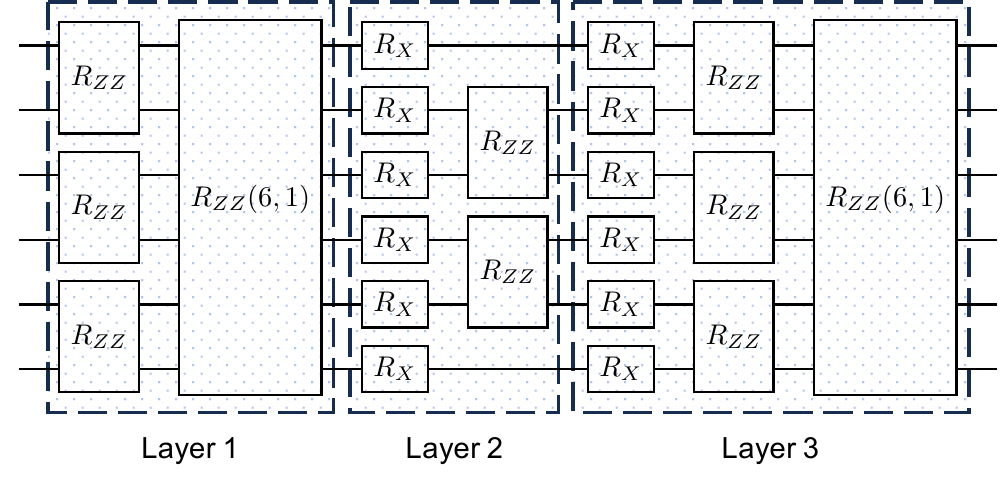}
	\caption{\justifying  \small{\textbf{The employed ansatz with the bricklayer topology}. The figure visualizes the implementation of the employed ansatz $U(\bx)$ with three layers and $N=6$ qubits. Each layer consists of alternating $\rm{R}_{\text{ZZ}}$  and $\RX$ gates. The $\rm{R}_{\text{ZZ}}$ gate labeled $\rm{R}_{\text{ZZ}}(6,1)$ in the odd layers connects the $6$-th and $1$st qubits. For circuits with more than three layers, Layers 2 and 3 are alternated to construct deeper circuits while preserving the bricklayer topology. Specifically, the gate arrangement in Layer $2$ is repeated for all even-numbered layers, while the gate arrangement in Layer $3$ is repeated for all odd-numbered layers.}}
	\label{fig:lowesa-ansatz}
\end{figure}

The hyper-parameter settings for LOWESA are as follows. The implementation is based on the open-source PauliPropagation.jl package available on GitHub. The truncation threshold for the weights is set to $5$, and the minimum coefficient magnitude is set to $4$. We provide $10$ input angles, i.e, $$x\in \{-\pi, -0.8\pi, -0.6\pi, -0.4\pi, -0.2\pi, 0,  0.2\pi, 0.4\pi, 0.6\pi, 0.8\pi \},$$ to the ansatz $U(x)$ and use LOWESA to estimate the corresponding expectation values. 

The hyper-parameter settings of the proposed learning model $h_{\mathsf{s}}$ are as follows. The number of training examples is set as $1000$. For simplicity, we consider the ideal scenario such that the infinite measurements are adopted to construct the dataset. Given access to the training dataset, we train the learning model by setting the truncation threshold $\Lambda$ as $2$. Once training is completed, we use the same $10$ input angles with LOWESA as test examples to evaluate the prediction performance and inference time.

\begin{table}[h!]
\centering
\caption{\justifying \textbf{Performance and wall-clock time of LOWESA and the proposed learning model}. The notation `$a \pm b$' indicates that the mean value is $a$ and the variance is $b$. The wall-clock time for the proposed learning approach focuses on the inference stage. }
\label{tab:Lowesa-learning-time}
\resizebox{0.55\textwidth}{!}{%
\begin{tabular}{@{}c|c|c@{}}
\toprule
\multicolumn{1}{l|}{}         & LOWESA            & Our work                                         \\ \midrule
Estimation (Prediction) error & $0 \pm 0.00$      & $0 \pm 0.0002$                                   \\ \midrule
Wall-clock time (seconds)              & $2.003\pm 0.0107$ & $8.6808\times 10^{-5} \pm  2.202 \times 10^{-8}$ \\ \bottomrule
\end{tabular}%
}
\end{table}

The performance and computational efficiency between LOWESA and the proposed learning model are summarized in Table~\ref{tab:Lowesa-learning-time}. Regarding the prediction (estimation) error, both approaches achieve comparable accuracy, with the proposed method showing a slight variance of $0.0002$. However, in the wall-clock time, our proposed learning model demonstrates a dramatic reduction in inference time compared to LOWESA, achieving a reduction by approximately $2.3\times 10^4$ times. This distinction underscores a scenario where the learning model surpasses classical simulators, emphasizing their complementary roles.

 \end{document}